\def\llncs{0}
\def\comments{0}
\def\anon{0}

\ifnum\llncs=1
\documentclass{llncs}
\else
\documentclass[11pt]{article}
\usepackage{fullpage}
\fi
\usepackage{graphicx} 
\usepackage{amsmath}

\usepackage{amsthm}
\usepackage{amssymb}
\usepackage{physics}
\usepackage{hyperref}
\usepackage{cleveref}
\usepackage{dsfont}
\usepackage{xcolor}
\usepackage{mathtools}
\usepackage{breakcites}
\usepackage{subcaption}

\ifnum\llncs=0
\newtheorem{theorem}{Theorem}[section]

\newtheorem{corollary}[theorem]{Corollary}
\newtheorem{definition}[theorem]{Definition}
\newtheorem{lemma}[theorem]{Lemma} 
\newtheorem{claim}{Claim}      
\newtheorem{proposition}[theorem]{Proposition}

\theoremstyle{definition}

\theoremstyle{remark}
\newtheorem{remark}[theorem]{Remark}

\fi

\newcommand{\A}{\mathcal{A}}

\renewcommand{\O}{\mathcal{O}}
\newcommand{\N}{\mathbb{N}}

\newcommand{\C}{\mathbb{C}}
\newcommand{\E}{\mathop{\mathbb{E}}}

\newcommand{\U}{\mathcal{U}}
\newcommand{\V}{\mathcal{V}}

\newcommand{\sA}{\mathsf{A}}
\newcommand{\sB}{\mathsf{B}}
\newcommand{\sC}{\mathsf{C}}

\newcommand{\poly}{\mathsf{poly}}
\newcommand{\negl}{\mathsf{negl}}

\newcommand{\wt}[1]{\widetilde{#1}}
\newcommand{\ol}[1]{\overline{#1}}

\DeclareMathOperator{\Ima}{Im}
\DeclareMathOperator{\Dom}{Dom}
\DeclareMathOperator{\Sim}{Sim}
\DeclareMathOperator{\Enc}{Enc}
\DeclareMathOperator{\Dec}{Dec}

\DeclareMathOperator{\Span}{Im}

\newcommand{\BigPr}[2]{
\Pr\left[
\begin{array}{c}
#1
\end{array}
:
\begin{array}{c}
#2
\end{array}
\right]
}

\newcommand{\xor}{\oplus}

\newcommand{\Haar}{Haar}
\newcommand{\HaarSt}{HaarSt}

\ifnum\comments=1
\newcommand{\todo}[1]{{\color{blue} (TODO: #1) }}
\newcommand{\eli}[1]{{\color{purple} (Eli: #1) }}
\newcommand{\james}[1]{{\color{red} (James: #1) }}
\else
\newcommand{\todo}[1]{}
\newcommand{\eli}[1]{}
\newcommand{\james}[1]{}
\fi

\newcommand{\Valid}{Val}
\newcommand{\Reg}[1]{\mathbf{#1}}

\newcommand{\Ciph}{\mathcal{CT}}

\newcommand{\Ora}[1]{\vec{#1}}

\title{Unclonable Encryption in the Haar Random Oracle Model}
\author{James Bartusek\thanks{Columbia University} \and Eli Goldin\thanks{New York University. Supported by an NSF graduate research fellowship.}}
\date{}
\ifnum\anon=1
\institute{}
\author{}
\date{}
\fi
\ifnum\llncs=1

\author{James Bartusek\inst{1} \and Eli Goldin\inst{2}}
\institute{Columbia University\\ \email{bartusek.james@gmail.com} \and New York University\\ \email{eli.goldin@nyu.edu}}
\fi

\begin{document}

\maketitle

\begin{abstract}
     We construct unclonable encryption (UE) in the Haar random oracle model, where all parties have query access to $U,U^\dagger,U^*,U^T$ for a Haar random unitary $U$. Our scheme satisfies the standard notion of unclonable indistinguishability security, supports reuse of the secret key, and can encrypt arbitrary-length messages. That is, we give the first evidence that (reusable) UE, which requires computational assumptions, exists in ``microcrypt'', a world where one-way functions may not exist. 
     
     As one of our central technical contributions, we build on the recently introduced path recording framework to prove a natural ``unitary reprogramming lemma'', which may be of independent interest.
\end{abstract}

\section{Introduction}

Unclonable encryption (UE), introduced by Broadbent and Lord \cite{BL20}, is a central primitive of study in quantum cryptography, both due to its potential real-world impact, say, to protect against ``store-now, decrypt-later'' attacks, and due to a plethora of open questions that remain about its feasibility. 

An unclonable encryption scheme encodes a message $m$ into a quantum ciphertext $\rho_{k,m}$ under secret key $k$ with the following security guarantee. For any pair of messages $m_0,m_1$ and $b \gets \{0,1\}$, no adversary can map $\rho_{k,m_b}$ to a state on two registers $\Reg{A},\Reg{B}$ such that $b$ can be recovered given either register along with the key $k$. To be precise, the standard security definition, termed ``unclonable indistinguishability'' by  \cite{BL20}, requires that for any cloning adversary, \[\Pr[b_A = b_B = b] \leq \frac{1}{2} + \negl(\lambda),\] where $b_A$ (resp. $b_B$) is the guess made given key $k$ and register $\Reg{A}$ (resp. $\Reg{B}$). We note that while several works have also considered the weaker notion of ``search security'', which bounds the probability that the adversary can guess an entropic message, our focus in this work is on the above notion of unclonable indistinguishability.

As in plain (private-key) encryption, one can consider two main variants: ``one-time'' UE, where each key is only used to encrypt a single message, and ``reusable'' (or many-time) UE, where a single private key can be used to encrypt any number of messages. While one-time UE is plausibly an information-theoretic object, reusable UE requires computational assumptions (as it implies reusable plain private-key encryption).

Surprisingly, several questions remain open about the feasibility of both one-time and reusable UE, even given the recent surge of activity in the field of unclonable cryptography. In the one-time setting, \cite{BL20} presented a candidate construction in the random oracle model, and variants of this basic scheme were later shown secure \cite{AKLLZ22,AKL23}. Unfortunately, these proofs rely heavily on the random oracle, despite the fact that cryptographic assumptions are not even believed to be necessary. Indeed, information-theoretic candidates have also been proposed (e.g. \cite{Botteron:2024orj,coladangelo_et_al:LIPIcs.ITCS.2026.41,bhattacharyya2026uncloneable}), with various partial results given about their security. For instance, \cite{bhattacharyya2026uncloneable} proves that one-time UE exists information-theoretically with advantage bound $\widetilde{O}(1/\lambda)$. 


While there has been less explicit study of the reusable variant, it is straightforward to modify the constructions of \cite{AKLLZ22,AKL23} to yield reusable UE in the random oracle model. As mentioned earlier, while reusable UE \emph{does} require computational assumptions (partially explaining the use of the random oracle), it is natural to wonder just how much cryptographic structure is required to build reusable UE. The goal of this work is to make progress on this question.

\paragraph{Minicrypt vs microcrypt.} As the constructions of \cite{AKLLZ22,AKL23} utilize a random oracle, they at the very least require the existence of one-way functions. This places (reusable) UE in the realm of ``minicrypt'', i.e. primitives that have (potentially heuristic) constructions from unstructured but classical forms of cryptographic hardness. The minimal form of such classical unstructured hardness is a one-way function.

However, the following question has been left unresolved: Does (reusable) UE \emph{require} one-way functions?\footnote{A recent work \cite{poremba_et_al:LIPIcs.ITCS.2026.109} has considered a similar question in the context of \emph{search-secure} unclonable encryption. We provide a detailed comparison with their work in \Cref{subsec:related}.} It may be instructive to contrast this situation with what is known about \emph{plain} (private-key) encryption and (private-key) quantum money. In fact, plain encryption can be built from pseudo-random unitaries \cite{10.1007/978-3-031-15802-5_8}, and quantum money can be built from pseudo-random states \cite{10.1007/978-3-319-96878-0_5}, both of which may exist even if one-way functions do not \cite{kre21,10.1145/3564246.3585225,10.1145/3717823.3718144}. 

That is, both live in the realm of ``microcrypt'', i.e. primitives that have (potentially heuristic) constructions from unstructured \emph{quantum} forms of cryptographic hardness. Just as one may consider any primitive with a construction in the random oracle model to live in minicrypt, one may also consider any primitive with a construction in the \emph{Haar} random oracle model to live in microcrypt. In this model, a random unitary $U$ is sampled from the Haar distribution at the beginning of time, and all parties get query access to $U,U^\dagger,U^*,$ and $U^T$. Indeed, any construction in the Haar random oracle model can be heuristically instantiated by replacing the Haar random oracle with a concrete (say, randomly sampled) quantum circuit, analogously to how random oracle constructions have heuristic instantiations from concrete cryptographic hash functions. This definition of microcrypt is robust, a series of recent works have shown that pseudo-random unitaries (and all derivatives, such as pseudo-random states, plain encryption, and quantum money) can be constructed in the Haar random oracle model~\cite{TCC:HhaYam25,EC:ABGL25,C:ABGL25}.

Now, unclonable encryption can be seen as a natural simultaneous strengthening of both plain encryption and quantum money. Indeed, it requires both the hiding properties of plain encryption and unclonable properties of quantum money. Thus, it is natural to ask whether unclonable encryption can also be placed in microcrypt.

\subsection{Results} 

We show that unclonable encryption is indeed a microcrypt primitive, establishing the following main theorem.

\begin{theorem}[Informal]
    For any polynomial message size, there exists a reusable unclonable encryption scheme with unclonable indistinguishability in the Haar random oracle model.
\end{theorem}

Following the techniques of~\cite{kre21}, it is then straightforward to show

\begin{corollary}
    There exists an oracle $\mathcal{O}$ relative to which
    \begin{enumerate}
        \item Reusable unclonable encryption with unclonable indistinguishability exists relative to $\mathcal{O}$.
        \item $BQP^{\mathcal{O}}=QMA^{\mathcal{O}}$, and so one-way functions do not exist relative to $\mathcal{O}$.
    \end{enumerate}
\end{corollary}

In fact, we exhibit a general compiler, from which our main theorem can be derived by plugging in the schemes of \cite{AKLLZ22,AKL23} proven secure in the quantum random oracle model.

\begin{theorem}[Informal]\label{thm:in-reduct}
    If reusable unclonable encryption exists with respect to \emph{any} distribution over unitary oracles,\footnote{Technically, we need the scheme to satisfy a mild structural property which we call ``pure'' (see \Cref{def:pure-scheme}).} then it exists in the Haar random oracle model.
\end{theorem}

Note that the random oracle (consisting of a random classical function) can be seen as one particular distribution over unitary oracles, and so the above theorem applies. 

Along the way, we prove the following technical lemma, which may be of independent interest. Let $\mathcal{H}([N])$ be a Hilbert space spanned by $\{\ket{i}\}_{i \in [N]}$, and suppose that $S_2 \subset [N]$ is sampled at random conditioned on containing some set of fixed points $S_1$, where $\frac{|S_2|}{N}, \frac{|S_1|}{|S_2|} = \negl(\lambda)$.

\begin{lemma}[Informal: Unitary reprogramming lemma]
    No adversary has better than $\negl(\lambda)$ advantage in distinguishing whether they are querying a ``monolithic'' Haar random oracle $U$ on the entire space $\mathcal{H}([N])$ or a product $U_1U_2$ of two ``disjoint'' Haar random oracles, where $U_1$ is on $\mathcal{H}(S_2)$ and $U_2$ is on $\mathcal{H}([N] \setminus S_2)$.
\end{lemma}

Although $S_1$ and $S_2$ are described here as subsets, the proof of this lemma is not sensitive to the basis chosen. Thus, the same lemma holds where $S_1$ is some fixed subspace of $\mathcal{H}([N])$ and $S_2$ is a random subspace containing $S_1$ as long as $\frac{\dim(S_2)}{N},\frac{\dim(S_1)}{\dim(S_2)} =\negl(\lambda)$.


\subsection{Techniques}

\paragraph{Construction.} We begin with our construction of unclonable encryption in the Haar random oracle model. Let $U$ be the Haar random unitary, $k$ be a classical key, and suppose that we wish to encrypt a message $m \in \{0,1\}^\ell$ (of some fixed length). To do so, the encryptor samples some randomness $r \gets \{0,1\}^n$ and outputs \[X^k U\ket{m,r},\] where $X^k$ is the operation that applies a bit flip to qubit $i$ if $k_i = 1$. We can think of $U$ (which is publicly-queryable) as defining a partitioning of Hilbert space into one subspace $\{U\ket{m,r}\}_r$ for each choice of $m$. Then, $X^k$ defines a random ``shift'' of these subspaces so that they are not accessible to the cloning adversary (who does not know $k$). An encryption of $m$ is thus a shifted sample from the $m$'th subspace.

\paragraph{Security game.} Our goal is then to show that if there exists an unclonable encryption scheme in \emph{any} oracle model (with some minor restrictions on the format of the construction), then the above scheme is secure. We do so via a strategy akin to a random self-reduction. In order to explain the reduction, we first set up some more notation. 

The adversary in the unclonable encryption game is specified by a cloner $\sC$ and two distinguishers $\sA$ and $\sB$. The security game proceeds as follows, where all adversaries have access to a Haar random $U$ via $U, U^\dagger, U^*, U^T$ queries.

\begin{itemize}
    \item The encryption key $k$ is sampled and $\sC$ obtains access to an encryption oracle $\Enc(k,\cdot)$ which they can query any number of times to get sample encryptions under key $k$.
    \item $\sC$ outputs a pair of challenge messages $m_0,m_1$, and receives $\ket{\psi} = X^k U\ket{m_b,r}$ for a random choice of $b \gets \{0,1\}$ and $r \gets \{0,1\}^n$. 
    \item $\sC$ continues to have access to $\Enc(k,\cdot)$, and eventually outputs a state on two registers $\Reg{A},\Reg{B}$.
    \item Run $b_A \gets \sA(k,\Reg{A})$ and $b_B \gets \sB(k,\Reg{B})$.
    \item The adversary wins if $b = b_A = b_B$.
\end{itemize}

We refer to the operation of $\sC$ in the above game as the \emph{cloning phase}, and operations of $\sA,\sB$ as the \emph{distinguishing phase}.

Now, suppose we are given an unclonable encryption scheme that has been proven secure in 7some other oracle model (meaning that no adversary has better than $1/2 + \negl(\lambda)$ advantage). We demand that the encryption procedure for this scheme can be written as \[E_{\wt{k}}^{\wt{U}}\ket{m,\wt{r}},\] where $E_{\wt{k}}$ is some oracle-aided unitary ($\wt{U}$ being the oracle) and $\wt{r}$ is some classical random string.\footnote{In fact, we also allow the encryption unitary to operate on some auxiliary 0 qubits, but still demand that the output is pure (given $m$ and $\wt{r}$). That is, the scheme does not trace out any part of the encryption unitary's output.} We will exhibit a reduction showing that any adversary against our Haar random oracle scheme can be used to break this scheme.

\paragraph{The reduction.} The main idea is to ``absorb'' the description of the unitary $E_{\wt{k}}^{\wt{U}}$ into $U$ so that the reduction can essentially re-encrypt its challenge ciphertexts encrypted under $E_{\wt{k}}^{\wt{U}}$ by additionally applying $X^k U$. In particular, we will set parameters so that $n$ (the size of the random string $r$ in our scheme) is twice the size of $\wt{r}$, and consider the encryption scheme that samples $\wt{r}, s \in \{0,1\}^{n/2}$ and outputs \[X^k U E_{\wt{k}}^{\wt{U}}\ket{m,\wt{r}}\ket{s}.\] The reduction can then absorb the description of $E_{\wt{k}}^{\wt{U}}$ into $U$, giving the adversary access to $U E_{\wt{k}}^{\wt{U}}$ rather than $U$, which is identically distributed due to the fact that $U$ is Haar random.

This same approach was successfully applied in~\cite{poremba_et_al:LIPIcs.ITCS.2026.109} to analyze unclonable encryption in their variant of the Haar random oracle model, where the cloner does not receive access to the Haar random unitary. In this setting, \cite{poremba_et_al:LIPIcs.ITCS.2026.109} showed that even without $X^k$ and $r$, encryptions $U\ket{m}$ have inverse exponential unclonable search security. However, this idea does \textit{not} work if the cloner is allowed query access to $U$ (as in our setting), since the reduction does not receive the description of $\wt{k}$ until \emph{after} the cloning phase. 

To address this inconsistency, we sample a random subset $R \subset \{0,1\}^{n/2}$ of size, say, $2^{n/4}$, and consider splitting $\mathcal{H}(\{0,1\}^{\ell + n})$ (the Hilbert space acted upon by $U$) into the following subspaces:\footnote{To simplify the proof, we actually use a slightly different definition of $S_2$, though the definition here still conveys the basic intuition.} \[S_2 \coloneqq \mathsf{Span}\{\ket{m,\wt{r},s}\}_{s \in R}, \quad S_2^\perp \coloneqq \mathsf{Span}\{\ket{m,\wt{r},s}\}_{s \notin R}.\] 

Now, we define the unitary queried by the adversary to be the product of two unitaries $U_1 U_2$, where $U_1$ applies $U E_{\wt{k}}^{\wt{U}}$ to $S_2$ (and acts as the identity on $S_2^\perp$), and $U_2$ applies $U$ to $S_2^\perp$ (and acts as the identity on $S_2$). Then, whenever the adversary makes an encryption query, we sample $s \gets R$ to use as the random coins. This takes care of our issue due to the following observations.

\begin{itemize}
    \item Since all of the adversary's ciphertexts are masked by $X^k$, they shouldn't learn anything about $R$, and therefore will not be able to query inside the subspace $S_2$ during the cloning phase. While intuitive, establishing this formally takes some care, which we do via a sequence of hybrids in \if\llncs=0 \Cref{sec:hardhyb}\else Appendix A of the full version\fi.
    \item Hence, the reduction can simulate the adversary's unitary using \emph{only} $U_2$ during the cloning phase, and then $U_1 U_2$ during the distinguishing phase, once it knows $E_{\wt{k}}^{\wt{U}}$.
\end{itemize}

\paragraph{The unitary reprogramming lemma.} While defining the adversary's unitary as $U = U_1U_2$ allows us to get the reduction to go through, we have now actually changed the distribution over oracles that the adversary sees. A key hybrid in our proof establishes that the adversary will not be able to distinguish this change in distribution.

In more detail, we abstract the setting as follows. Let $\mathcal{H}([N])$ be a Hilbert space spanned by $\{\ket{i}\}_{i \in [N]}$, and let $S_1 \subset [N]$ be some set of public fixed points. $S_1$ will correspond to the set of random strings $s$ used to respond to the adversary's encryption queries, which must be known to the reduction. Suppose that we sample a subset $S_2$ of some fixed size uniformly at random such that $S_1 \subset S_2 \subset [N]$. This larger set corresponds to the $S_2$ in the reduction described above. Then, as long as $\frac{|S_2|}{N} = \negl(\lambda)$ and $\frac{|S_1|}{|S_2|} = \negl(\lambda)$, we can show that no adversary has better than $\negl(\lambda)$ advantage in distinguishing the following cases.
\begin{itemize}
    \item Sample $U \gets \mathsf{Haar}([N])$ and give query access to $U,U^\dagger,U^*,U^T$.
    \item Sample $U_1 \gets \mathsf{Haar}(S_2)$ and $U_2 \gets \mathsf{Haar}([N] \setminus S_2)$ and give query access to $(U_1U_2),\allowbreak(U_1U_2)^\dagger,\allowbreak(U_1U_2)^*,\allowbreak(U_1U_2)^T$.
\end{itemize}

\paragraph{Path recording and proof intuition for unitary reprogramming}
To show this, we build on the recently introduced path recording framework \cite{MH25,SMLBH25}. For the purpose of exposition, we will discuss the case where adversaries only have forward access to $U$ or $U_1U_2$. 

In the path recording framework, oracle queries to a Haar random unitary $U\gets \Haar([N])$ can be efficiently simulated by making queries to the partial isometry $V$ operating on some internal state $\ket{D}$ where $D \subseteq [N]\times [N]$ is an injective relation (i.e. if $D=\{(x_1,y_1),\dots,(x_t,y_t)\}$, the $y_i$'s are all unique). $V$ acts as follows:
$$V\ket{x}\ket{D} \propto \sum_{y \in [N]\setminus \Ima(D)} \ket{y}\ket{D\cup \{(x,y)\}}.$$

Generalizing this, we can simulate queries to $U_1U_2$ by the following map $V^{S_2}$ defined by
\begin{enumerate}
    \item On input $\ket{x}\ket{D_1}\ket{D_2}$ for $x\in S_2$,
    $$V^{S_2}\ket{x}\ket{D_1}\ket{D_2}\propto\sum_{y\in S_2\setminus \Ima(D_1)}\ket{y}\ket{D_1\cup \{(x,y)\}}\ket{D_2}.$$
    \item On input $\ket{x}\ket{D_1}\ket{D_2}$ for $x\notin S_2$,
    $$V^{S_2}\ket{x}\ket{D_1}\ket{D_2}\propto\sum_{y\in [N]\setminus (S_2\cup \Ima(D_2))}\ket{y}\ket{D_1}\ket{D_2\cup \{(x,y)\}}.$$
\end{enumerate}

We then look at the purification of this process, where we begin with the state 
$$\sum_{\substack{S_1\subseteq S_2\\|S_2|=M_2}}\ket{S_2}$$
stored in an ancilla register, and then apply $V^{S_2}$ controlled on this register.

It turns out that as long as only polynomially many queries are made, the state stored in this ancilla register will always be approximately determined by the register containing $D_1$. In particular, the internal state of the oracle will always be close to a mixture of superpositions of states of the form
$$\ket{D_1}\ket{D_2}\sum_{\substack{S_1\cup \Ima(D_1)\subseteq S_2\\|S_2|=M_2}}\ket{S_2}.$$ 
Here, $S_1\cup \Ima(D_1)$ represents the strings the adversary is able to find in $S_2$. In particular, these are the strings which are either contained in $S_1$, or returned to the adversary from queries on strings it already knew were in $S_2$ (and thus are stored in $D_1$).

Simple algebra shows that applying the purified $V^{S_2}$ oracle to a state of this form will result in a superposition of states of this form. When $x\notin \Ima(D_1)\cup S_1$, this process results in some negligible error corresponding to the event that $x$ happens to lie in $S_2$. The fact that the same holds for inverse, transpose, and conjugate queries follows from a direct combinatorial bound.

Thus, since the $S_2$ register is always determined information-theoretically by the $D_1$ register, we can delete it with no consequences. Doing so properly will show that query access to $U_1U_2$ as determined by the process above is indistinguishable from query access to $\widetilde{V}$ defined by
\begin{enumerate}
    \item On input $\ket{x}\ket{D_1}\ket{D_2}$ for $x\in \Ima(D_1)\cup S_1$,
    $$\widetilde{V}\ket{x}\ket{D_1}\ket{D_2}\propto\sum_{y\in [N]\setminus (\Ima(D_1) \cup S_1)}\ket{y}\ket{D_1\cup \{(x,y)\}}\ket{D_2}.$$
    \item On input $\ket{x}\ket{D_1}\ket{D_2}$ for $x\notin \Ima(D_1)\cup S_1$,
    $$\widetilde{V}\ket{x}\ket{D_1}\ket{D_2}\propto\sum_{y\in N\setminus (\Ima(D_1)\cup S_1\cup \Ima(D_2))}\ket{x}\ket{D_1}\ket{D_2\cup \{(x,y)\}}.$$
\end{enumerate}

At this point, standard techniques can be used to show that this is equivalent to the original path recording oracle $V$ corresponding to a single random unitary $U$. In particular, there is an isometry $T$ which takes an internal state $\ket{D}$ of $V$ and turns it into an internal state consistent with $\widetilde{V}$, and since this isometry acts only on the internal state, queries to $V$ and $\widetilde{V}$ are indistinguishable. In slightly more detail, $T$ will look at the relations in $D$ and partition them between $D_1$ and $D_2$ as follows: for $(x,y)\in D$, if there is some path $(x_1,y_1),(x_2,y_2)\dots,(x_t,y_t) \in D$ such that 
\begin{enumerate}
    \item $(x_y,x_{t+1})=(x,y)$
    \item $x_1\in S_1$
    \item $x_i=y_{i-1}$
\end{enumerate}
then $T$ will add $(x,y)$ to $D_1$. Otherwise, it will add $(x,y)$ to $D_2$.

\subsection{Concurrent and independent work}

A recent preprint \cite{bhattacharyya2026uncloneablebitexists} studies the question of information-theoretic one-time secure UE for single-bit messages. They propose a construction based on Haar random unitaries and claim negligible security, improving on the inverse-polynomial security of \cite{bhattacharyya2026uncloneable}. The secret key in their scheme includes the description of a Haar random unitary that is applied to the message bit concatenated with a random string. Hence, their proposed construction is not efficient, and they leave open the question of information-theoretic one-time UE with efficient encryption and decryption.

\subsection{Other related work}\label{subsec:related}

\paragraph{Unclonable encryption and microcrypt.}

To the best of our knowledge, \cite{poremba_et_al:LIPIcs.ITCS.2026.109} is the only other work that has explicitly considered the question of whether unclonable encryption lies in microcrypt. However, their result is weaker than ours in the following respects:

\begin{itemize}
    \item They consider \emph{search-secure} unclonable encryption, which, as mentioned earlier, is weaker than the standard notion of unclonable indistinguishability.
    \item They prove that search-secure unclonable encryption exists in a non-standard variant of the Haar random oracle model where there exists an \emph{exponentially large family} $\{U_i\}_{i \in [2^\lambda]}$ of Haar random oracles for each security parameter. Moreover, the adversary is not given access to these oracles during the cloning phase, and is only given access to $U_i,U_i^\dagger$ for a single choice of $i$ in the distinguishing phase.
    \item They do not consider \emph{reusable} security, and only allow the adversary to see a single ciphertext per key. Like us, they do allow the message size to be bigger than the key size, which means their notion is indeed computational.
\end{itemize}

We also mention that they write down a search to decision reduction for a general class of ``oracular cloning games'' establishing that Haar random oracles are the hardest type of oracles for such games. Conceptually, this is similar to the compiler we analyze in this work for building UE in the Haar random oracle model. However, a crucial distinction is that in their setting, the adversary does not get access to the oracle in the cloning phase, which introduces significant complications in our setting.

\paragraph{Cryptography in the Haar random oracle model.} The works of~\cite{EC:ABGL25,C:ABGL25} construct pseudorandom unitaries in the Haar random oracle model by interleaving queries to $U$ by applications of $X^k$. Our construction of unclonable encryption is inspired by this, although it turns out that in our setting a simplified construction using only one query to $U$ is sufficient. Moreover, our proof structure is radically different. 

\paragraph{Unitary reprogramming.} \cite{TCC:HhaYam25} also establishes a lemma that they call a ``unitary reprogramming lemma" (a similar lemma, under the name ``one-way to hiding for unitary oracles", appears in~\cite{CGR26}). Intuitively, their lemma states that for any oracle, if there is a subspace for which it is difficult to find a member state, then the oracle is indistinguishable from the same oracle pre-composed with an arbitrary unitary acting on that subspace. This is essentially a generalization of the one-way to hiding lemma from~\cite{BBBV}. Our reprogramming lemma is unrelated to this: we show that a Haar random unitary is indistinguishable from the direct sum of two Haar random unitaries on random subspaces, even if a large portion of the subspaces are revealed. 

One of the key challenges in our work (which we use our new unitary reprogramming lemma to solve) is to show that the cloner cannot query its oracle on the subspace corresponding to encryptions under the correct key ($S_2$). However, since the cloner does get some information about $S_2$ in the form of valid ciphertexts, it may be possible for them to find member states of the relevant subspace, and thus we cannot immediately apply a unitary version of one-way to hiding. Instead, we use our unitary reprogramming oracle to isolate the information the adversary can obtain about the encryption subspace, after which we later apply one-way to hiding style techniques \ifnum\llncs=0(see~\Cref{sec:hardhyb})\else(see Appendix A of the full version)\fi.

More generally, our new reprogramming lemma is very similar to techniques used in reprogramming for random \textit{permutations} in the classical setting. This is quite different from a generalization of one-way to hiding as seen in~\cite{TCC:HhaYam25}, and the two techniques can be applied in conjunction (as we do here).

\subsection{Paper structure}

We begin with preliminaries in~\Cref{sec:prelims}. We then prove the ``unitary reprogramming lemma" described earlier in~\Cref{sec:reprog}. In~\Cref{sec:mainthm}, we give a formal proof of~\Cref{thm:in-reduct}. Finally, in~\Cref{sec:qrom}, we conclude the argument by showing that existing constructions of unclonable encryption in the QROM can be modified to achieve many-time security with a pure construction, thus establishing our main result.

\section{Preliminaries}
\label{sec:prelims}
\subsection{Notation}

We will use $n$ to refer to the security parameter. For a variable $X$, $\ell_{X}$ will generally denote some polynomial in $n$ where $\{0,1\}^{\ell_{X}(n)}$ will correspond to the domain of $X$. When the security parameter is clear from context, we will omit it.

$X$ will denote Pauli $X$. For $k\in \{0,1\}^{\ell_k(n)}$, when the ambient space is clear from context, $X^k$ will refer to the map $X^{k_1}\otimes \dots \otimes X^{k_{\ell_k}}\otimes I$.

For $S$ a set, we will define $\mathcal{H}(S)$ to be the Hilbert space with basis vectors $\ket{s}$ for $s\in S$. A register is a named Hilbert space, and will be represented by a bolded string, for example $\Reg{A}$. For a set $S$, $\Haar(S),\HaarSt(S)$ will refer to the Haar distribution over unitaries and states respectively on $\mathcal{H}(S)$. When there is an implicit larger ambient space $\mathcal{H}(T)$ for $S\subseteq T$, $\Haar(S)$ will refer to $\Haar(S) \oplus I_{\mathcal{H}(T\setminus S)}$.

For a set $T$ and a natural number $M$, we define ${T\choose M} = \{S\subseteq T : |S| = M\}$. We define $\mathcal{P}(T) = \{S \subseteq T\} = \bigcup_{i=0}^{|T|}{T\choose i}$ to be the powerset of $T$.
\subsection{Oracle algorithms}


\begin{definition}
    An oracle $\mathcal{O}$ is an arbitrary linear map with input registers $\Reg{In}$ and $\Reg{St}$ and output registers $\Reg{Out}$ and $\Reg{St}$ along with an initial state $\ket{\phi_0}$ in $\Reg{St}$. An oracle algorithm $\A^{(\cdot)} = (\A_1,\dots,\A_t)$ is a sequence of CPTP maps with input $\Reg{A},\Reg{Out}$ and output $\Reg{A},\Reg{In}$. We define the (mixed) state $\A^{\mathcal{O}}(\ket{\psi_0}_{\Reg{A},\Reg{Out},\Reg{St}})$ to be the state
    $$\A^{\mathcal{O}}(\ket{\psi_0}_{A,St}) = \left(\A_t \cdot \mathcal{O}\cdot \A_{t-1}\cdot \mathcal{O}\cdot \A_{t-2} \cdots  \mathcal{O}\cdot \A_1\right)\left( \ketbra{\psi_0}_{\Reg{A},\Reg{Out}} \otimes \ketbra{\phi_0}_{\Reg{St}}\right)$$

    We will also define oracle algorithms with access to multiple oracles, possibly interacting on the same internal state $\Reg{St}$. In this case, we will arbitrarily sequence the queries made to the oracles. In particular, we define
    $$\A^{\mathcal{O}_1,\dots,\mathcal{O}_s}(\ket{\psi_0}) = \left(\A_t \cdots \mathcal{O}_1\cdot\A_{s+1}\cdot \mathcal{O}_s \cdot \A_s \cdots \mathcal{O}_2 \cdot \A_2 \cdot \mathcal{O}_1 \cdot \A_1\right)\left( \ketbra{\psi_0}_{\Reg{A},\Reg{Out}} \otimes \ketbra{\phi_0}_{\Reg{St}}\right)$$

    For $V$ a partial isometry, we will denote $\A^V=\A^{V(\cdot)V^\dagger}$. We will frequently denote $\vec{V} = (V,V^\dagger,V^*,V^T)$ and likewise $\A^{\vec{V}}=\A^{(V(\cdot)V^\dagger,V^\dagger(\cdot)V,V^*(\cdot)V^T,V^T(\cdot)V^*)}$.
\end{definition}

\begin{definition}
    A unitary oracle ensemble $\mathcal{U} = \{\mathcal{U}_n\}_{n\in \N}$ is a family of distributions over unitaries. We will sometimes refer to unitary oracle ensembles as unitary oracles.
\end{definition}

\begin{definition}
    Let $\ell_u(n)$ be some polynomial. The quantum Haar random unitary model (QHROM) of length $\ell_u$ is the unitary oracle ensemble $\U = \{\U_n\}_{n\in \N}$ where $\U_n \coloneqq \Haar(\{0,1\}^{\ell_u(n)})$.
\end{definition}

\begin{definition}
    Let $\ell_m,\ell_n$ be some polynomials and let $\mathcal{F}_n=\{f:\{0,1\}^{\ell_m(n)}\to \{0,1\}^{\ell_n(n)}\}$. Given $f$ in $\mathcal{F}$, define $A_f \ket{x}\ket{y} = \ket{x}\ket{y\xor f(x)}$. The quantum query secure random oracle model (QROM) mapping $\ell_m$ to $\ell_n$ is the unitary oracle ensemble $\U=\{\U_n\}_{n\in \N}$ where $\U_n$ samples $f\gets \mathcal{F}_n$ and outputs $A_f$.
\end{definition}
\subsection{Distance metrics on quantum channels}

\begin{definition}[Diamond norm]
    Let $\Phi$ be a quantum channel operating on $n$ qubits. The diamond norm of $\Phi$ is
    $$\norm{\Phi}_{\diamond}=\sup_{\ket{\phi}}\norm{(\Phi \otimes I_n)(\ketbra{\phi})}_1$$

    The diamond distance between two quantum channels $\Phi_1,\Phi_2$ operating on $n$ qubits is then the diamond norm of their difference
    $$\norm{\Phi_1-\Phi_2}_{\diamond}=\sup_{\ket{\phi}}\norm{(\Phi_1 \otimes I_n)(\ketbra{\phi})-(\Phi_2 \otimes I_n)(\ketbra{\phi})}_1$$
\end{definition}

\begin{theorem}\label{thm:querydiam}
    Let $\A^{(\cdot)}$ be any quantum oracle algorithm making at most $t$ queries. Let $\Phi_1,\Phi_2$ be two quantum channels. Then
    $$\abs{\Pr[\A^{\Phi_1}\to 1] - \Pr[\A^{\Phi_2} \to 1]} \leq t\cdot \norm{\Phi_1-\Phi_2}_{\diamond}$$
\end{theorem}

\begin{definition}[Operator norm]
    Let $V$ be a linear map. Then the operator norm of $V$ is
    $$\norm{V}_{op} = \sup_{\ket{\phi}} \norm{V\ket{\phi}}_2$$

    The operator distance between two linear operators $V_1,V_2$ is then the operator norm of their difference
    \begin{equation*}
        \begin{split}
            \norm{V_1-V_2}_{op} = \sup_{\ket{\phi}} \norm{V_1\ket{\phi} - V_2\ket{\phi}}_2
        \end{split}
    \end{equation*}
\end{definition}

\begin{proposition}\label{prop:diamop}
    Let $V_1,V_2$ be two partial isometries, let $\Phi_1$ be the operator $V_1(\cdot)V_1^\dagger$, and let $\Phi_2$ be the operator $V_2(\cdot)V_2^\dagger$. Then
    $$\norm{\Phi_1-\Phi_2}_\diamond \leq 2\norm{V_1-V_2}_{op}$$
\end{proposition}

The proof is included in~\cite{kre21}, Lemma 9. Note that while~\cite{kre21} proves this statement for unitaries, in fact the only fact about unitaries needed is that the absolute value of the eigenvalues of a unitary are bounded above by $1$ (i.e. the operator norm is bounded by $1$), and this also holds for partial isometries.

\begin{lemma}\label{lem:opprod}
    Let $V_1,V_2,W_1,W_2$ be partial isometries. Then
    $$\norm{V_1W_1 - V_2W_2}_{op} \leq \norm{V_1 - V_2}_{op}+\norm{W_1-W_2}_{op}$$
    and
    $$\norm{(V_1+V_2)-(W_1+W_2)}_{op} \leq \norm{V_1-W_1}_{op}+\norm{V_2-W_2}_{op}$$
\end{lemma}

\begin{proof}
    The first inequality follows from the following calculation (taken from~\cite{stackexchange}):
    \begin{equation}
        \begin{split}
            &\norm{V_1W_1-V_2W_2}_{op}\\
            =& \norm{V_1W_1-V_2W_2-V_2W_1+V_2W_1}_{op}\\
            \leq& \norm{V_1W_1 - V_2W_1}_{op} + \norm{V_2W_1-V_2W_2}_{op}\\
            =& \norm{(V_1-V_2)W_1}_{op} + \norm{V_2(W_1-W_2)}_{op}\\
            \leq& \norm{V_1-V_2}_{op}\norm{W_1}_{op} + \norm{V_2}_{op}\norm{W_1-W_2}_{op}\\
            \leq& \norm{V_1-V_2}_{op} + \norm{W_1-W_2}_{op}
        \end{split}
    \end{equation}
    since $\norm{W_1}_{op} = \norm{V_2}_{op}=1$ as they are partial isometries.

    The second inequality follows immediately from the triangle inequality.
\end{proof}

\begin{proposition}\label{prop:opinv}
    Let $V$ be any linear map, then
    $$\norm{V}_{op} = \norm{V^\dagger}_{op}$$
\end{proposition}

\begin{proposition}\label{prop:opbound}
    For all quantum states $\ket{\phi},\ket{\psi}$,
    $$\norm{\ketbra{\phi} -\ketbra{\psi}}_2 \leq \norm{\ketbra{\phi} - \ketbra{\psi}}_1 = 2\sqrt{1 - \abs{\bra{\phi}\ket{\psi}}^2}$$
\end{proposition}
\subsection{Path Recording Framework}

Here we present the path recording framework for simulating random unitaries defined in~\cite{MH25}, along with the extension to conjugate and transpose queries from~\cite{SMLBH25}. In particular, we expand their definition to encompass random unitaries over subsets of $\{0,1\}^n$. These restated theorems follow from the same proofs as in~\cite{MH25}.

\begin{definition}
    Given a set $S\subseteq \{0,1\}^n$, define $R_S=\mathcal{P}(\{0,1\}^{2n})$ to be the set of relations over $S$. 
    
    Furthermore, define $R^{inj}_{S}\subseteq R_S$ to be the set of injective relations over $S$. Formally,
    $$R^{inj}_S = \{R\subseteq \{0,1\}^{2n}:\forall(x,y)\neq(x',y')\in R, y\neq y'\}$$
    
    We similarly define $R_{S,t} \subseteq R_S$ to be
    $$R_{S,t} = \{R\in R_S:|R|\leq t\}$$
    $R^{inj}_{S,t}\subseteq R^{inj}_{S}$ is defined similarly.
    
    Let $D=\{(x_1,y_1),\dots,(x_\ell,y_\ell)\}\in R_{S,t}$. We say $\Ima(D) = \{y_1,\dots,y_\ell\}$ and $\Dom(D)=\{(x_1,\dots,x_\ell)\}$.
\end{definition}

\begin{definition}[Forward query path recording oracle]
    Let $S\subseteq \{0,1\}^n$ be any set and let $t_{max} \leq |S|$. Let $V_S$ be the partial isometry over $\mathcal{H}(\{0,1\}^n)\otimes \mathcal{H}(R^{inj}_S)$ defined as follows: for $x\in S,D\in R^{inj}_{S,t_{max}-1}$
    $$V\ket{x}\ket{D} \mapsto \frac{1}{\sqrt{|S|-|\Ima(D)|}}\sum_{y\in S\setminus \Ima(D)}\ket{y}\ket{D\cup \{(x,y)\}}$$
    for $x\notin S,D\in R^{inj}_S$,
    $$V_S\ket{x}\ket{D} \mapsto \ket{x}\ket{D}$$
\end{definition}

\begin{theorem}
    Let $\Reg{A}$ be an arbitrary Hilbert space, let $\Reg{In}=\mathcal{H}(\{0,1\}^n)$ and let $\Reg{S}=\mathcal{H}(R_{S}^{inj})$. Let $\mathcal{A}^{(\cdot)}$ be any $t$-query algorithm operating on registers $\Reg{A},\Reg{In}$. Let $t\leq t_{max}$ and let $V_{S,t_{max}}$ operate on registers $\Reg{In},\Reg{St}$. Then,
    $$TD\left(\E_{U\gets \Haar(S)}\left[\A^{U}\right] , \A^{V}\right)\leq \frac{2t(t-1)}{|S|+1}$$
\end{theorem}

\begin{theorem}[Efficient simulation over $\{0,1\}^n$ - \cite{MH25} Appendix A]
    There exists an efficient algorithm $\Sim(1^{t_{max}})$ acting on registers $\Reg{In},\Reg{St}$ such that the following holds.
    Let $\mathcal{A}^{(\cdot)}$ be any $t$-query algorithm operating on registers $\Reg{A},\Reg{In}$, where register $\Reg{In}$ is over $\mathcal{H}(\{0,1\}^n)$. Let $t\leq t_{max}$ and let $V^{\{0,1\}^n,t_{max}}$ operate on registers $\Reg{In},\Reg{St}$.
    $$\A^{\Sim(1^{t_{max}})} = \A^{V_{\{0,1\}^n,t_{max}}}$$
\end{theorem}

\begin{proof}[Proof Sketch]
    The simulator will represent elements of $R^{inj}_{\{0,1\}^n,t_{max}}$ as sorted lists of pairs of elements of $\{0,1\}^n$. That is, for $D\in R^{inj}_{\{0,1\}^n,t_{max}}$, define $Rep(D)\in (\{0,1\}^{2n}\cup \{(\bot,\bot)\})^{t_{max}}$ as follows: let $(x_1,y_1),\dots,(x_t,y_t)$ be the elements of $D$ in sorted order.
    $$Rep(D) = ((x_1,y_1),\dots,(x_t,y_t),\underbrace{(\bot,\bot),\dots,(\bot,\bot)}_{t_{max}-t})$$

    $\Reg{St}$ will then be a register over $\mathcal{H}((\{0,1\}^{2n}\cup \{(\bot,\bot)\})^{t_{max}})$ which will always contain an element of $R^{inj}_{S,t_{max}}$.

    $\Sim(1^{t_{max}})\ket{x}\ket{D}$ will then act as follows
    \begin{enumerate}
        \item Generate the state $\sum_{y\in \{0,1\}^n\setminus \Ima(D)} \ket{y}$.
        \item Perform the isometry
        $\ket{x}\ket{D}\ket{y} \mapsto \ket{x}\ket{D}\ket{y}\ket{D\cup \{(x,y)\}}$. This leaves the state of the simulator as
        $$\sum_{y\in \{0,1\}^n\setminus \Ima(D)}\ket{x}\ket{D}\ket{y}\ket{D\cup\{(x,y)\}}$$
        \item Finally, uncompute the registers containing $x$ and $D$ (which is possible since these are exactly determined by $y,D\cup\{(x,y)\}$). This results in the final state
        $$\sum_{y\in \{0,1\}^n\setminus \Ima(D)}\ket{y}\ket{D\cup\{(x,y)\}}$$
    \end{enumerate}
\end{proof}

Note that the only part of this simulator which requires knowing that the set $S$ is $\{0,1\}^n$ is computing 
$$\sum_{y\in S\setminus \Ima(D)} \ket{y} = \sum_{y\in \{0,1\}^n\setminus \Ima(D)} \ket{y}$$
For an arbitrary $S\subseteq \{0,1\}^n$, this state can be constructed by the following process:
\begin{enumerate}
    \item Construct the state $\sum_{y\in \{0,1\}^n}\ket{y}$
    \item Apply the measurement checking if $y\in S$
    \item Apply the measurement checking if $y\notin D$
\end{enumerate}
If both measurements succeed, then the residual state is exactly
$$\sum_{y\in S\setminus \Ima(D)} \ket{y}$$
and both measurements succeed with probability $\frac{|S|-|\Ima(D)|}{2^n}$. These measurements can further be performed using a reflection oracle for $S$. This leads to the following corollary

\begin{corollary}\label{cor:forwardsim}
    There exists an efficient oracle algorithm $\Sim^{(\cdot)}(1^{t_{max}})$ acting on registers $\Reg{In},\Reg{St}$ such that the following holds.
    Let $\mathcal{A}^{(\cdot)}$ be any $t$-query algorithm operating on registers $\Reg{A},\Reg{In}$, where register $\Reg{In}$ is over $\mathcal{H}(\{0,1\}^n)$. Let $S\subseteq [N]$ and let $R_S$ be reflection around $S$. Let $t\leq t_{max}$. Then,
    $$\A^{\Sim^{R_S}(1^{t_{max}})} = \A^{V_{S,t_{max}}}$$
\end{corollary}

\begin{definition}[Inverse, conjugate, transpose capable path recording oracle~\cite{SMLBH25}]
    Let $S\subseteq \{0,1\}^n$ be any set and let $t_{max} \leq |S|$. Let $V^L_S,V^R_S,\ol{V}_S^L,\ol{V}_S^R$ be the partial isometries over $\mathcal{H}(\{0,1\}^n)\otimes \mathcal{H}(R_S)\otimes \mathcal{H}(R_S)$ defined as follows: 
    for $x\in S,$ for $L,R\in R_{S,t_{max}-1}$
    $$V_S^L\ket{x}\ket{L}\ket{R} \mapsto \frac{1}{\sqrt{|S|-|\Ima(L\cup R)|}}\sum_{y\in S\setminus \Ima(L\cup R)}\ket{y}\ket{L\cup \{(x,y)\}}\ket{R}$$
    $$\ol{V}_S^L\ket{x}\ket{L}\ket{R} \mapsto \frac{1}{\sqrt{|S|-|\Ima(L\cup R)|}}\sum_{y\in S\setminus \Ima(L\cup R)}\ket{y}\ket{L}\ket{R\cup \{(x,y)\}}$$
    for $y\in S,$ for $L,R\in R_{S,t_{max}-1}$
    $$V^R_S\ket{y}\ket{L}\ket{R} \mapsto \frac{1}{\sqrt{|S|-|\Dom(L\cup R)|}}\sum_{x\in S\setminus \Dom(L\cup R)}\ket{x}\ket{L}\ket{R\cup \{(x,y)\}}$$
    $$\ol{V}^R_S\ket{y}\ket{L}\ket{R} \mapsto \frac{1}{\sqrt{|S|-|\Dom(L\cup R)|}}\sum_{x\in S\setminus \Dom(L\cup R)}\ket{x}\ket{L\cup \{(x,y)\}}\ket{R}$$
    for $x\notin S,(L,R)\in R^2_S$,
    $$V^L_S\ket{x}\ket{L}\ket{R}=V^R_S\ket{x}\ket{L}\ket{R} = \ket{x}\ket{L}\ket{R}$$

    Finally, define 
    $$V_S = V^L_S (I-V^R\cdot V^{R,\dagger})+(I-V^L_S\cdot V^{L,\dagger}_S)V^{R,\dagger}_S$$
    $$\ol{V}_S = \ol{V}^L_S (I-\ol{V}^R_S\cdot \ol{V}^{R,\dagger}_S)+(I-\ol{V}^L_S\cdot \ol{V}^{L,\dagger}_S)\ol{V}^{R,\dagger}_S$$
\end{definition}

\begin{theorem}\label{thm:pathrecordingreverse}
    Let $\mathcal{A}^{(\cdot)}$ be any $t$-query algorithm operating on registers $AB$, where register $A$ is over $\mathcal{H}(\{0,1\}^n)$. Let $t\leq t_{max}$ and let $V_{S,t_{max}}$ operate on registers $AR$. Define $\vec{V}_{S,t_{max}} = (V_{S,t_{max}},V_{S,t_{max}}^\dagger,\ol{V}_{S,t_{max}},\ol{V}_{S,t_{max}}^\dagger)$ Then,
    $$TD\left(\E_{U\gets \Haar(S)}\left[\A^{\vec{U}}\right] , \A^{\vec{V}_{S,t_{max}}}\right)\leq \frac{9t(t+1)}{|S|^{1/8}}$$
\end{theorem}

Observing the efficient implementation of the inverse capable path recording oracle from~\cite{MH25}, it is fortunately the case that~\Cref{cor:forwardsim} generalizes to the inverse setting as well. Furthermore, by applying the symmetric construction,~\Cref{cor:forwardsim} also applies in the case of conjugate and transpose queries.

\begin{corollary}\label{cor:effpr}
    There exists efficient oracle algorithms $\Sim^{(\cdot)}(1^{t_{max}}),\ol{\Sim}^{(\cdot)}(1^{t_{max}})$ acting on registers $\Reg{In},\Reg{St}$ such that the following holds.
    Let $\mathcal{A}^{(\cdot)}$ be any $t$-query algorithm operating on registers $\Reg{A},\Reg{In}$, where register $A$ is over $\mathcal{H}(\{0,1\}^n)$. Let $S\subseteq \{0,1\}^n$ and let $R_{S}$ be reflection around $S$
    $$R_S\ket{x} \mapsto (-1)^{\mathds{1}[x\in S]}\ket{x}$$
    Let $t\leq t_{max}$. Define 
    $$\Ora{\Sim}^{R_S}(1^{t_{max}}) = (\Sim^{R_S}(1^{t_{max}}),(\Sim^{R_S}(1^{t_{max}}))^\dagger,\ol{\Sim}^{R_S}(1^{t_max}),(\ol{\Sim}^{R_S}(1^{t_{max}}))^\dagger)$$
    then we have
    $$\A^{\Ora{\Sim}^{R_S}(1^{t_{max}})} = \A^{\vec{V}_{S,t_{max}}}$$
    when clear from context, we will omit $1^{t_{max}}$.
\end{corollary}
\section{Unitary reprogramming lemma}\label{sec:reprog}

\begin{lemma}[Unitary reprogramming lemma\label{lem:reprogramming}]
    Let $N,M_1,M_2,t$ be functions such that $\frac{M_1}{M_2}=\negl(n)$, $\frac{M_2}{N} = \negl(n)$, and $t=\poly(n)$. Let $S_1\subseteq [N]$ be any set of size $|S_1| = M_1$. Define $\O$ as follows
    \begin{enumerate}
        \item Sample $S_2 \subseteq [N]$ uniformly at random such that $S_1\subseteq S_2$ and $|S_2| = M_2$.
        \item Sample $U_1$ a Haar random unitary acting on $S_2$.
        \item Sample $U_2$ a Haar random unitary acting on $[N]\setminus S_2$.
        \item Let $\O = U_1U_2$.
    \end{enumerate}
    Let $\A^{(\cdot)}$ be any $t$ query oracle algorithm (allowed to make inverse, transpose, and conjugate queries). Then
    $$\abs{\Pr[\A^{\Ora{\mathcal{O}}}(S_1) \to 1] - \Pr_{U\gets \Haar([N])}[\A^{\Ora{U}}(S_1)\to 1]} \leq \negl(n)$$
\end{lemma}

We will hybrid over a series of partial isometries $\Ora{V}=(V_i,V_i^\dagger,\ol{V_i},\ol{V_i}^\dagger)$ representing forward, inverse, transpose, and conjugate queries respectively. We will frequently define operators $V_i^{L_1},\allowbreak V_i^{R_1},\allowbreak V_i^{L_2}, \allowbreak V_i^{R_2},\allowbreak \ol{V}_i^{L_1},\allowbreak \ol{V}_i^{R_1},\allowbreak \ol{V}_i^{L_2},\allowbreak \ol{V}_i^{R_2}$, in which case, in the style of path recording, we will implicitly set (unless otherwise stated)
\begin{equation}
    \begin{split}
        V_i^L &= V_i^{L_1} V_i^{L_2}\\
        V_i^R &= V_i^{R_1} V_i^{R_2}\\
        V_i &= V_i^L \cdot (I-V_i^R\cdot V_i^{R,\dagger})+(I-V_i^L\cdot V_i^{L,\dagger})\cdot V_i^{R,\dagger}\\
        \ol{V}_i &= \ol{V}_i^L \cdot (I-\ol{V}_i^R\cdot \ol{V}_i^{R,\dagger})+(I-\ol{V}_i^L\cdot \ol{V}_i^{L,\dagger})\cdot \ol{V}_i^{R,\dagger}
    \end{split}
\end{equation}
Each of these operators will operate on a register $\Reg{In}$ over $\mathcal{H}([N])$ and $\Reg{St}$ over $\mathcal{H}(R_{[N]})^{\otimes 4}\otimes \mathcal{H}(\mathcal{P}([N]))$. This space is spanned by states of the form $\ket{x}\ket{L_1}\ket{L_2}\ket{R_1}\ket{R_2}\ket{S_2}$ where $x$ is in $\{0,1\}^n$, the $L_1,L_2,R_1,R_2$ are relations in $R_{[N]}$, and $S_2$ is a subset of $[N]$.

We will further define notation for basis vectors $D$ of the space $\mathcal{H}([N])\otimes \mathcal{H}(R_{[N]})^{\otimes 4}$. These will be elements of $\mathcal{DB} = [N]\times R_{[N]}^{4}$. Given $D=(x,L_1,R_1,L_2,R_2)$, we will write $D_x=x,D_{L_1}=L_1,D_{R_1}=R_1,D_{L_2}=L_2,D_{R_2}=R_2$. We will define $\Dom(D_b)=\Dom(L_b)\cup \Dom(R_b)$ and $\Ima(D_b)=\Ima(L_b)\cup \Ima(R_b)$. For ease of notation, we will write $D_{1,all} = \Ima(D_1)\cup \Dom(D_1) \cup S_1$ and $D_{2,all} = \Ima(D_2)\cup \Dom(D_2)$ (note the asymmetry here). We will write $D_{all}\coloneqq D_{1,all}\cup D_{2,all}\cup \{D_x\}$. We will also define $D_y^{L_1,f} = (y,L_1\cup \{(x,y)\},R_1,L_2,R_2)$ and $D_y^{L_1,b} = (y,L_1\cup \{(y,x)\},R_1,L_2,R_2)$. We define $D_y^{R_1,f},D_y^{L_2,f},D_y^{R_2,f},D_y^{R_1,b},D_y^{L_2,b},D_y^{R_2,b}$ analogously.

For all hybrids, we will work with the initial internal state 
$$\ket{\phi_0} = \frac{1}{\sqrt{{N\choose M_2}}}\ket{\emptyset}^{\otimes 4} \sum_{\substack{S_2\in \mathcal{P}([N])\\|S_2|=M_2}} \ket{S_2}$$

We set $V_0 = U_1U_2$ and $\Ora{V_0} = (V_0, V_0^\dagger, V_0^*, V_0^T)$ for $U_1,U_2$ distributed as in the lemma statement. We set $V_f = U$ and $\Ora{V_f} = (V_f, V_f^\dagger, V_f^*, V_f^T)$ for $U$ a Haar random unitary over $\mathcal{H}([N])$. Our remaining hybrids are then as follows
\begin{enumerate}
    \item In this hybrid we replace the truly random unitaries with a path recording variant. On input $\ket{D}\ket{S_2}$ such that $D_x\in S_2$, define
    \begin{equation}
        \begin{split}
            V_{1}^{L_1} \ket{D}\ket{S_2} &= \frac{1}{\sqrt{|S_2\setminus \Ima(D_1)|}}\sum_{y\in S_2\setminus \Ima(D_1)} \ket{D_y^{L_1,f}}\ket{S_2} \\
            V_{1}^{R_1} \ket{D}\ket{S_2} &= \frac{1}{\sqrt{|S_2\setminus \Dom(D_1)|}}\sum_{y\in S_2\setminus \Dom(D_1)} \ket{D_y^{R_1,b}}\ket{S_2} \\
            \ol{V_{1}}^{L_1} \ket{D}\ket{S_2} &= \frac{1}{\sqrt{|S_2\setminus \Ima(D_1)|}}\sum_{y\in S_2\setminus \Ima(D_1)} \ket{D_y^{R_1,f}}\ket{S_2}\\
            \ol{V_{1}}^{R_1} \ket{D}\ket{S_2} &= \frac{1}{\sqrt{|S_2\setminus \Dom(D_1)|}}\sum_{y\in S_2\setminus \Dom(D_1)}\ket{D_y^{L_1,b}}\ket{S_2}\\
        \end{split}
    \end{equation}
    with $V_{1}^{L_2},V_1^{R_2},\ol{V_1}^{L_2},\ol{V_1}^{R_2}$ acting as identity on all such states. On input $\ket{D}\ket{S_2}$ such that $D_x \in [N]\setminus S_2$, define
    \begin{equation}
        \begin{split}
            V_{1}^{L_2} \ket{D}\ket{S_2} &= \frac{1}{\sqrt{|[N]\setminus (S_2\cup \Ima(D_2))|}}\sum_{y\in [N]\setminus (S_2\cup \Ima(D_2))} \ket{D_y^{L_2,f}}\ket{S_2} \\
            V_{1}^{R_2} \ket{D}\ket{S_2} &= \frac{1}{\sqrt{|[N]\setminus (S_2\cup \Dom(D_2))|}}\sum_{y\in [N]\setminus (S_2\cup \Dom(D_2))}\ket{D_y^{R_2,b}}\ket{S_2}\\
            \ol{V_{1}}^{L_2} \ket{D}\ket{S_2} &= \frac{1}{\sqrt{|[N]\setminus (S_2\cup \Ima(D_2))|}}\sum_{y\in [N]\setminus (S_2\cup \Ima(D_2))} \ket{D_y^{R_2,f}}\ket{S_2}\\
            \ol{V_{1}}^{R_2} \ket{D}\ket{S_2} &= \frac{1}{\sqrt{|[N]\setminus (S_2\cup \Dom(D_2))|}}\sum_{y\in [N]\setminus (S_2\cup \Dom(D_2))}\ket{D_y^{L_2,b}}\ket{S_2}\\
        \end{split}
    \end{equation}
    with $V_{1}^{L_1},V_1^{R_1},\ol{V_1}^{L_1},\ol{V_1}^{R_1}$ acting as identity on all such states.
    \item In this hybrid we adjust the domains for which new outputs of the path-recording oracle are sampled. On input $\ket{D}\ket{S_2}$ such that $D_x\in S_2$, define
    \begin{equation}
        \begin{split}
            V_{2}^{L_1} \ket{D}\ket{S_2} &= \frac{1}{\sqrt{|S_2\setminus D_{1,all}|}}\sum_{y\in S_2\setminus D_{1,all}} \ket{D_y^{L_1,f}}\ket{S_2} \\
            V_{2}^{R_1} \ket{D}\ket{S_2} &= \frac{1}{\sqrt{|S_2\setminus D_{1,all}|}}\sum_{y\in S_2\setminus D_{1,all}} \ket{D_y^{R_1,b}}\ket{S_2}\\
            \ol{V_{2}}^{L_1} \ket{D}\ket{S_2} &= \frac{1}{\sqrt{|S_2\setminus D_{1,all}|}}\sum_{y\in S_2\setminus D_{1,all}} \ket{D_y^{R_1,f}}\ket{S_2}\\
            \ol{V_{2}}^{R_1} \ket{D}\ket{S_2} &= \frac{1}{\sqrt{|S_2\setminus D_{1,all}|}}\sum_{y\in S_2\setminus D_{1,all}}\ket{D_y^{L_1,b}}\ket{S_2}\\
        \end{split}
    \end{equation}
    with $V_{2}^{L_2},V_2^{R_2},\ol{V_2}^{L_2},\ol{V_2}^{R_2}$ acting as identity on all such states. On input $\ket{D}\ket{S_2}$ such that $D_x \in [N]\setminus S_2$, define
    \begin{equation}
        \begin{split}
            V_{2}^{L_2} \ket{D}\ket{S_2} &= \frac{1}{\sqrt{|[N]\setminus (S_2\cup D_{all})|}}\sum_{y\in [N]\setminus (S_2\cup D_{all})} \ket{D_y^{L_2,f}}\ket{S_2} \\
            V_{2}^{R_2} \ket{D}\ket{S_2} &= \frac{1}{\sqrt{|[N]\setminus (S_2\cup D_{all})|}}\sum_{y\in [N]\setminus (S_2\cup D_{all})}\ket{D_y^{R_2,b}}\ket{S_2}\\
            \ol{V_{2}}^{L_2} \ket{D}\ket{S_2} &= \frac{1}{\sqrt{|[N]\setminus (S_2\cup D_{all})|}}\sum_{y\in [N]\setminus (S_2\cup D_{all})} \ket{D_y^{R_2,f}}\ket{S_2}\\
            \ol{V_{2}}^{R_2} \ket{D}\ket{S_2} &= \frac{1}{\sqrt{|[N]\setminus (S_2\cup D_{all})|}}\sum_{y\in [N]\setminus (S_2\cup D_{all})}\ket{D_y^{L_2,b}}\ket{S_2}\\
        \end{split}
    \end{equation}
    with $V_{2}^{L_1},V_2^{R_1},\ol{V_2}^{L_1},\ol{V_2}^{R_1}$ acting as identity on all such states.
    \item 
    \begin{enumerate}
        \item We call a set $S_2$ valid if the following hold
        \begin{enumerate}
            \item $D_{1,all}\subseteq S_2$
            \item $D_{2,all}\cap S_2 = \emptyset$
            \item If $D_x\notin D_{1,all}$, then $D_x\notin S_2$
        \end{enumerate}
        We define $\Valid(D) = \{S_2\subseteq [N]:|S_2|=M_2\text{ and }S_2\text{ is valid for }D\}$.

        Intuitively, the first two conditions state the $D_1$ contains only values in $S_2$ and $D_2$ contains no values in $S_2$ (which would occur if $D_1$ and $D_2$ were the internal states of independent Haar random unitaries). The last condition enforces that the adversary never guesses a point in $S_2$ without being previously given that information (either via $S_1$ or given in a previous query). 
        \item Define the state 
        $$\ket{ValSt_{D}} = \frac{1}{\sqrt{|\Valid(D)|}}\ket{D}\sum_{S_2\in \Valid(D)} \ket{S_2}$$
        Note that for $D\neq D'$, $\ket{ValSt_D}$ and $\ket{ValSt_{D'}}$ are orthogonal. Define the projector
        $$\Pi_{good} = \sum_{D} \ketbra{ValSt_D}$$
        \item For all $X$, define $V_{3}^X = \Pi_{good} V_{2}^X \Pi_{good}$ and $\ol{V_{3}}^X=\Pi_{good} \ol{V_{2}}^X \Pi_{good}$.
    \end{enumerate}
    \item 
    \begin{enumerate}
    \item Define the following isometry $ConS$ to reconstruct the $S_2$ register from the rest of the state. That is,
    $$ConS\ket{D} = \frac{1}{\sqrt{|\Valid(D)|}}\ket{D}\sum_{S_2\in \Valid(D)} \ket{S_2}$$
    Note $\Pi_{good}=\Ima(ConS)$.
    \item We will define a new set of partial isometries $W^{L_1}_4,W^{R_1}_4,\dots$ acting only on $\mathcal{H}([N])\otimes \mathcal{H}(R_{[N]})^{\otimes 4}$ as follows (which will implicitly define $\Ora{W_4}=(W_4,W_4^\dagger,\ol{W}_4,\ol{W}_4^\dagger)$). On input $\ket{D}$ such that $D_x\in D_{1,all}$, define
    \begin{equation}
        \begin{split}
            W^{L_1}_4\ket{D} = \frac{1}{\sqrt{|[N]\setminus D_{all}|}} \sum_{y\in [N]\setminus D_{all}} \ket{D_y^{L_1,f}}\\
            W^{R_1}_4\ket{D} = \frac{1}{\sqrt{|[N]\setminus D_{all}|}} \sum_{y\in [N]\setminus D_{all}} \ket{D_y^{R_1,b}}\\
            \ol{W_4}^{L_1}\ket{D} = \frac{1}{\sqrt{|[N]\setminus D_{all}|}} \sum_{y\in [N]\setminus D_{all}} \ket{D_y^{R_1,f}}\\
            \ol{W_4}^{R_1}\ket{D} = \frac{1}{\sqrt{|[N]\setminus D_{all}|}} \sum_{y\in [N]\setminus D_{all}} \ket{D_y^{L_1,b}}\\
        \end{split}
    \end{equation}
    with $W_4^{L_2},W_4^{R_2},\ol{W}_4^{L_2},\ol{W}_4^{R_2}$ acting as identity on all such states. On input $\ket{D}$ such that $D_x\notin D_{1,all}$, define
    \begin{equation}
        \begin{split}
            W^{L_2}_4\ket{D} = \frac{1}{\sqrt{|[N]\setminus D_{all}|}} \sum_{y\in [N]\setminus D_{all}} \ket{D_y^{L_2,f}}\\
            W^{R_2}_4\ket{D} = \frac{1}{\sqrt{|[N]\setminus D_{all}|}} \sum_{y\in [N]\setminus D_{all}} \ket{D_y^{R_2,b}}\\
            \ol{W_4}^{L_2}\ket{D} = \frac{1}{\sqrt{|[N]\setminus D_{all}|}} \sum_{y\in [N]\setminus D_{all}} \ket{D_y^{R_2,f}}\\
            \ol{W_4}^{R_2}\ket{D} = \frac{1}{\sqrt{|[N]\setminus D_{all}|}} \sum_{y\in [N]\setminus D_{all}} \ket{D_y^{L_2,b}}\\
        \end{split}
    \end{equation}
    with $W_4^{L_1},W_4^{R_1},\ol{W}_4^{L_1},\ol{W}_4^{R_1}$ acting as identity on all such states.
    \item For all $X$, define $V_4^{X} = ConS\cdot W^{X}_4 \cdot ConS^{\dagger}$ and $\ol{V}_4^X = ConS\cdot \ol{W}^X_4 \cdot ConS^{\dagger}$.
    \end{enumerate}
    \item We set $V_5 = W_4 \otimes I$ and $\ol{V_5}=\ol{W_4}\otimes I$.
    \item $V_6$ will correspond to a single path-recording oracle no longer split into two parts. In particular, define partial isometries $W_6^{L_1},W_6^{R_1},\ol{W_6}^{L_1},\ol{W_6}^{R_1}$ acting on registers $\mathcal{H}([N])\otimes \mathcal{H}(R_{[N]})^{\otimes 4}$ as follows
    \begin{equation}
        \begin{split}
            W_6^{L_1} \ket{D} \mapsto \frac{1}{\sqrt{|[N]\setminus D_{all}|}} \sum_{y\in [N]\setminus D_{all}} \ket{D_y^{L_1,f}}\\
            W_6^{R_1} \ket{D} \mapsto \frac{1}{\sqrt{|[N]\setminus D_{all}|}} \sum_{y\in [N]\setminus D_{all}} \ket{D_y^{R_1,b}}\\
            \ol{W_6}^{L_1} \ket{D} \mapsto \frac{1}{\sqrt{|[N]\setminus D_{all}|}} \sum_{y\in [N]\setminus D_{all}} \ket{D_y^{R_1,f}}\\
            \ol{W_6}^{R_1} \ket{D} \mapsto \frac{1}{\sqrt{|[N]\setminus D_{all}|}} \sum_{y\in [N]\setminus D_{all}} \ket{D_y^{L_1,b}}
        \end{split}
    \end{equation}
    We will define 
    $$W_6 = W_6^{L_1} (I-W_6^{R_1}\cdot W_6^{R_1,\dagger})+(I-W_6^{L_1}\cdot W_6^{L_1,\dagger})W_6^{R_1,\dagger}$$
    $$\ol{W}_6 = \ol{W}_6^{L_1} (I-\ol{W}_6^{R_1}\cdot \ol{W}_6^{R_1,\dagger})+(I-\ol{W}_6^{L_1}\cdot \ol{W}_6^{L_1,\dagger})\ol{W}_6^{R_1,\dagger}$$
    Then we will set $V_6 = W_6 \otimes I$ and $\ol{V_6} = \ol{W_6} \otimes I$.
    \item $V_7$ will correspond to the standard path-recording oracle for a single Haar random unitary.
\end{enumerate}

Before arguing about the hybrids, we will first show that each of the claimed partial isometries are actually partial isometries. This follows immediately from the following observation
\begin{lemma}\label{lem:outputortho}
    For all $A\neq B\in \mathcal{DB}$, for all $y\in [N]\setminus (A_1\cap B_1)$,
    \begin{enumerate}
        \item $\braket{A_{y}^{L_1,f}}{B_{y}^{L_1,f}}=0$
        \item $\braket{A_{y}^{L_1,b}}{B_{y}^{L_1,b}} = 0$
        \item $\braket{A_{y}^{R_1,f}}{B_{y}^{R_1,f}}=0$
        \item $\braket{A_{y}^{R_1,b}}{B_{y}^{R_1,b}} = 0$
    \end{enumerate}
    and for all $y \in [N] \setminus (A_2\cap B_2)$,
    \begin{enumerate}
        \item $\braket{A_{y}^{L_2,f}}{B_{y}^{L_2,f}}=0$
        \item $\braket{A_{y}^{L_2,b}}{B_{y}^{L_2,b}} = 0$
        \item $\braket{A_{y}^{R_2,f}}{B_{y}^{R_2,f}}=0$
        \item $\braket{A_{y}^{R_2,b}}{B_{y}^{R_2,b}} = 0$
    \end{enumerate}
\end{lemma}

\begin{proof}

    We will explicitly write the proof for $\braket{A_{y}^{L_1,f}}{B_{y}^{L_1,f}}=0$, and the rest follow by the same argument. 
    
    If $(A_{L_2},A_{R_1},A_{R_2})\neq (B_{L_2},B_{R_1},B_{L_2})$ then this is trivial. Thus, without loss of generality, we assume $(A_{L_2},A_{R_1},A_{R_2})= (B_{L_2},B_{R_1},B_{L_2})$. We will then show that $A_{L_1}\cup \{(A_x,y)\} \neq B_{L_1}\cup \{(B_x,y)\}$. 

    If $A_{L_1}=B_{L_1}$, since $A_x\neq B_x$ we by definition have $A_{L_1}\cup \{(A_x,y)\} \neq B_{L_1}\cup \{(B_x,y)\}$. If $A_x=B_x$ and $A_{L_1}\neq B_{L_1}$, then again by definition we have $A_{L_1}\cup \{(A_x,y)\} \neq B_{L_1}\cup \{(A_x,y)\}$. If $|A_{L_1}|\neq |B_{L_1}|$, then the result trivially follows. In the final case, we consider $A_x\neq B_x$, $A_{L_1} \neq B_{L_1}$, and $|A_{L_1}|=|B_{L_1}|$. In this case, the only way we can have $A_{L_1}\cup \{(A_x,y)\} = B_{L_1}\cup \{(B_x,y)\}$ is if $A_{L_1}\setminus \{(B_x,y)\} = B_{L_1} \setminus \{(A_x,y)\}$. 

    But since $A_x\neq B_x$, this implies that $\{(B_x,y)\}\in A_{L_1}$ and $\{(A_x,y)\}\in B_{L_1}$. And so, $y\in A_1 \cap B_1$, which is a contradiction.
\end{proof}

To conclude the proof, we argue the formal proofs that $$\norm{\A^{\Ora{V_i}} - \A^{\Ora{V_{i-1}}}}_2 \leq \negl(n)$$ for all $i$. \Cref{lem:reprogramming} follows immediately. As these arguments are all technically involved, we will defer them to \ifnum\llncs=0~\Cref{sec:rephyb} \else Appendix B of the full version \fi and provide a proof sketch of the hybrids here. 

\begin{enumerate}
    \item $\norm{\A^{\Ora{V_1}} - \A^{\Ora{V_0}}}_2 \leq \negl(n)$: This follows from path recording (\Cref{thm:pathrecordingreverse}), although some care needs to be taken to verify that $\Ora{V_1}$ is indeed the correct instantiation of the path recording oracle.
    \item $\norm{\A^{\Ora{V_2}} - \A^{\Ora{V_{1}}}}_2 \leq \negl(n)$: $V_1^L$ and $V_2^L$ are close in diamond distance, since $V_2^L$ chooses $y$ from a space which has negligibly fewer options. A similar argument holds for the other components of $\Ora{V_1},\Ora{V_2}$. Since all maps here are partial isometries, this implies that $\Ora{V_1}$ and $\Ora{V_2}$ are close in diamond distance and so by induction the final states of the adversary when interacting with both maps will be close.
    \item $\norm{\A^{\Ora{V_3}} - \A^{\Ora{V_{2}}}}_2 \leq \negl(n)$: We will show that all maps making up $\A^{\Ora{V_2}}$ approximately commute with $\Pi_{good}$. And so the claim follows by induction. In particular, let $\ket{\phi}\in \Span(\Pi_{good})$. We can observe directly that $V_2^{L_1}\ket{\phi}$ will also be in $\Span(\Pi_{good})$. Some more care needs to be taken to ensure that $V_2^{L_1,\dagger}\ket{\phi}$ will be close to a state in $\Span(\Pi_{good})$, and will follow from a direct combinatorial argument given in detail in \ifnum\llncs=0~\Cref{sec:combinatorics}\else Appendix C of the full version\fi. Similarly, for any map $A$ operating only in $\Reg{In}$, $A\ket{\phi}$ will be close to a state in $\Pi_{good}$, with error corresponding to the event that the new component in register $\Reg{In}$ lands in $S_2\setminus D_{1,all}$. Since this is effectively a random set, the corresponding error is neglgible. 
    \item $\Ora{V_3}=\Ora{V_4}$: In particular, $\Ora{W_4}$ is an exact description of the map you get when you conjugate $\Ora{V_3}$ by $ConS$. 
    \item $\A^{\Ora{V_5}} = \A^{\Ora{V_4}}$: This follows immediately from the map that $\Ora{V_4}$ and $\Ora{V_5}$ are identical up to conjugation by the isometry $ConS$ which acts only on the internal state register.
    \item $\A^{\Ora{V_6}} = \A^{\Ora{V_5}}$: We develop a partial isometry $T$ acting only on the internal register such that $\Ora{V_6}$ and $\Ora{V_5}$ are identical up to conjugation by $T$. In particular, $T$ will describe how to split a database $L,R$ coming from $\A^{\Ora{V_6}}$ into a database $(L_1,R_1,L_2,R_2)$ consistent with $\A^{\Ora{V_5}}$. The idea is that we partition all $(x,y)\in L\cup R$ into either $(L_1,R_1)$ or $(L_2,R_2)$. In particular, we will place $(x,y)$ in $(L_1,R_1)$ if and only if there is a "path" in $(L,R)$ reaching $(x,y)$ from a string $s\in S_2$. That is, there exists $(x_1,y_1),(x_2,y_2),\dots,(x_t,y_t)\in L\cup R$ such that
    \begin{enumerate}
        \item $(x_t,y_t)=(x,y)$
        \item $x_1\in S_1$ or $y_1\in S_2$
        \item $x_i=y_{i-1}$
    \end{enumerate}
    \item $\norm{\A^{\Ora{V_7}} - \A^{\Ora{V_{6}}}}_2 \leq \negl(n)$: This follows from the same argument as for $\norm{\A^{\Ora{V_2}} - \A^{\Ora{V_{1}}}}_2 \leq \negl(n)$
    \item $\norm{\A^{\Ora{V_f}} - \A^{\Ora{V_7}}}_2 \leq \negl(n)$: This follows from path recording (\Cref{thm:pathrecordingreverse}).
\end{enumerate}
\section{Unclonable encryption compiler}\label{sec:mainthm}

\begin{definition}
    A (many-time) unclonable encryption scheme with keys of length $\ell_k(n)$ for messages of length $\ell_m(n)$ is a pair of uniform quantum algorithms $\Enc,\Dec$ with the following syntax
    \begin{enumerate}
        \item $\Enc(1^n,k,m) \to \rho_{k,m}$ takes as input a security parameter $n$, a key $k$, and a message $m$, and outputs a mixed state $\rho_{k,m}$.
        \item $\Dec(1^n, k, \rho) \to m$ takes as input a security parameter $n$, a key $k$, and a quantum ciphertext $\rho$, and outputs a decrypted message $m$.
    \end{enumerate}
    satisfying the following properties.
    \begin{enumerate}
        \item (Perfect) correctness: for all $n\in \N, m\in \{0,1\}^{\ell_m(n)}$,
        $$\BigPr{\Dec(1^n, k, \rho_{k,m}) = m}{k\gets \{0,1\}^{\ell_k(n)}\\\Enc(1^n, k, m)\to \rho_{k,m}} = 1.$$
        \item Security:
        For all (non-uniform) QPT algorithms $A,B$ and for all (non-uniform) QPT oracle algorithms $C^{(\cdot)}$, define $GAME_{A,B,C}^n(\Enc,\Dec)$ to be the following game
        \begin{enumerate}
            \item Sample $k \gets \{0,1\}^{\ell_k(n)}$
            \item Run $C^{\Enc(1^n, k,\cdot)} \to (m_0,m_1,\textbf{st})$.
            \item Sample $b\gets \{0,1\}$
            \item Set $\rho_b = \Enc(1^n, k, m_b)$.
            \item Run $C^{\Enc(1^n, k,\cdot)}(\rho_b, \textbf{st}) \to \sigma_{\Reg{A}\Reg{B}}$. In full detail, the oracle $\Enc(1^n,k,\cdot)$ first measures its input register to produce a string $m$ and then outputs $\Enc(1^n,k,m)$.
            \item Run $A(1^n,k,\Reg{A})\to b_A$, $B(1^n,k,\Reg{B})\to b_B$.
            \item Output $1$ if and only if $b=b_A=b_B$.
        \end{enumerate}
        The unclonable encryption scheme must satisfy that for all $n$, QPT $A,B,C$,
        $$\Pr[GAME_{A,B,C}^n(\Enc,\Dec) \to 1] \leq \frac{1}{2}+\negl(n).$$
    \end{enumerate}

    We will often omit the security parameter when it is clear from context.
\end{definition}

\begin{definition}
    Let $\mathcal{U} = \{\mathcal{U}_n\}_{n\in \N}$ be a unitary oracle. An unclonable encryption scheme relative to $\mathcal{U}$ is a pair of quantum oracle algorithms $\Enc^{\vec{U}},\Dec^{\vec{U}}$ with the following syntax
    \begin{enumerate}
        \item $\Enc^{\vec{U}}(1^n,k,m) \to \rho_{k,m}$: takes as input a security parameter $n$, a key $k$, and a message $m$, and outputs a mixed state $\rho_{k,m}$.
        \item $\Dec^{\vec{U}}(1^n, k, \rho) \to m$ takes as input a security parameter $n$, a key $k$, and a quantum ciphertext $\rho$, and outputs a decrypted message $m$.
    \end{enumerate}
    satisfying the following notions of security.
    \begin{enumerate}
        \item (Perfect) correctness: for all $n$, $m\in \{0,1\}^{\ell_m(n)}$,
        $$\BigPr{\Dec^{\vec{U}}(1^n, k, \rho_{k,m}) = m}{U\gets \mathcal{U}_n\\k\gets \{0,1\}^{\ell_k(n)}\\\Enc^{\vec{U}}(1^n, k, m)\to \rho_{k,m}} = 1$$
        \item Security: For oracle algorithms $A^{\Ora{U}},B^{\Ora{U}},C^{\Ora{U},(\cdot)}$, define $GAME_{A,B,C}^n(\Enc,\Dec)$ to be the following game
        \begin{enumerate}
            \item Sample $U\gets \mathcal{U}_n$.
            \item Sample $k \gets \{0,1\}^{\ell_k(n)}$
            \item Run $C^{\vec{U},\Enc^{\vec{U}}(1^n,k,\cdot)}\to (m_0,m_1,\textbf{st})$.
            \item Sample $b\gets \{0,1\}$
            \item Set $\rho_b = \Enc^{\vec{U}}(1^n, k, m_b)$.
            \item Run $C^{\vec{U},\Enc^{\vec{U}}(1^n,k,\cdot)}(\rho_b,\textbf{st})\to \sigma_{\Reg{A}\Reg{B}}$.
            \item Run $A^{\vec{U}}(1^n,k,\Reg{A})\to b_A$, $B^{\vec{U}}(1^n,k,\Reg{B})\to b_B$.
            \item Output $1$ if and only if $b=b_A=b_B$.
        \end{enumerate}
        Security requires that for all polynomials $t$, for all oracle adversaries $A^{(\cdot)},B^{(\cdot)},C^{(\cdot)}$ making at most $t$ oracle queries,
        $$\Pr[GAME_{A,B,C}^n(\Enc,\Dec) \to 1] \leq \frac{1}{2}+\negl(n)$$
    \end{enumerate}
\end{definition}

\begin{definition}\label{def:pure-scheme}
    We will call a candidate unclonable encryption scheme $(\Enc,\Dec)$  \textit{pure} if the encryption algorithm $\Enc$ is of the following form:
    \begin{enumerate}
        \item On input a key $k$ and a message $m$.
        \item Sample $r\gets \{0,1\}^{\ell_r(n)}$.
        \item Output $E_k \ket{m,r,0}$ where $E_k$ is some efficient unitary.
    \end{enumerate}
    We will call $E = \{E_k\}$ the purification of $\Enc$.

    We will call an oracle unclonable encryption scheme $(\Enc^{\vec{U}},\Dec^{\vec{U}})$ pure if the same definition holds for some oracularly efficient unitary $E_k^{\vec{U}}$.
    
    Note that if a (secure) scheme is pure, then without loss of generality we may assume that $\Dec(1^n,k,\rho)$ is the map which applies $E_k^{-1}$ to $\rho$ and then measures the first register in the standard basis.

    For $r\in \{0,1\}^{\ell_r}$, we will denote $\Enc(k,m;r) = E_k \ket{m,r,0}$.
\end{definition}

\todo{These parameters are really bad, it is inherited from the bad Grover's search argument.}

\begin{theorem}\label{thm:unclonable-encryption}
   Let $\ell_r,\ell_k$ be any polynomials in $n$, and let $\V$ be any unitary oracle model. Let $(\wt{\Enc}^{\Ora{V}},\wt{\Dec}^{\Ora{V}})$ be any pure unclonable encryption scheme relative to $\V$ operating on message space $\{0,1\}^{\ell_m}$, key space $\{0,1\}^{\ell_{\wt{k}}(n)}$, randomness space $\{0,1\}^{\ell_{\wt{r}}(n)}$, and using $\ell_a(n)$ ancilla qubits. Let $\U$ be the QHROM of length $\ell_m+\ell_{\wt{r}}+\ell_a+\ell_r$. Then the following encryption scheme is a (pure) unclonable encryption scheme relative to $\U$ with keys in $\{0,1\}^{\ell_k}$.
    \begin{enumerate}
        \item $\Enc^{\Ora{U}}(1^n, k, m)$:
        \begin{enumerate}
            \item Sample $r_1\gets \{0,1\}^{\ell_{\wt{r}}}$, $r_2\gets \{0,1\}^{\ell_r}$
            \item Output $X^k U\ket{m,r_1,0^{\ell_a(n)},r_2}$.
        \end{enumerate}
        \item $\Dec^{\Ora{U}}(1^n, k, \rho)$: computes $\sigma = U^{\dagger} X^k \rho X^k U$, then measures the first register.
    \end{enumerate}
\end{theorem}

\begin{proof}
    Correctness follows by observation.

    To prove security, we will define the following reduction. Let $C^{\Ora{U},\Enc^{\vec{U}}},A^{\Ora{U}},B^{\Ora{U}}$ be any attacker against $\Enc$ making at most $t$ queries to the encryption oracle. We will define an attacker $\wt{C}^{\Ora{V},\Enc^{\Ora{V}}},\wt{A}^{\Ora{V}},\wt{B}^{\Ora{V}}$ against $\wt{\Enc}$.

    In particular, let $\wt{E}^{\Ora{V}}$ be the purification of $\wt{\Enc}^{\Ora{V}}$. Let $\wt{\Ciph} = \{0,1\}^{\ell_m}\times \{0,1\}^{\ell_{\wt{r}}} \times \{0,1\}^{\ell_{\wt{a}}}$ and let $\Ciph = \wt{\Ciph}\times \{0,1\}^{\ell_r}$. Note that $\mathcal{H}(\wt{\Ciph}),\mathcal{H}(\Ciph)$ are the ciphertext spaces of $\Enc$ and $\wt{\Enc}$ respectively. Let $d=\frac{1}{2}\min(\ell_k,\ell_r)$.

    We will define $\wt{C}^{\Ora{V},\wt{\Enc}^{\vec{V}}(\wt{k},\cdot)}$ as follows
    \begin{enumerate}
        \item Sample $R\gets {\{0,1\}^{\ell_r} \choose t}$.\footnote{Recall this means to sample a uniformly random subset of $\{0,1\}^{\ell_r}$ of size $t$.} Write $R = \{r_1,\dots,r_t\}$ in some arbitrary order.
        \item Set $S_1 = \{(x,r):r\in R,x\in \wt{\Ciph}\}$
        \item Sample $S' \gets {\Ciph \choose \frac{|\Ciph|}{2^{d}}-|S_1|}$
        \item Set $S_2 = S_1 \cup S'$
        \item Sample $U_1 \gets \Haar(S_2)$
        \item Sample $U_2\gets \Haar(\Ciph\setminus S_2)$
        \item Sample $k\gets \Ciph$
        \item Define $\mathcal{O}(m)$ as follows
        \begin{enumerate}
            \item On the $i$th query, set $r= r_i$.
            \item Output $X^{k} U_1(\wt{\Enc}^{\vec{V}}(\wt{k},m)\otimes \ket{r})$
        \end{enumerate}
        \item Run $C^{\Ora{U_2},\mathcal{O}}\to (m_0,m_1,\textbf{st})$
        \item Output challenge $m_0,m_1$, receive ciphertext $\ket{\phi}$
        \item Set $\ket{\psi} = X^{k}U_1(\ket{\phi}\otimes \ket{r_t})$
        \item Run $C^{\Ora{U_2},\mathcal{O}}(\ket{\psi},\textbf{st})\to \sigma_{\Reg{A}\Reg{B}}$
        \item Output $\wt{\Reg{A}} = (k,R,U_1,U_2,\Reg{A})$
        \item Output $\wt{\Reg{B}} = (k,R,U_1,U_2,\Reg{B})$
    \end{enumerate}

    We will then define $\wt{A}^{\Ora{V}}(\wt{k},\wt{\Reg{A}})$ as follows
    \begin{enumerate}
        \item Parse $\wt{\Reg{A}}=(k,R,U_1,U_2,\Reg{A})$
        \item Let $c_R\wt{E}_{\wt{k}}^{\Ora{V}}$ be the unitary $\wt{E}_{\wt{k}}^{\Ora{V}}$ controlled on the following register containing a string in $R$. Formally,
        $$c_R \wt{E}_{\wt{k}}^{\Ora{V}} \coloneqq \wt{E}_{\wt{k}}^{\Ora{V}} \otimes \sum_{r\in R}\ketbra{r} + I \otimes \left(I-\sum_{r\in R}\ketbra{r}\right)$$
        \item Set $U_1' = U_1 \cdot c\wt{E}_k^{\Ora{V}}$
        \item Set $U = U_1' U_2$
        \item Run $A^{\vec{U}}(k,\Reg{A}) \to b_A$
    \end{enumerate}

    $\wt{B}^{\Ora{V}}(k,\wt{\Reg{B}})$ will be defined symmetrically.

    The proof of security then follows by applying a sequence of hybrids described in~\Cref{fig:hybrids12,fig:hybrids34,fig:hybrids56}.

\begin{figure}
    \fbox{
    \begin{minipage}{0.48\textwidth}
        $G_1$:\\
        $k \gets \{0,1\}^{\ell_k}$\\
        $U\gets \Haar(\Ciph)$\\
        $C^{\Ora{U},\mathcal{O}}\to (m_0,m_1,\textbf{st})$\\
        $b \gets \{0,1\}$\\
        $\ket{\phi_b} \gets \mathcal{O}(m_b)$\\
        $C^{\Ora{U},\mathcal{O}}(\ket{\phi_b},\textbf{st})\to \sigma_{\Reg{A}\Reg{B}}$\\
        $A^{\Ora{U}}(k,\Reg{A})\to b_A$\\
        $B^{\Ora{U}}(k,\Reg{B})\to b_B$\\
        Output $1$ if and only if $b=b_A=b_B$\\
        \newline
        $\mathcal{O}(m)$:\\
        $r_{orig}\gets \{0,1\}^{\ell_{\wt{r}}}$\\
        $r\gets \{0,1\}^{\ell_r}$\\
        Output $\Enc^{\Ora{U}}(k,m;r_{orig},r)$\\
    \hrule
    \subcaption*{Hybrid $1$: The security game for unclonable encryption}
    \end{minipage}
    }
    \fbox{
    \begin{minipage}{0.48\textwidth}
        $G_2$:\\
        $k \gets \{0,1\}^{\ell_k}$\\
        \color{red}$R:\{r_i\}_{i\in [t]}\gets \{0,1\}^{\ell_r}$\\
        $t=1$\\
        \color{black}
        $U\gets \Haar(\Ciph)$\\
        $C^{\Ora{U},\mathcal{O}}\to (m_0,m_1,\textbf{st})$\\
        $b \gets \{0,1\}$\\
        $\ket{\phi_b} \gets \mathcal{O}(m_b)$\\
        $C^{\Ora{U},\mathcal{O}}(\ket{\phi_b},\textbf{st})\to \sigma_{\Reg{A}\Reg{B}}$\\
        $A^{\Ora{U}}(k,\Reg{A})\to b_A$\\
        $B^{\Ora{U}}(k,\Reg{B})\to b_B$\\
        Output $1$ if and only if $b=b_A=b_B$\\
        \newline
        $\mathcal{O}(m)$:\\
        $r_{orig} \gets \{0,1\}^{\ell_{\wt{r}}}$\\
        \color{red}
        $r \gets r_i$\\
        $t\gets t+1$\\
        \color{black}
        Output $\Enc^{\Ora{U}}(k,m;r_{orig},r)$\\
    \hrule
    \subcaption*{Hybrid $2$: We sample the randomness used for encryption at the beginning of the game.}
    \end{minipage}
    }
\hfill
\caption{The hybrid games used in the proof of \Cref{thm:unclonable-encryption}.}\label{fig:hybrids12}
\end{figure}

\begin{figure}
    \fbox{
    \begin{minipage}{0.48\textwidth}
        $G_3$:\\
        $k \gets \{0,1\}^{\ell_k}$\\
        $R:\{r_i\}_{i\in [t]}\gets \{0,1\}^{\ell_r}$\\
        $t=1$\\
        \color{red}
        $S_1 \gets \{(x,r):r\in R,x\in \wt{\Ciph}\}$\\
        $S' \gets {\Ciph \choose \frac{|\Ciph|}{2^d} - |S_1|}$\\
        $S_2 \gets S_1 \cup S'$\\
        $U_1 \gets \Haar(S_2)$\\
        $U_2 \gets \Haar(\Ciph\setminus S_2)$\\
        $U \gets U_1U_2$\\
        \color{black}
        $C^{\Ora{U},\mathcal{O}}\to (m_0,m_1,\textbf{st})$\\
        $b \gets \{0,1\}$\\
        $\ket{\phi_b} \gets \mathcal{O}(m_b)$\\
        $C^{\Ora{U},\mathcal{O}}(\ket{\phi_b},\textbf{st})\to \sigma_{\Reg{A}\Reg{B}}$\\
        $A^{\Ora{U}}(k,\Reg{A})\to b_A$\\
        $B^{\Ora{U}}(k,\Reg{B})\to b_B$\\
        Output $1$ if and only if $b=b_A=b_B$\\
        \newline
        $\mathcal{O}(m)$:\\
        $r_{orig} \gets \{0,1\}^{\ell_{\wt{r}}}$\\
        $r \gets r_i$\\
        $t\gets t+1$\\
        Output $\Enc^{\Ora{U}}(k,m;r_{orig},r)$\\
    \hrule
    \subcaption*{Hybrid $3$: We split the \Haar\ random unitary $U$ into two random unitaries: $U_1$ acting on a superset of the encryption randomness and $U_2$ acting on the orthogonal space.}
    \end{minipage}
    }
    \fbox{
    \begin{minipage}{0.48\textwidth}
        $G_4$:\\
        $k \gets \{0,1\}^{\ell_k}$\\
        $R:\{r_i\}_{i\in [t]}\gets \{0,1\}^{\ell_r}$\\
        $t=1$\\
        $S_1 \gets \{(x,r):r\in R,x\in \wt{\Ciph}\}$\\
        $S' \gets {\Ciph \choose \frac{|\Ciph|}{2^d} - |S_1|}$\\
        $S_2 \gets S_1 \cup S'$\\
        $U_1 \gets \Haar(S_2)$\\
        $U_2 \gets \Haar(\Ciph\setminus S_2)$\\
        $U \gets U_1U_2$\\
        $C^{{\color{red}\Ora{U}_2},\mathcal{O}}\to (m_0,m_1,\textbf{st})$\\
        $b \gets \{0,1\}$\\
        $\ket{\phi_b} \gets \mathcal{O}(m_b)$\\
        $C^{{\color{red}\Ora{U}_2},\mathcal{O}}(\ket{\phi_b},\textbf{st})\to \sigma_{\Reg{A}\Reg{B}}$\\
        $A^{\Ora{U}}(k,\Reg{A})\to b_A$\\
        $B^{\Ora{U}}(k,\Reg{B})\to b_B$\\
        Output $1$ if and only if $b=b_A=b_B$\\
        \newline
        $\mathcal{O}(m)$:\\
        $r_{orig} \gets \{0,1\}^{\ell_{\wt{r}}}$\\
        $r \gets r_i$\\
        $t\gets t+1$\\
        \color{red}
        Output $X^{k}U_1\ket{m,r_{orig},0^{\ell_a},r}$\\
        \color{black}
    \hrule
    \subcaption*{Hybrid $4$: We rewrite encryption as $X^kU_1 \ket{b,r}$, and we replace $C$'s oracle access to $U$ with oracle access to $U_2$.}
    \end{minipage}
    }
\hfill
\caption{The hybrid games used in the proof of \Cref{thm:unclonable-encryption}.}\label{fig:hybrids34}
\end{figure}

\begin{figure}
    \fbox{
    \begin{minipage}{0.48\textwidth}
        $G_5$:\\
        $k \gets \{0,1\}^{\ell_k}$\\
        $R:\{r_i\}_{i\in [t]}\gets \{0,1\}^{\ell_r}$\\
        $t=1$\\
        $S_1 \gets \{(x,r):r\in R,x\in \wt{\Ciph}\}$\\
        $S' \gets {\Ciph \choose \frac{|\Ciph|}{2^d} - |S_1|}$\\
        $S_2 \gets S_1 \cup S'$\\
        \color{red}
        $\wt{k} \gets \{0,1\}^n$\\
        $V\gets \wt{\V_n}$\\
        $U_1 \gets \Haar(S_2)$\\
        $U_1' \gets U_1' \cdot c_R \wt{E}_{\wt{k}}^{\Ora{V}}$\\
        \color{black}
        $U_2 \gets \Haar(\Ciph\setminus S_2)$\\
        $U \gets U_1'U_2$\\
        $C^{\Ora{U},\mathcal{O}}(\ket{\phi_b},\textbf{st})\to (m_0,m_1,\textbf{st})$\\
        $b \gets \{0,1\}$\\
        $\ket{\phi_b} \gets \mathcal{O}(m_b)$\\
        $C^{{\Ora{U}_2},\mathcal{O}}(\ket{\phi_b})\to \sigma_{\Reg{A}\Reg{B}}$\\
        $A^{\Ora{U}}(k,\Reg{A})\to b_A$\\
        $B^{\Ora{U}}(k,\Reg{B})\to b_B$\\
        Output $1$ if and only if $b=b_A=b_B$\\
        \newline
        $\mathcal{O}(m)$:\\
        $r_{orig} \gets \{0,1\}^{\ell_{\wt{r}}}$\\
        $r \gets r_i$\\
        $t\gets t+1$\\
        Output $X^{k}{\color{red}U_1'}\ket{m,r_{orig},0^{\ell_a},r}$\\
    \hrule
    \subcaption*{Hybrid $5$: We sample $V$, and then set $U_1$ to first apply $\wt{E}_k^{\Ora{V}}$.}
    \end{minipage}
    }
    \fbox{
    \begin{minipage}{0.48\textwidth}
        $G_6$:\\
        $V\gets \V_n$\\
        $\wt{C}^{\Ora{V},\wt{\Enc}^{\Ora{V}}(\wt{k},\cdot)}\to (m_0,m_1,\textbf{st})$
        $b \gets \{0,1\}$\\
        $\ket{\phi_b} \gets \mathcal{O}(m_b)$\\
        $\wt{C}^{\Ora{V},\wt{\Enc}^{\Ora{V}}(\wt{k},\cdot)}(\ket{\phi_b},\textbf{st}) \to \wt{\sigma}_{\Reg{A}\Reg{B}}$\\
        $\wt{A}^{\wt{V}}(k,\wt{\Reg{A}})\to b_A$\\
        $\wt{B}^{\wt{V}}(k,\wt{\Reg{B}})\to b_B$\\
        Output $1$ if and only if $b=b_A=b_B$\\
    \hrule
    \subcaption*{Hybrid $6$: The security game for the secure scheme under $\wt{\V}$.}
    \end{minipage}
    }
\hfill
\caption{The hybrid games used in the proof of \Cref{thm:unclonable-encryption}.}\label{fig:hybrids56}
\end{figure}
    
    In particular, we will show that for each $i$, $|\Pr[G_i \to 1] - \Pr[G_{i-1}\to 1]|\leq \negl(n)$.

    \begin{lemma}
        $G_1\equiv G_2$
    \end{lemma}

    \begin{proof}
        The two games are equivalent other than the fact that in $G_2$, the values $r_1,\dots,r_t$ are sampled in the beginning instead of when they are used. But the time when they are sampled does not affect the overall distribution of the game, and so the claim follows.
    \end{proof}

    \begin{lemma}
        $|\Pr[G_3 \to 1] - \Pr[G_2 \to 1]|\leq \negl(n)$
    \end{lemma}

    \begin{proof}
        This follows immediately from our ``unitary reprogramming lemma"~\Cref{lem:reprogramming} by setting $S_1$ to be $S_1$ from $G_3$, $N=|\Ciph|$ $M_2 = \frac{|\Ciph|}{2^{d}}$, $M_2 = |S_1| = t\cdot |\wt{\Ciph}|$. It is clear that 
        $$\frac{M_1}{M_2} = \frac{|\wt{\Ciph}|}{\frac{|\Ciph|}{2^{d}}} = \frac{\frac{|\Ciph|}{2^{\ell_r}}}{\frac{|\Ciph|}{2^{d}}} = 2^{d-\ell_r}\leq 2^{-\ell_r/2}=\negl(n)$$
        and
        $$\frac{M_2}{N} = \frac{|\Ciph|\cdot 2^{-d}}{|\Ciph|} = 2^{-d}=\negl(n)$$
        Furthermore, the only information used in the game about $U_1,U_2$ is the values in $R$, which is less information than contained in $S_1$. And so the lemma applies.
    \end{proof}

    \begin{lemma}\label{lem:hardhyb}
        $|\Pr[G_4\to 1] - \Pr[G_3 \to 1]|\leq \negl(n)$
    \end{lemma}

    This claim is significantly more complicated, and so the proof is deferred to \ifnum\llncs=0~\Cref{sec:hardhyb}\else Appendix A of the full version\fi. The intuition is that $\wt{C}$ cannot query $U$ on any state with more than negligible support on $S_2$. In particular, the only information about $S_2$ which is revealed to $\wt{C}$ are ciphertext outputs, which are of the form $X^k \ket{\psi}$ for some $\ket{\psi}$ in the support of $S_2$. Since only one sample of these is provided, these look like strings of the form $k \oplus s$ for some $s\in S_2$, and given many samples of strings of this form it should be difficult to find a single string in $S_2$. Some care needs to be taken to ensure that these facts hold with unitary queries to $U$, and the key idea to handle this is the fact that unitary queries to $U_1$ can be simulated efficiently using a reflection oracle around states with support in $S_2$. We then reduce to a Grover's search bound on the ability to find $s\in S_2$ given all strings of the form $s\xor k$ as well as a reflection oracle around $S_2$.

    \begin{lemma}
        $G_5\equiv G_4$
    \end{lemma}

    \begin{proof}
        This follows immediately from unitary invariance. In particular, $c_R\wt{E}_{\wt{k}}^{\Ora{V}}$ is the identity on all basis states not ending with a string in $R$, and so in particular is identity on all states with support in $\Ciph\setminus S_1 \supseteq \Ciph\setminus S_2$, just as $U_1$ is.
    \end{proof}

    \begin{lemma}
        $G_6 \equiv G_5$
    \end{lemma}

    \begin{proof}
        This follows from the definitions of $\wt{A},\wt{B},\wt{C}$.
    \end{proof}
\end{proof}
\section{Unclonable encryption in the QROM}\label{sec:qrom}

Here, we observe that (reusable) unclonable encryption in the quantum random oracle model readily follows from prior work~\cite{AKLLZ22,AKL23}. 

\begin{theorem}\label{thm:in-qrom}
    For any message length $\ell_m = \poly(n)$, there exists a pure many-time unclonable encryption scheme in the QROM.
\end{theorem}

\begin{proof}
    Prior work \cite{BL20,AKLLZ22,AKL23} has established that the following is a one-time unclonable encryption scheme, where $F$ is a random oracle, $k \gets \{0,1\}^n$, and $m \in \{0,1\}^{\ell_m}$.

    \begin{itemize}
        \item $\Enc^F(1^n,k,m):$
        \begin{itemize}
            \item Sample $x \gets \{0,1\}^n$
            \item Output $E_k\ket{m}\ket{x} = \ket{m \oplus F(x)}H^k\ket{x}$
        \end{itemize}
    \end{itemize}

    To complete the proof, we demonstrate a compiler that takes any pure one-time unclonable encryption scheme in the QROM and outputs a pure many-time unclonable encryption scheme in the QROM. Let $(\Enc^F,\Dec^F)$ be a one-time scheme, with encryption unitary $E_k$ and randomness space $\mathcal{R}$. The many-time scheme is the following.

    \begin{itemize}
        \item $\widetilde{\Enc}^F(1^n,k,m):$
        \begin{itemize}
            \item Sample $s \gets \{0,1\}^n$.
            \item Sample $r \gets \mathcal{R}$.
            \item Output $\widetilde{E}_k\ket{m}\ket{r}\ket{s} = (E_{F(k,s)}\ket{m}\ket{r})\ket{s}$.
        \end{itemize}
        \item $\widetilde{\Dec}^F(1^n,k,\rho):$ Parse $\rho$ as $\rho' \otimes \ketbra{s}{s}$ and output $\Dec^F(1^n,F(k,s),\rho')$.
    \end{itemize}

    The many-time security of $(\widetilde{\Enc}^F,\widetilde{\Dec}^F)$ reduces in a standard way to the one-time security of $(\Enc^F,\Dec^F)$ by re-sampling $F(k,\cdot)$ at the point $s$ sampled when generating the challenge ciphertext.
    
\end{proof}

The main result of this paper follows as a direct corollary of \Cref{thm:in-qrom} and \Cref{thm:unclonable-encryption}.

\begin{corollary}
    Let $n$ be the security parameter. For any message length $\ell_m = \poly(n)$, there exists reusable unclonable encryption with $n$-bit keys in the Haar random oracle model (QHROM).
\end{corollary}

\bibliographystyle{alpha}

\bibliography{main}

@InProceedings{BL20,
  author =	{Broadbent, Anne and Lord, S\'{e}bastien},
  title =	{{Uncloneable Quantum Encryption via Oracles}},
  booktitle =	{15th Conference on the Theory of Quantum Computation, Communication and Cryptography (TQC 2020)},
  pages =	{4:1--4:22},
  series =	{Leibniz International Proceedings in Informatics (LIPIcs)},
  ISBN =	{978-3-95977-146-7},
  ISSN =	{1868-8969},
  year =	{2020},
  volume =	{158},
  editor =	{Flammia, Steven T.},
  publisher =	{Schloss Dagstuhl -- Leibniz-Zentrum f{\"u}r Informatik},
  address =	{Dagstuhl, Germany},
  URL =		{https://drops.dagstuhl.de/entities/document/10.4230/LIPIcs.TQC.2020.4},
  URN =		{urn:nbn:de:0030-drops-120639},
  doi =		{10.4230/LIPIcs.TQC.2020.4}
}

@InProceedings{AKLLZ22,
author="Ananth, Prabhanjan
and Kaleoglu, Fatih
and Li, Xingjian
and Liu, Qipeng
and Zhandry, Mark",
editor="Dodis, Yevgeniy
and Shrimpton, Thomas",
title="On the Feasibility of Unclonable Encryption, and More",
booktitle="Advances in Cryptology -- CRYPTO 2022",
year="2022",
publisher="Springer Nature Switzerland",
address="Cham",
pages="212--241",
isbn="978-3-031-15979-4"
}

@inproceedings{AKL23,
author = {Ananth, Prabhanjan and Kaleoglu, Fatih and Liu, Qipeng},
title = {Cloning Games: A General Framework for Unclonable Primitives},
year = {2023},
isbn = {978-3-031-38553-7},
publisher = {Springer-Verlag},
address = {Berlin, Heidelberg},
url = {https://doi.org/10.1007/978-3-031-38554-4_3},
doi = {10.1007/978-3-031-38554-4_3},
booktitle = {Advances in Cryptology – CRYPTO 2023: 43rd Annual International Cryptology Conference, CRYPTO 2023, Santa Barbara, CA, USA, August 20–24, 2023, Proceedings, Part V},
pages = {66–98},
numpages = {33},
location = {Santa Barbara, CA, USA}
}

@inproceedings{MH25,
author = {Ma, Fermi and Huang, Hsin-Yuan},
title = {How to Construct Random Unitaries},
year = {2025},
isbn = {9798400715105},
publisher = {Association for Computing Machinery},
address = {New York, NY, USA},
url = {https://doi.org/10.1145/3717823.3718254},
doi = {10.1145/3717823.3718254},
booktitle = {Proceedings of the 57th Annual ACM Symposium on Theory of Computing},
pages = {806–809},
numpages = {4},
location = {Prague, Czechia},
series = {STOC '25}
}

@InProceedings{AGKL24,
author="Ananth, Prabhanjan
and Gulati, Aditya
and Kaleoglu, Fatih
and Lin, Yao-Ting",
editor="Joye, Marc
and Leander, Gregor",
title="Pseudorandom Isometries",
booktitle="Advances in Cryptology -- EUROCRYPT 2024",
year="2024",
publisher="Springer Nature Switzerland",
address="Cham",
pages="226--254",
isbn="978-3-031-58737-5"
}

@MISC{stackexchange,
    TITLE = {Is the diamond norm subadditive under composition?},
    AUTHOR = {Niel de Beaudrap (https://quantumcomputing.stackexchange.com/users/124/niel-de-beaudrap)},
    HOWPUBLISHED = {Quantum Computing Stack Exchange},
    NOTE = {URL:https://quantumcomputing.stackexchange.com/questions/7043/is-the-diamond-norm-subadditive-under-composition (version: August 18, 2019)},
    URL = {https://quantumcomputing.stackexchange.com/questions/7043/is-the-diamond-norm-subadditive-under-composition}
}

@article{kre21,
    author = "Kretschmer, William",
    title = "{Quantum Pseudorandomness and Classical Complexity}",
    eprint = "2103.09320",
    archivePrefix = "arXiv",
    primaryClass = "quant-ph",
    doi = "10.4230/LIPIcs.TQC.2021.2",
    journal = "Leibniz Int. Proc. Inf.",
    volume = "197",
    pages = "2:1--2:20",
    year = "2021"
}

@article{Botteron:2024orj,
    author = "Botteron, Pierre and Broadbent, Anne and Culf, Eric and Nechita, Ion and Pellegrini, Cl{\'e}ment and Rochette, Denis",
    title = "{Towards Unconditional Uncloneable Encryption}",
    eprint = "2410.23064",
    archivePrefix = "arXiv",
    primaryClass = "quant-ph",
    month = "10",
    year = "2024"
}

@article{bhattacharyya2026uncloneable,
  author       = {Archishna Bhattacharyya and Eric Culf},
  title        = {Uncloneable Encryption from Decoupling},
  journal      = {Nature Physics},
  year         = {2026},
  note         = {published online 4 Feb 2026},
  doi          = {10.1038/s41567-025-03154-7},
  url          = {https://doi.org/10.1038/s41567-025-03154-7},
}

@InProceedings{10.1007/978-3-031-15802-5_8,
author="Ananth, Prabhanjan
and Qian, Luowen
and Yuen, Henry",
editor="Dodis, Yevgeniy
and Shrimpton, Thomas",
title="Cryptography from Pseudorandom Quantum States",
booktitle="Advances in Cryptology -- CRYPTO 2022",
year="2022",
publisher="Springer Nature Switzerland",
address="Cham",
pages="208--236",
isbn="978-3-031-15802-5"
}

@InProceedings{10.1007/978-3-319-96878-0_5,
author="Ji, Zhengfeng
and Liu, Yi-Kai
and Song, Fang",
editor="Shacham, Hovav
and Boldyreva, Alexandra",
title="Pseudorandom Quantum States",
booktitle="Advances in Cryptology -- CRYPTO 2018",
year="2018",
publisher="Springer International Publishing",
address="Cham",
pages="126--152",
isbn="978-3-319-96878-0"
}

@misc{SMLBH25,
      title={Strong random unitaries and fast scrambling}, 
      author={Thomas Schuster and Fermi Ma and Alex Lombardi and Fernando Brandao and Hsin-Yuan Huang},
      year={2025},
      eprint={2509.26310},
      archivePrefix={arXiv},
      primaryClass={quant-ph},
      url={https://arxiv.org/abs/2509.26310}, 
}

@InProceedings{coladangelo_et_al:LIPIcs.ITCS.2026.41,
  author =	{Coladangelo, Andrea and Liu, Qipeng and Xie, Ziyi},
  title =	{{The Curious Case of "XOR Repetition" of Monogamy-Of-Entanglement Games}},
  booktitle =	{17th Innovations in Theoretical Computer Science Conference (ITCS 2026)},
  pages =	{41:1--41:20},
  series =	{Leibniz International Proceedings in Informatics (LIPIcs)},
  ISBN =	{978-3-95977-410-9},
  ISSN =	{1868-8969},
  year =	{2026},
  volume =	{362},
  editor =	{Saraf, Shubhangi},
  publisher =	{Schloss Dagstuhl -- Leibniz-Zentrum f{\"u}r Informatik},
  address =	{Dagstuhl, Germany},
  URL =		{https://drops.dagstuhl.de/entities/document/10.4230/LIPIcs.ITCS.2026.41},
  URN =		{urn:nbn:de:0030-drops-253281},
  doi =		{10.4230/LIPIcs.ITCS.2026.41}
}

@inproceedings{10.1145/3564246.3585225,
author = {Kretschmer, William and Qian, Luowen and Sinha, Makrand and Tal, Avishay},
title = {Quantum Cryptography in Algorithmica},
year = {2023},
isbn = {9781450399135},
publisher = {Association for Computing Machinery},
address = {New York, NY, USA},
url = {https://doi.org/10.1145/3564246.3585225},
doi = {10.1145/3564246.3585225},
booktitle = {Proceedings of the 55th Annual ACM Symposium on Theory of Computing},
pages = {1589–1602},
numpages = {14},
location = {Orlando, FL, USA},
series = {STOC 2023}
}

@inproceedings{10.1145/3717823.3718144,
author = {Kretschmer, William and Qian, Luowen and Tal, Avishay},
title = {Quantum-Computable One-Way Functions without One-Way Functions},
year = {2025},
isbn = {9798400715105},
publisher = {Association for Computing Machinery},
address = {New York, NY, USA},
url = {https://doi.org/10.1145/3717823.3718144},
doi = {10.1145/3717823.3718144},
booktitle = {Proceedings of the 57th Annual ACM Symposium on Theory of Computing},
pages = {189–200},
numpages = {12},
location = {Prague, Czechia},
series = {STOC '25}
}

@InProceedings{poremba_et_al:LIPIcs.ITCS.2026.109,
  author =	{Poremba, Alexander and Ragavan, Seyoon and Vaikuntanathan, Vinod},
  title =	{{Cloning Games, Black Holes and Cryptography}},
  booktitle =	{17th Innovations in Theoretical Computer Science Conference (ITCS 2026)},
  pages =	{109:1--109:21},
  series =	{Leibniz International Proceedings in Informatics (LIPIcs)},
  ISBN =	{978-3-95977-410-9},
  ISSN =	{1868-8969},
  year =	{2026},
  volume =	{362},
  editor =	{Saraf, Shubhangi},
  publisher =	{Schloss Dagstuhl -- Leibniz-Zentrum f{\"u}r Informatik},
  address =	{Dagstuhl, Germany},
  URL =		{https://drops.dagstuhl.de/entities/document/10.4230/LIPIcs.ITCS.2026.109},
  URN =		{urn:nbn:de:0030-drops-253961},
  doi =		{10.4230/LIPIcs.ITCS.2026.109}
}

@article{Dohotaru2008ExactQL,
  title={Exact quantum lower bound for grover's problem},
  author={Catalin Dohotaru and Peter H{\o}yer},
  journal={Quantum Inf. Comput.},
  year={2008},
  volume={9},
  pages={533-540},
  url={https://api.semanticscholar.org/CorpusID:10757294}
}

@article{Hoeffding,
 ISSN = {01621459, 1537274X},
 URL = {http://www.jstor.org/stable/2282952},
 author = {Wassily Hoeffding},
 journal = {Journal of the American Statistical Association},
 number = {301},
 pages = {13--30},
 publisher = {[American Statistical Association, Taylor & Francis, Ltd.]},
 title = {Probability Inequalities for Sums of Bounded Random Variables},
 urldate = {2026-03-05},
 volume = {58},
 year = {1963}
}

@misc{bhattacharyya2026uncloneablebitexists,
      title={The uncloneable bit exists}, 
      author={Archishna Bhattacharyya and Anne Broadbent and Eric Culf},
      year={2026},
      eprint={2603.08916},
      archivePrefix={arXiv},
      primaryClass={quant-ph},
      url={https://arxiv.org/abs/2603.08916}, 
}

@inproceedings{C:ABGL25,
author = {Ananth, Prabhanjan and Bostanci, John and Gulati, Aditya and Lin, Yao-Ting},
title = {Pseudorandom Unitaries in the Haar Random Oracle Model},
year = {2025},
isbn = {978-3-032-01877-9},
publisher = {Springer-Verlag},
address = {Berlin, Heidelberg},
url = {https://doi.org/10.1007/978-3-032-01878-6_10},
doi = {10.1007/978-3-032-01878-6_10},
booktitle = {Advances in Cryptology – CRYPTO 2025: 45th Annual International Cryptology Conference, Santa Barbara, CA, USA, August 17–21, 2025, Proceedings, Part II},
pages = {301–333},
numpages = {33},
location = {Santa Barbara, CA, USA}
}

@InProceedings{TCC:HhaYam25,
author="Hhan, Minki
and Yamada, Shogo",
editor="Applebaum, Benny
and Lin, Huijia (Rachel)",
title="Pseudorandom Function-Like States from Common Haar Unitary",
booktitle="Theory of Cryptography",
year="2026",
publisher="Springer Nature Switzerland",
address="Cham",
pages="134--165",
isbn="978-3-032-12296-4"
}

@InProceedings{EC:ABGL25,
author="Ananth, Prabhanjan
and Bostanci, John
and Gulati, Aditya
and Lin, Yao-Ting",
editor="Fehr, Serge
and Fouque, Pierre-Alain",
title="Pseudorandomness in the (Inverseless) Haar Random Oracle Model",
booktitle="Advances in Cryptology -- EUROCRYPT 2025",
year="2025",
publisher="Springer Nature Switzerland",
address="Cham",
pages="138--166",
isbn="978-3-031-91098-2"
}

@article{BBBV,
author = {Bennett, Charles H. and Bernstein, Ethan and Brassard, Gilles and Vazirani, Umesh},
title = {Strengths and Weaknesses of Quantum Computing},
journal = {SIAM Journal on Computing},
volume = {26},
number = {5},
pages = {1510-1523},
year = {1997},
doi = {10.1137/S0097539796300933},

URL = { 
    
        https://doi.org/10.1137/S0097539796300933
    
    

},
eprint = { 
    
        https://doi.org/10.1137/S0097539796300933
    
    

}
}

@misc{CGR26,
      author = {Alper {\c{C}}akan and Vipul Goyal and Justin Raizes},
      title = {How to Delete Without a Trace: Certified Deniability in a Quantum World},
      howpublished = {Cryptology {ePrint} Archive, Paper 2024/1832},
      year = {2024},
      url = {https://eprint.iacr.org/2024/1832}
}

\ifnum\llncs=0
\newpage 
\appendix
\section{Proof of~\Cref{lem:hardhyb}}\label{sec:hardhyb}

The goal of this section will be to prove the following.
\begin{lemma}[\Cref{lem:hardhyb} restated]
    $|\Pr[G_3\to 1] - \Pr[G_4\to 1]|\leq \negl(n)$, where $G_2,G_3$ are as described in~\Cref{fig:hybrids12,fig:hybrids34}.
\end{lemma}

To do this, it suffices to show that $C$ will never make a query to its first oracle that has weight on the support of $S_2$. In particular, the claim will immediately follow from the following lemma

\begin{lemma}\label{lem:hardhybtrick}
    Let $C^{(\cdot)}$ be any oracle algorithm making at most $t \leq \poly(n)$ queries. Let $\ell_m,\allowbreak\ell_k,\allowbreak\ell_{\wt{k}},\allowbreak\ell_r,\allowbreak\ell_{\wt{r}},\allowbreak\ell_a,\allowbreak\Ciph,\allowbreak\wt{\Ciph}$ be as in~\Cref{thm:unclonable-encryption} and its proof. Consider the following game $FindSEnc(C)$
    \begin{enumerate}
        \item $k \gets \{0,1\}^{\ell_k}$
        \item $R = \{r_i\}_{i\in [t]} \gets \{0,1\}^{\ell_r}$
        \item $S_1 \gets \{(x,r):r\in R,x\in \wt{\Ciph}\}$
        \item $S' \gets {\Ciph \choose \frac{|\Ciph|}{2^{d}} - |S_1|}$
        \item $S_2 \gets S_1 \cup S'$
        \item $U_1 \gets \Haar(S_2)$
        \item $U_2 \gets \Haar(\Ciph\setminus S_2)$
        \item $U \gets U_1U_2$
        \item $C^{\Ora{U},\mathcal{O}} \to s$
        \item $C$ wins if $s \in S_2$
    \end{enumerate}
    where $\mathcal{O}(m)$ is the oracle which on the $i$th query responds with
    \begin{enumerate}
        \item Sample $r_{orig}\gets \{0,1\}^{\ell_{\wt{r}}}$
        \item Output $X^k U_1 \ket{m,r_{orig},0^{\ell_a},r_i}$
    \end{enumerate}
    Then 
    $$\Pr[FindSEnc(C)\to 1]\leq \negl(n)$$
\end{lemma}

We will prove the lemma by reducing to the information-theoretic lower bound on unstructured search.

\begin{theorem}[Grover's search lower bound~\cite{Dohotaru2008ExactQL}]\label{thm:grover}
    Let $N\in \N$ and for $k\in [N]$, let $\delta_k \coloneqq R_{\{k\}}$. Let $\A^{(\cdot)}$ be any QPT oracle algorithm making at most $T$ queries such that
    $$\Pr_{k\gets [N]}\left[\A^{\delta_k} \to k\right] \geq p$$
    Then
    $$T \geq \frac{\sqrt{N}}{2}\left(1+\sqrt{p}-\sqrt{1-p}-\frac{2}{\sqrt{N}}\right)$$
\end{theorem}

We will rely on the following useful concentration bound.

\begin{lemma}[Sampling Chernoff bound (Section 5 of ~\cite{Hoeffding})]\label{lem:chernoff}
    Let $N,M\in \N$, $\delta>0$, and let $S\subseteq [N]$ be any subset. Then
    $$\Pr_{S'\gets {[N]\choose M}}\left[|S\cap S'| \geq (1+\delta)\frac{M\cdot |S|}{N}\right] \leq e^{-\frac{\delta^2\cdot M\cdot |S|}{3 N}}$$
\end{lemma}

Our search problem is of a slightly different flavor than standard unstructured search. In particular, to prove~\Cref{lem:hardhybtrick}, we will reduce to the information-theoretic hardness of the following task: for a random set $S$, given access to a reflection oracle for $S$ as well as a classical description of the set $\{s \xor k:s\in S\}$ for a random $k$, the goal is to output a single value $s\in S$.

\begin{lemma}\label{lem:grovermod}
    Let $\ell_n(n),\ell_k(n),M(n)$ be functions. Define $N=2^{\ell_n}$. Let $C^{(\cdot)}$ be any oracle algorithm making at most $t$ queries. Given a set $S\subseteq \mathcal{P}(\{0,1\}^{\ell_n})$ and a value $k \in \{0,1\}^{\ell_k}$, define $S \xor k = \{s \xor (k||0^{\ell_n-\ell_k}) : s \in S\}$. Consider the following game $FindSXOR(C)$
    \begin{enumerate}
        \item Sample $S\gets {[N] \choose M}$
        \item Sample $k\gets \{0,1\}^{\ell_k}$
        \item Run $C^{R_S}(S\xor k) \to s$
        \item Output $1$ if $s\in S$
    \end{enumerate}
    Then $\Pr[C\text{ wins}] \leq \frac{9 t^2 M^3}{2^{3\ell_n-2\ell_k}} - 2^{\ell_n-\ell_k}\cdot e^{-\frac{4\cdot M}{3\cdot 2^{\ell_n-\ell_k}}}$
\end{lemma}

\begin{proof}
    We begin by defining the following notation. For $x\in \{0,1\}^{\ell_n}$, define $x_{\leq \ell_k}$ to be the first $\ell_k$ bits of $x$, and define $x_{>\ell_k}$ to be the last $\ell_n-\ell_k$ bits of $x$. For $k\in \{0,1\}^{\ell_k}$, we will for simplicity write $x\xor k\coloneqq x\xor (k||0^{\ell_n-\ell_k})$.

    For $S\subseteq \{0,1\}^{\ell_n},v\in \{0,1\}^{\ell_n-\ell_k}$, define $S_v = \{x\in S:x_{>\ell_k} = v\}$.

    We then reduce to the lower bound for Grover's search~\Cref{thm:grover}. In particular, for a random $k \in \{0,1\}^{\ell_n}$, given quantum query access to $\delta_k \coloneqq R_{\{k\}}$, we give an algorithm $\A^{\delta_k}$ to find $k$
    \begin{enumerate}
        \item Sample $S' \gets {[N] \choose M}$
        \item If there exists some $v\in \{0,1\}^{\ell_n-\ell_k}$ such that $|S_v|\geq \frac{2M}{2^{\ell_n-\ell_k}}$, output $\bot$.
        \item Define $\mathcal{O}\ket{s}$ to be the coherent version of the following process
        \begin{enumerate}
            \item For each $s'\in S'$ such that $s'_{>\ell_k} = s_{>\ell_k}$, run $\delta_k(s\xor s')$. If the result is $1$, output $1$.
            \item If the result is $0$ for all such $s'$, output $0$.
        \end{enumerate}
        Formally,
        \begin{equation}
        \begin{split}
            \mathcal{O} \coloneqq I-2\sum_{\substack{s' \in S'\\s'_{>\ell_k} = s_{>\ell_k}}} \ket{s'\oplus k}\bra{s'\oplus k}\\
            = \prod_{\substack{s' \in S'\\s'_{>\ell_k} = s_{>\ell_k}}} X^{s'}\delta_{k} X^{s'}
        \end{split}
        \end{equation}
        \item Run $C^{\mathcal{O}}(S') \to s$
        \item Pick a random $s'\in S'$, output $s-s'$
    \end{enumerate}

    Clearly, this algorithm makes at most $t\cdot \frac{M}{2^{\ell_n-\ell_k}}$ queries.

    Note that for each $v\in \{0,1\}^{\ell_n-\ell_k}$, by~\Cref{lem:chernoff} we have that 
    $$\Pr\left[|S_v| \geq \frac{2M}{2^{\ell_n-\ell_k}}\right] \leq e^{-\frac{4\cdot M}{3\cdot 2^{\ell_n-\ell_k}}}$$
    Define 
    $$\gamma \coloneqq 2^{\ell_n-\ell_k}\cdot e^{-\frac{4\cdot M}{3\cdot 2^{\ell_n-\ell_k}}}$$
    by union bound, the algorithm passes step 2 with probability
    $$\geq 1-\gamma$$

    Thus, it only remains to lower bound its success probability.

    Define $S \coloneqq S'\xor k = \{s' \xor k:s'\in S'\}$. Note that $S$ is also a uniformly random set, and from the perspective of $C$, $S'$ is identically distributed to $S\xor k$. It is not hard to see that $\mathcal{O}$ exactly implements $R_S$.

    If $C$ succeeds with probability $\epsilon$ in its own game, then by union bound, with probability $\epsilon-\gamma$, in step 4 $\A^{\delta_k}$ produces some $s\in S=S'\xor k$. But note that if $s\in S$, then there exists exactly one point $s'\in S'$ such that $s\xor s' = k$. And so $\A^{\delta_k}$ succeeds with probability $\frac{1}{M}\cdot \left(\epsilon - \gamma\right)$.

    And so we have an algorithm for unstructured search succeeding with probability $\frac{\epsilon-\gamma}{M}$ and making at most $\frac{t\cdot M}{2^{\ell_n-\ell_k}}$ queries.

    Set $p\coloneqq \frac{\epsilon-\gamma}{M}$

    By~\Cref{thm:grover}, we have
    \begin{equation}
    \begin{split}
        \frac{t\cdot M}{2^{\ell_n-\ell_k}} &\geq \frac{\sqrt{N}}{2}\left(1+\sqrt{p} - \sqrt{1-p}-\frac{2}{\sqrt{N}}\right)\\
        &= \frac{\sqrt{N}}{2}\left((1-\sqrt{1-p})+\sqrt{p}\right) - 1\\
        &\geq \frac{\sqrt{Np}}{2}-1\\
        &\geq \frac{\sqrt{\frac{N(\epsilon-\gamma)}{M}}}{3}
    \end{split}
    \end{equation}
    And so
    $$\epsilon \leq \frac{9 t^2 M^3}{2^{3\ell_n-2\ell_k}}-\gamma$$
\end{proof}

\begin{corollary}\label{cor:grovermod}
    Let $\ell_n(n),\ell_k(n)$ be any polynomials such that $\ell_k<\ell_n$. Let $\ell_M(n)$ be any polynomial satisfying $\ell_M \geq \ell_n - \frac{\ell_k}{2}$ and $\ell_n-\ell_M\geq n^c$ for some $c>0$. Define $N=2^{\ell_n}$, $M(n)=2^{\ell_M}$. Let $C^{(\cdot)}$ be any oracle algorithm making at most $t=\poly(n)$ queries. Given a set $S\subseteq \mathcal{P}(\{0,1\}^{\ell_n})$ and a value $k \in \{0,1\}^{\ell_k}$, define $S \xor k = \{s \xor (k||0^{\ell_n-\ell_k}) : s \in S\}$. Consider the game $FindSXOR(C)$ as in ~\Cref{lem:grovermod}.
    Then $\Pr[C\text{ wins}] \leq \negl(n)$
\end{corollary}

\begin{proof}
    Note that $2(\ell_n-\ell_M)\leq \ell_k$, and so since increasing the length of the key only decreases $C$'s advantage, we can assume without loss of generality that $\ell_k = 2(\ell_n-\ell_M)$.

    Observe that $\ell_n-\ell_k = \ell_n-2\ell_n+2\ell_M=2\ell_M-\ell_n$.

    Plugging in~\Cref{lem:grovermod}, we get
    \begin{equation}
        \begin{split}
            \Pr[C\text{ wins}] &\leq \frac{9 t^2 M^3}{2^{3\ell_n-2\ell_k}} - 2^{\ell_n-\ell_k}\cdot e^{-\frac{4\cdot M}{3\cdot 2^{\ell_n-\ell_k}}}\\
            &\leq \frac{9t^22^{3\ell_M}}{2^{\ell_n+2\ell_M}}-\exp\left(\ln(2)(2\ell_M-\ell_n)-\frac{4\cdot 2^{\ell_M}}{3\cdot 2^{2\ell_M-\ell_n}}\right)\\
            &=9t^2\cdot 2^{\ell_M-\ell_n}-\exp\left(\ln(2)(2\ell_M-\ell_n)-\frac{4}{3}\cdot 2^{\ell_n-\ell_M}\right)\\
            &\leq \poly(n)\cdot  2^{-n^c} - \exp\left(\poly(n)-\frac{4}{3}\cdot 2^{n^c}\right)\\
            &\leq \negl(n)
        \end{split}
    \end{equation}
\end{proof}

Finally, before we begin the proof of~\Cref{lem:hardhybtrick}, we will cite an important fact about random unitaries.

\begin{lemma}[Modified Lemma 4.1 from~\cite{AGKL24}]\label{lem:separatesample}
    Let $S$ be a set. Let $n,s,t\in \N$ and $\vec{x}= (x_1,\dots,x_s)\in S^s$ such that $\vec{x}$ has no repeating coordinates. Let
    $$\rho \coloneqq \E_{U\gets \Haar(S)}\left[\bigotimes_{j=1}^s\left(U\ketbra{0}U^\dagger\right)^{\otimes t}\right]$$
    and let
    $$\sigma \coloneqq \bigotimes_{j=1}^t\E_{\ket{\phi_j}\gets \Haar(S)}\left[\ketbra{\phi_j}^{\otimes t}\right]$$
    Then
    $$TD(\rho,\sigma)\leq \frac{s^2t}{|S|}$$
\end{lemma}

\begin{remark}
    Lemma 4.1 from~\cite{AGKL24} defined the set $S$ as explicitly $\{0,1\}^n$, but the same proof applies.
\end{remark}

We will now proceed to the proof of~\Cref{lem:hardhybtrick}.

\begin{proof}
    Given an adversary $C$ for $FindSEnc$, we will come up with an adversary $D$ for $FindSXOR$ (with $N = |\Ciph|$ and $M=\frac{|\Ciph|}{2^d} - t$) such that
    $$|\Pr[FindSEnc(C)\to 1] - \Pr[FindSXOR(D)\to 1]|\leq \negl(n)$$
    \Cref{lem:hardhybtrick} will then follow immediately from~\Cref{lem:grovermod}.

    Let $\Ora{\Sim}^{R_S}$ be the simulator from~\Cref{cor:effpr} with $t_{max}=t$. $D^{R_S}(S\xor k)$ will simulate $C^{\Ora{U},\mathcal{O}}$ as follows
    \begin{enumerate}
        \item $D^{R_S}$ will replace all queries to $\Ora{U}$ with queries to $\Ora{\Sim}^{R_S}$
        \item $D^{R_S}(S\xor k)$ will replace all queries to $\mathcal{O}$ with the following oracle, which we will name $\Sim_{\Enc}(m)$:
        \begin{enumerate}
            \item Sample $r \gets S\xor k$
            \item Output $r$
        \end{enumerate}
        \item At the end, $D^{R_S}(S\xor k)$ will output the same value as $C^{\Ora{\Sim}^{R_S}}$
    \end{enumerate}

    To show 
    $$|\Pr[FindSEnc(C)\to 1] - \Pr[FindSXOR(D)\to 1]|\leq \negl(n)$$
    we will proceed through a sequence of hybrids $H_0,\dots,H_f$ with $H_0 = FindSEnc(C)$ and $H_f = FindSXOR(D)$. The hybrids are defined in~\Cref{fig:hardlemmahybs,fig:hardlemmahybs2,fig:hardlemmahybs3}.

    \begin{claim}
        $|\Pr[H_1\to 1] - \Pr[H_0\to 1]|\leq \negl(n)$
    \end{claim}

    \begin{proof}
        Note that conditioned on $R$ containing no duplicates in $H_0$, $H_0$ and $H_1$ are identically distributed. But the probability that $R$ contains a duplicate is 
        $$\leq \frac{t^2}{2^{\ell_r}} = \negl(n)$$
        by the union bound.
    \end{proof}

    \begin{claim}
        $|\Pr[H_2\to 1]-\Pr[H_1\to 1]|\leq \negl(n)$
    \end{claim}

    \begin{proof}
        This follows immediately from~\Cref{lem:separatesample} setting $S = S_2$ and $s=t$ since
        $$\frac{t^2}{|S_2|} \leq \frac{\poly(n)}{\frac{|\Ciph|}{2^d} - t\cdot |\wt{\Ciph}|} = \negl(n)$$
    \end{proof}

    \begin{claim}
        $H_3 \equiv H_2$
    \end{claim}

    \begin{proof}
        Observe that in $H_2$, each $(m,r_{orig},0^{\ell_a},r_i)$ is unique since each $r_i$ is distinct. Thus, $H_3$ is the same as $H_2$, but we simply sample each $\ket{\phi_s}$ later. This makes no difference to the games output probability.
    \end{proof}

    \begin{claim}
        $H_4 \equiv H_3$
    \end{claim}

    \begin{proof}
        Note that the mixed state corresponding to $\ket{\phi}\gets \HaarSt(S_2)$ is exactly the maximally mixed state. This is the same density matrix one would obtain from sampling a random string from $S_2$. And so since only one copy of $\ket{\phi}$ is ever used, replacing it with a random string has no effect on the hybrid.
    \end{proof}

    \begin{claim}
        $|\Pr[H_5\to 1]-\Pr[H_4\to 1]|\leq \negl(n)$
    \end{claim}

    \begin{proof}
        This follows immediately from~\Cref{cor:effpr}.
    \end{proof}

    \begin{claim}
        $|\Pr[H_6\to 1]-\Pr[H_5\to 1]|\leq \negl(n)$
    \end{claim}

    \begin{proof}
        It is sufficient to bound the probability that given access to $R_{S_2},\mathcal{O}$ as in $H_5$, you can output something in $S_1$.
    
        Note that the probability that in $H_5$, $\mathcal{O}$ ever outputs something in $S_1$ is bounded by 
        $$t\cdot \frac{|S_1|}{|S_2|} = \poly(n)\cdot \frac{2^d\cdot |\wt{\Ciph}|}{|\Ciph|} \leq 2^{d-\ell_r}\leq \negl(n)$$

        Thus, the problem reduces to unstructured search. We want to bound the probability that given polynomial queries to $R_{S_2}$. This is clearly bounded by the probability that given polynomial queries to $R_{S_1}$, you can find something in $S_1$. But this is negligible by optimality of Grover's.
    \end{proof}

    \begin{claim}
        $H_f\equiv H_6$
    \end{claim}

    \begin{proof}
        These are identical games by observation.
    \end{proof}

    We then conclude the proof of the lemma by applying apply~\Cref{cor:grovermod}. In particular, we set $S=S$, $\ell_n = \ell_m + \ell_a + \ell_{\wt{r}} + \ell_r = |\Ciph|$, $\ell_k=\ell_k$, $\ell_M = \ell_n - d$. We then observe that $d=\frac{1}{2}\min(\ell_k,\ell_r)$, $\ell_M \geq \ell_n - \frac{\ell_k}{2}$ and $\ell_n-\ell_M=d\geq n^c$ for some $c>0$, and so~\Cref{cor:grovermod} gives us $\Pr[H_f\to 1]\leq \negl(n)$.

    The hybrids then imply that $\Pr[H_1\to 1]\leq \negl(n)$, and so we are done.
\end{proof}

\begin{figure}
\fbox{
    \begin{minipage}{0.48\textwidth}
        $H_0 = FindSEnc(C)$:\\
        $k\gets \{0,1\}^{\ell_n}$\\
        $R = \{r_i\}_{i\in [t]}\gets \{0,1\}^{\ell_r}$\\
        $i \gets 1$\\
        $S_1 \gets \{(x,r):r\in R,x\in \wt{\Ciph}\}$\\
        $S' \gets {\Ciph \choose \frac{|\Ciph|}{2^d} - |S_1|}$\\
        $S_2 \gets S_1 \cup S'$\\
        $U_1 \gets \Haar(S_2)$\\
        $U_2 \gets \Haar(\{0,1\}^n\setminus S_2)$\\
        $C^{\Ora{U}_2,\mathcal{O}}\to s$
        Output $1$ if and only if $s\in S$\\
        \newline
        $\mathcal{O}(m)$:\\
        $r_{orig} \gets \{0,1\}^{\ell_{\wt{r}}}$\\
        $r \gets r_i$\\
        $i \gets i+1$\\
        Output $X^k U_1 \ket{m,r_{orig},0,r_i}$\\
    \hrule
    \subcaption*{Hybrid $0$: The game $FindSEnc(C)$.}
    \end{minipage}
    }
    \fbox{
    \begin{minipage}{0.48\textwidth}
        $H_1$:\\
        $k\gets \{0,1\}^{\ell_n}$\\
        {\color{red}$R \gets {\{0,1\}^{\ell_r} \choose t}$}\\
        $i \gets 1$\\
        $S_1 \gets \{(x,r):r\in R,x\in \wt{\Ciph}\}$\\
        {\color{red}$S' \gets {\Ciph \choose \frac{|\Ciph|}{2^d} - t\cdot |\wt{\Ciph}|}$}\\
        $S_2 \gets S_1 \cup S'$\\
        $U_1 \gets \Haar(S_2)$\\
        $U_2 \gets \Haar(\{0,1\}^n\setminus S_2)$\\
        $C^{\Ora{U}_2,\mathcal{O}}\to s$
        Output $1$ if and only if $s\in S$\\
        \newline
        $\mathcal{O}(m)$:\\
        $r_{orig} \gets \{0,1\}^{\ell_{\wt{r}}}$\\
        $r \gets r_i$\\
        $i \gets i+1$\\
        Output $X^k U_1 \ket{m,r_{orig},0,r_i}$\\
    \hrule
    \subcaption*{Hybrid $1$: Here we guarantee that there are no collisions in $R$.}
    \end{minipage}
    }
    
\hfill
\caption{The hybrid games used in the proof of~\Cref{lem:hardhybtrick}.}\label{fig:hardlemmahybs}
\end{figure}

\begin{figure}
\fbox{
    \begin{minipage}{0.48\textwidth}
        $H_2$:\\
        $k\gets \{0,1\}^{\ell_n}$\\
        {$R \gets {\{0,1\}^{\ell_r} \choose t}$}\\
        $i \gets 1$\\
        $S_1 \gets \{(x,r):r\in R,x\in \wt{\Ciph}\}$\\
        {$S' \gets {\Ciph \choose \frac{|\Ciph|}{2^d} - t\cdot |\wt{\Ciph}|}$}\\
        $S_2 \gets S_1 \cup S'$\\
        {\color{red}$\{\ket{\phi_s}\}_{s \in S_1} \gets \HaarSt(S)$}\\
        $U_2 \gets \Haar(\{0,1\}^n\setminus S_2)$\\
        $C^{\Ora{U}_2,\mathcal{O}}\to s$
        Output $1$ if and only if $s\in S$\\
        \newline
        $\mathcal{O}(m)$:\\
        $r_{orig} \gets \{0,1\}^{\ell_{\wt{r}}}$\\
        $r \gets r_i$\\
        $i \gets i+1$\\
        {\color{red}Output $X^k \ket{\phi_{m,r_{orig},0^{\ell_a},r}}$}\\
    \hrule
    \subcaption*{Hybrid $2$: Instead of sampling $U_1 \gets \Haar(S_2)$, we sample each $\ket{\phi_s} = U_1 \ket{s}$ from $\HaarSt(S_2)$.}
    \end{minipage}
    }
    \fbox{
    \begin{minipage}{0.48\textwidth}
        $H_3$:\\
        $k\gets \{0,1\}^{\ell_n}$\\
        {$R \gets {\{0,1\}^{\ell_r} \choose t}$}\\
        $i \gets 1$\\
        $S_1 \gets \{(x,r):r\in R,x\in \wt{\Ciph}\}$\\
        {$S' \gets {\Ciph \choose \frac{|\Ciph|}{2^d} - t\cdot |\wt{\Ciph}|}$}\\
        $S_2 \gets S_1 \cup S'$\\
        $U_2 \gets \Haar(\{0,1\}^n\setminus S_2)$\\
        $C^{\Ora{U}_2,\mathcal{O}}\to s$
        Output $1$ if and only if $s\in S$\\
        \newline
        $\mathcal{O}(m)$:\\
        {\color{red}$\ket{\phi}\gets \HaarSt(S_2)$\\
        Output $X^k \ket{\phi}$}\\
    \hrule
    \subcaption*{Hybrid $3$: Instead of sampling $\ket{\phi_{s}}$ at the beginning, we sample it live when it is needed.}
    \end{minipage}
    }
    
\hfill
\caption{The hybrid games used in the proof of~\Cref{lem:hardhybtrick}.}\label{fig:hardlemmahybs2}
\end{figure}

\begin{figure}
\fbox{
    \begin{minipage}{0.48\textwidth}
        $H_4$:\\
        $k\gets \{0,1\}^{\ell_n}$\\
        {$R \gets {\{0,1\}^{\ell_r} \choose t}$}\\
        $i \gets 1$\\
        $S_1 \gets \{(x,r):r\in R,x\in \wt{\Ciph}\}$\\
        {$S' \gets {\Ciph \choose \frac{|\Ciph|}{2^d} - t\cdot |\wt{\Ciph}|}$}\\
        $S_2 \gets S_1 \cup S'$\\
        $U_1 \gets \Haar(S_2)$\\
        $U_2 \gets \Haar(\{0,1\}^n\setminus S_2)$\\
        $C^{\Ora{U}_2,\mathcal{O}}\to s$
        Output $1$ if and only if $s\in S$\\
        \newline
        $\mathcal{O}(m)$:\\
        {\color{red}$s \gets S_2$\\
        Output $s \xor k$\\}
    \hrule
    \subcaption*{Hybrid $4$: Instead of sampling $\ket{\phi_s}$ as a Haar random state, we sample it as a random element of $S$.}
    \end{minipage}
    }
    \fbox{
    \begin{minipage}{0.48\textwidth}
        $H_5$:\\
        $k\gets \{0,1\}^{\ell_n}$\\
        {$R \gets {\{0,1\}^{\ell_r} \choose t}$}\\
        $i \gets 1$\\
        $S_1 \gets \{(x,r):r\in R,x\in \wt{\Ciph}\}$\\
        {$S' \gets {\Ciph \choose \frac{|\Ciph|}{2^d} - t\cdot |\wt{\Ciph}|}$}\\
        $S_2 \gets S_1 \cup S'$\\
        $C^{{\color{red}\Ora{\Sim}^{R_{S_2}}(t_{max})},\mathcal{O}}\to s$
        Output $1$ if and only if $s\in S$\\
        \newline
        $\mathcal{O}(m)$:\\
        {$s \gets S_2$\\
        Output $s \xor k$\\}
    \hrule
    \subcaption*{Hybrid $5$: We replace $\vec{U}_2$ with $\Ora{\Sim}^{R_{S_2}}$.}
    \end{minipage}
    }
    
\hfill
\caption{The hybrid games used in the proof of~\Cref{lem:hardhybtrick}.}\label{fig:hardlemmahybs3}
\end{figure}

\begin{figure}
\fbox{
    \begin{minipage}{0.48\textwidth}
        $H_6$:\\
        $k\gets \{0,1\}^{\ell_n}$\\
        {$R \gets {\{0,1\}^{\ell_r} \choose t}$}\\
        $i \gets 1$\\
        {\color{red}$S \gets {\Ciph \choose \frac{|\Ciph|}{2^d} - t\cdot |\wt{\Ciph}|}$}\\
        $C^{{\Ora{\Sim}^{R_{\color{red}S}}(t_{max})},\mathcal{O}}\to s$
        Output $1$ if and only if $s\in S$\\
        \newline
        $\mathcal{O}(m)$:\\
        {$s \gets {\color{red}S}$\\
        Output $s \xor k$\\}
    \hrule
    \subcaption*{Hybrid $6$: We no longer include $S_1$ inside of $S_2$.}
    \end{minipage}
    }
    \fbox{
    \begin{minipage}{0.48\textwidth}
        $H_f = FindSXOR(D)$:\\
        {$S \gets {\Ciph \choose \frac{|\Ciph|}{2^d} - t\cdot |\wt{\Ciph}|}$}\\
        $D^{R_S}(S\xor k)\to s$
        Output $1$ if and only if $s\in S$\\
    \hrule
    \subcaption*{The final hybrid: $FindSXOR$.}
    \end{minipage}
    }
    
\hfill
\caption{The hybrid games used in the proof of~\Cref{lem:hardhybtrick}.}\label{fig:hardlemmahybs4}
\end{figure}

\section{Proof of hybrids for~\Cref{lem:reprogramming}}\label{sec:rephyb}

We provide the formal proofs that $$\norm{\A^{\Ora{V_i}} - \A^{\Ora{V_{i-1}}}}_2 \leq \negl(n)$$ for all $i$.

\begin{lemma}
    $$\norm{\A^{\Ora{V_1}}-\A^{\Ora{V_0}}}_2 \leq \negl(n)$$
\end{lemma}

\begin{proof}
    While this seems like it should immediately follow from ~\Cref{thm:pathrecordingreverse}, in fact there is a slight mismatch in terms of how the operators are defined. 
    
    In particular, we can write $\mathcal{H}([N])\otimes \mathcal{H}(R_{[N]})^{\otimes 4}\otimes \mathcal{H}(\mathcal{P}([N])) = A_1 \oplus A_2$ for some subspaces $A_1,A_2$ defined as follows. $A_1$ will be spanned by states of the form $\ket{D}\ket{S_2}$ where $D_x\in S_2$, and $A_2$ will be spanned by states of the form $\ket{D}\ket{S_2}$ where $D_x\notin S_2$.

    For any $X$, we can write $V_1^X = V_{1,1}^X + V_{1,2}^X$ where $V_{1,1}^X$ is a partial isometry on $A_1$ and $V_{1,2}^X$ is a partial isometry on $A_2$. 

    Then, the proper path-recording oracle for $U_1$ is
    $$PR_1 = (V_{1,1}^{L_1}\cdot (I_{A_1}-V_{1,1}^{R_1}\cdot V_{1,1}^{R_1,\dagger}) + (I_{A_1}-V_{1,1}^{L_1}\cdot V_{1,1}^{L_1,\dagger})V_{1,1}^{R_1}) + I_{A_2}$$
    $$\ol{PR}_1 = (\ol{V}_1^{L_1}\cdot (I_{A_1}-\ol{V}_1^{R_1}\cdot \ol{V}_1^{R_1,\dagger}) + (I_{A_1}-\ol{V}_1^{L_1}\cdot \ol{V}_1^{L_1,\dagger})\ol{V}_1^{R_1}) + I_{A_2}$$

    Similarly, the proper path-recording oracle for $U_2$ is
    $$PR_2 = I_{A_1}+ (V_{1,2}^{L_2}\cdot (I_{A_2}-V_{1,2}^{R_2}\cdot V_{1,2}^{R_2,\dagger}) + (I_{A_2}-V_{1,2}^{L_2}\cdot V_{1,2}^{L_2,\dagger})V_{1,2}^{R_2})$$
    $$\ol{PR}_2 = I_{A_1}+ (\ol{V}_1^{L_2}\cdot (I_{A_2}-\ol{V}_1^{R_2}\cdot \ol{V}_1^{R_2,\dagger}) + (I_{A_2}-\ol{V}_1^{L_2}\cdot \ol{V}_1^{L_2,\dagger})\ol{V}_1^{R_2})$$

    ~\Cref{thm:pathrecordingreverse} tells us that 
    $$\norm{\A^{(PR_1PR_2,PR_2^\dagger PR_1^\dagger,\ol{PR}_1\ol{PR}_2,\ol{PR}_1^\dagger \ol{PR}_2^\dagger)}-\A^{\Ora{V_0}}}_2^2\leq \negl(n)$$

    We will show that in fact $V_1=PR_1PR_2$ and so~\Cref{thm:pathrecordingreverse} applies directly.

    We make the following simple observations
    \begin{enumerate}
        \item $V_1^{L_1} = V_{1,1}^{L_1} + I_{A_2}$
        \item $V_1^{R_1} = V_{1,1}^{R_1} + I_{A_2}$
        \item $V_1^{L_2} = I_{A_1} + V_{1,1}^{R_1}$
        \item $V_1^{R_2} = I_{A_1} + V_{1,1}^{R_2}$
    \end{enumerate}
    We can then compute
    \begin{equation}
        \begin{split}
            &V_1^{L_1}V_1^{L_2}(I-V_1^{R_1}V_1^{R_2}V_1^{R_2,\dagger}V_1^{R_1,\dagger})\\
            &=(V_{1,1}^{L_1} + V_{1,2}^{L_2})(I_{A_1}+I_{A_2} - (V_{1,1}^{R_1} + V_{1,2}^{R_2})(V_{1,1}^{R_1,\dagger} + V_{1,2}^{R_2,\dagger}))\\
            &=V_{1,1}^{L_1}(I_{A_1}-V_{1,1}^{R_1}\cdot V_{1,1}^{R_1,\dagger})+V_{1,2}^{L_2}(I_{A_2}-V_{1,2}^{R_2}\cdot V_{1,2}^{R_2,\dagger})
        \end{split}
    \end{equation}
    and so
    \begin{equation}
        \begin{split}
            V_1&\\
            =& V_1^{L_1}V_1^{L_2}(I-V_1^{R_1}V_1^{R_2}V_1^{R_2,\dagger}V_1^{R_1,\dagger})+(I-V_1^{L_1}V_1^{L_2}V_1^{L_2,\dagger}V_1^{L_1,\dagger})V_1^{R_2,\dagger}V_1^{R_1,\dagger}\\
            =&V_{1,1}^{L_1}(I_{A_1}-V_{1,1}^{R_1}\cdot V_{1,1}^{R_1,\dagger})+V_{1,2}^{L_2}(I_{A_2}-V_{1,2}^{R_2}\cdot V_{1,2}^{R_2,\dagger})\\
            &+(I_{A_1}-V_{1,1}^{L_1}\cdot V_{1,1}^{L_1,\dagger})V_{1,1}^{R_1,\dagger} + (I_{A_2}-V_{1,2}^{L_2}\cdot V_{1,2}^{L_2,\dagger})V_{1,2}^{R_2,\dagger}\\
            =&V_{1,1}^{L_1}(I_{A_1}-V_{1,1}^{R_1}\cdot V_{1,1}^{R_1,\dagger})+(I_{A_1}-V_{1,1}^{L_1}\cdot V_{1,1}^{L_1,\dagger})V_{1,1}^{R_1,\dagger}\\
            &+V_{1,2}^{L_2}(I_{A_2}-V_{1,2}^{R_2}\cdot V_{1,2}^{R_2,\dagger}) + (I_{A_2}-V_{1,2}^{L_2}\cdot V_{1,2}^{L_2,\dagger})V_{1,2}^{R_2,\dagger}\\
            =&(V_{1,1}^{L_1}(I_{A_1}-V_{1,1}^{R_1}\cdot V_{1,1}^{R_1,\dagger})+(I_{A_1}-V_{1,1}^{L_1}\cdot V_{1,1}^{L_1,\dagger})V_{1,1}^{R_1,\dagger}+I_{A_2})\\
            &\cdot(I_{A_1}+V_{1,2}^{L_2}(I_{A_2}-V_{1,2}^{R_2}\cdot V_{1,2}^{R_2,\dagger}) + (I_{A_2}-V_{1,2}^{L_2}\cdot V_{1,2}^{L_2,\dagger})V_{1,2}^{R_2,\dagger})\\
            =&PR_1\cdot PR_2
        \end{split}
    \end{equation}
\end{proof}

\begin{lemma}\label{claim:hyb2}
    $$\norm{\A^{\Ora{V_2}}-\A^{\Ora{V_1}}}_2 \leq \negl(n)$$
\end{lemma}

\begin{proof}
    We will show that $\norm{V_2^{L_1} - V_1^{L_1}}_{op} \leq \negl(n)$. The same will hold for all other components of $V_2,\ol{V_2}$ by a symmetric argument. And so by~\Cref{prop:diamop,lem:opprod,prop:opinv} we get that $\norm{V_2 - V_1}_{\diamond} \leq \negl(n)$ (and similarly for the other components of $\Ora{V_2}$). Applying~\Cref{thm:querydiam} gives the full lemma.

    To begin, let $D\in \mathcal{DB}$. Let $A,B \subseteq [N]\setminus D_{1,all}$ be arbitrary subsets. Then
    \begin{equation}
    \begin{split}
        \left(\frac{1}{\sqrt{|A|}}\sum_{y\in A}\bra{D_y^{L_1,f}}\right)\left(\frac{1}{\sqrt{|B|}}\sum_{y\in B}\ket{D_y^{L_1,f}}\right)\\
        =\frac{1}{\sqrt{|A|\cdot|B|}} \sum_{y \in A\cap B} \braket{D_y^{L_1,f}}{D_y^{L_1,f}}\\
        =\frac{|A\cap B|}{\sqrt{|A|\cdot |B|}}\\
        \geq \frac{|A\cap B|}{\sqrt{|A\cap B|\cdot |A\cup B|}}\\
        = \sqrt{\frac{|A\cap B|}{|A\cup B|}}
    \end{split}
    \end{equation}

    For all $D,S_2$ such that $D_x \notin S_2$, it is clear that 
    $$\bra{D,S_2} V_2^{L_1,\dagger} V_1^{L_1}\ket{D,S_2} = \braket{D,S_2}{D,S_2} = 1$$

    For all $D,S_2$ such that $D_x\in S_2$, we have 
    \begin{equation}
        \begin{split}
            &\bra{D,S_2} V_2^{L_1,\dagger} V_1^{L_1}\ket{D,S_2}\\
            &=\left(\frac{1}{\sqrt{|S_2\setminus D_{1,all}|}} \sum_{y \in S_2\setminus D_{1,all}} \bra{D_y^{L_1,f},S_2}\right)\left(\frac{1}{\sqrt{|S_2\setminus \Ima(D_1)|}} \sum_{y \in S_2\setminus \Ima(D_1)} \ket{D_y^{L_1,f},S_2}\right)\\
            &\geq \sqrt{\frac{|(S_2\setminus D_{1,all})\cap (S_2\setminus \Ima(D_1))|}{|(S_2\setminus D_{1,all})\cup (S_2\setminus \Ima(D_1))|}}\\
            &\geq \sqrt{\frac{{|S_2|-2t}}{|S_2|}}\\
            &=\sqrt{\frac{M_2-2t}{M_2}}
        \end{split}
    \end{equation}

    But note that by~\Cref{lem:outputortho} we have that for all $(D,S_2)\neq (D',S_2')$, $V_2^{L_1} \ket{D,S_2}$ and $V_1^{L_1} \ket{D',S_2'}$ are orthogonal. And so for all $$\ket{\phi} = \sum_{D\in \mathcal{DB},S_2\in \mathcal{P}([N])} \alpha_{D,S_2}\ket{D,S_2},$$
    we have
    \begin{equation}
        \begin{split}
            &\bra{\phi} V_2^{L_1,\dagger} V_1^{L_1} \ket{\phi}\\
            &=\sum_{D\in \mathcal{DB},S_2\in \mathcal{P}([N])} |\alpha_{D,S_2}|^2 \bra{D,S_2} V_2^{L_2,\dagger} V_1^{L_1} \ket{D,S_2}\\
            &\geq \sum_{D\in \mathcal{DB},S_2\in \mathcal{P}([N])} |\alpha_{D,S_2}|^2 \sqrt{\frac{M_2-2t}{M_2}}\\
            &=\sqrt{\frac{M_2-2t}{M_2}}\\
            \geq 1-\negl(n)
        \end{split}
    \end{equation}
    and so by~\Cref{prop:opbound} we are done.
\end{proof}

Before proving the next hybrid, we will make a few crucial observations.
\begin{lemma}\label{claim:validswapin}
    Let $D\in \mathcal{DB}$ such that $D_x\in D_{1,all}$. Then for all $y\in [N]\setminus D_{1,all}$,
    $$\Valid(D) \cap \{S_2:y\in S_2\} = \Valid(D_y^{L_1,f}) = \Valid(D_y^{R_1,f}) = \Valid(D_y^{L_1,b}) = \Valid(D_y^{R_1,b})$$
\end{lemma}

\begin{proof}
    We will only show $\Valid(D) \cap \{S_2:y\in S_2\} = \Valid(D_y^{L_1,f})$, and the rest will follow by symmetric arguments.

    $\subseteq$: Let $S_2\in \Valid(D)$ and $y\in S_2$. We show
    \begin{enumerate}
        \item $(D_y^{L_1,f})_1 \subseteq S_2$: Note that $(D_y^{L_1,f})_1 = D_{1,all}\cup \{(D_x,y)\}$. We know $D_x\in D_{1,all}\subseteq S_2$ and $y\in S_2$, so $(D_y^{L_1,f})_1 \subseteq S_2$
        \item $(D_y^{L_1,f})_2\cap S_2 = \emptyset$: $(D_y^{L_1,f})_2 = D_{2,all}$ and $D_{2,all}\cap S_2=\emptyset$
        \item If $(D_y^{L_1,f})_x \notin (D_y^{L_1,f})_1$, then $(D_y^{L_1,f})_x \notin S_2$: $(D_y^{L_1,f})_x = y\in (D_y^{L_1,f})_1$, so this holds vacuously.
    \end{enumerate}
    and so $S_2\in \Valid(D_y^{L_1,f})$.

    $\supseteq$: Let $S_2 \in \Valid(D_y^{L_1,f})$. We show
    \begin{enumerate}
        \item $D_{1,all} \subseteq S_2$: Follows from $D_{1,all}\subseteq (D_y^{L_1,f})_1\subseteq S_2$
        \item $D_{2,all}\cap S_2 = \emptyset$: Follows from $D_{2,all} \subseteq (D_y^{L_1,f})_2$ and $(D_y^{L_1,f})_2 \cap S_2 = \emptyset$
        \item If $D_x\notin D_{1,all}$, then $D_x\notin S_2$: We know $D_x\in (D_y^{L_1,f})_2$ so this holds vacuously.
        \item $y\in S_2$: $y \in D_{1,all} \cup \{(D_x,y)\} = (D_y^{L_1,f})_1 \subseteq S_2$.
    \end{enumerate}
    and so $S_2\in \Valid(D)\cap \{S_2:y\in S_2\}$.
\end{proof}

\begin{lemma}\label{claim:validswapout}
    Let $D\in \mathcal{DB}$ such that $D_x\notin D_{1,all}$. Then for all $y\in [N]\setminus D_{all}$,
    $$\Valid(D) \cap \{S_2:y\notin S_2\} = \Valid(D_y^{L_2,f}) = \Valid(D_y^{R_2,f}) = \Valid(D_y^{L_2,b}) = \Valid(D_y^{R_2,b})$$
\end{lemma}

\begin{proof}
    We will only show $\Valid(D) \cap \{S_2:y\notin S_2\} = \Valid(D_y^{L_2,f})$, and the rest will follow by symmetric arguments.

    $\subseteq$: Let $S_2\in \Valid(D)$ and $y\notin S_2$. We show
    \begin{enumerate}
        \item $(D_y^{L_2,f})_1 \subseteq S_2$: This follows from $(D_y^{L_2,f})_1 = D_{1,all}\subseteq S_2$.
        \item $(D_y^{L_2,f})_2\cap S_2 = \emptyset$: We know $(D_y^{L_1,f})_2 = D_{2,all} \cup \{(D_x,y)\}$. $y\notin S_2$ by assumption. Since $S_2$ is valid and $D_x\notin D_{1,all}$, $D_x\notin S_2$ and $D_{2,all}\cap S_2=\emptyset$.
        \item If $(D_y^{L_2,f})_x \notin (D_y^{L_2,f})_1$, then $(D_y^{L_2,f})_x \notin S_2$: $(D_y^{L_2,f})_x = y \notin S_2$.
    \end{enumerate}
    and so $S_2\in \Valid(D_y^{L_2,f})$.

    $\supseteq$: Let $S_2 \in \Valid(D_y^{L_2,f})$. We show
    \begin{enumerate}
        \item $D_{1,all} \subseteq S_2$: Follows from $D_{1,all}= (D_y^{L_2,f})_1\subseteq S_2$
        \item $D_{2,all}\cap S_2 = \emptyset$: Follows from $D_{2,all} \subseteq (D_y^{L_1,f})_2$ and $(D_y^{L_1,f})_2 \cap S_2 = \emptyset$.
        \item If $D_x\notin D_{1,all}$, then $D_x\notin S_2$: We know $D_x\in D_{1,all}$ so this holds vacuously.
        \item $y\notin S_2$: $y \in D_{2,all} \cup \{(D_x,y)\} = (D_y^{L_1,f})_2$. But $(D_y^{L_1,f})_2\cap S_2=\emptyset$, so $y\notin S_2$.
    \end{enumerate}
    and so $S_2\in \Valid(D)\cap \{S_2:y\in S_2\}$.
\end{proof}

\begin{lemma}\label{claim:painful}
    $$\norm{\A^{\Ora{V_3}}-\A^{\Ora{V_2}}}_2 \leq \negl(n)$$
\end{lemma}

The proof of this lemma is involved and is thus deferred to~\Cref{sec:painful}. In particular, the error in this reduction stems from the fact that for a state in $\Pi_{good}$, $V_2^{L_1,\dagger}$ may map it to a state slightly outside of $\Pi_{good}$. We then bound this overlap by a combinatorial argument given in~\Cref{sec:combinatorics}.

\begin{lemma}
    $V_4=V_3$ and $\ol{V}_4=\ol{V}_3$.
\end{lemma}

\begin{proof}
    This follows from directly reasoning about the partial isometries. As an example, we will show $V_4^{L_1} = V_3^{L_1}$ and $V_4^{L_2} = V_3^{L_2}$, and the other cases will be analogous.

    To show this, it is sufficient to show that they behave the same on $\ket{ValSt_D}$ for all databases $D$, since $\Span(\Pi_{good})=\Ima(ConS)$.

    Consider the case $D_x \in D_{1,all}$. Then
    \begin{equation}
        \begin{split}
            &V_4^{L_1} \frac{1}{\sqrt{|\Valid(D)|}}\ket{D}\sum_{S_2 \in \Valid(D)} \ket{S_2}\\
            &=ConS \cdot W_4^{L_1}\cdot  ConS^\dagger \frac{1}{\sqrt{|\Valid(D)|}}\ket{D}\sum_{S_2 \in \Valid(D)} \ket{S_2}\\
            &=ConS\cdot W_4^{L_1} \ket{D}\\
            &=ConS \frac{1}{\sqrt{|[N]\setminus D_{all}|}} \sum_{y\in [N]\setminus D_{all}} \ket{D_y^{L_1,f}} \\
            &=\frac{1}{\sqrt{|[N]\setminus D_{all}|}} \sum_{y\in [N]\setminus D_{all}}\frac{1}{\sqrt{|\Valid(D_y^{L_1,f})|}}\ket{D_y^{L_1,f}} \sum_{S_2\in \Valid(D_y^{L_1,f})} \ket{S_2}
        \end{split}
    \end{equation}
    We can then compute the normalization factor by
    \begin{equation}
        \begin{split}
            &|\Valid(D_y^{L_1,f})|\cdot |[N]\setminus D_{all}|\\
            &={N-|D_{all}|-1\choose M_2-|D_{1,all}|-1}\cdot (N-|D_{all}|)\\
            &=\frac{(N-|D_{all}|-1)!(N-|D_{all}|)}{(M_2-|D_{1,all}|-1)!(N-|D_{all}|-M_2+|D_{1,all}|)!}\\
            &=\frac{(N-|D_{all}|)!(M_2-|D_{1,all}|)}{(M_2-|D_{1,all}|)!(N-|D_{all}|-M_2+|D_{1,all}|)!}\\
            &={N-|D_{all}|\choose M_2-|D_{1,all}|}\cdot (M_2-|D_{1,all}|)
        \end{split}
    \end{equation}
    and so
    \begin{equation}
        \begin{split}
            &V_4^{L_1} \frac{1}{\sqrt{|\Valid(D)|}}\ket{D}\sum_{S_2 \in \Valid(D)} \ket{S_2}\\
            &=\frac{1}{\sqrt{{N-|D_{all}|\choose M_2-|D_{1,all}|}(M_2-|D_{1,all}|)}} \sum_{y\in [N]\setminus D_{all}}\ket{D_y^{L_1,f}} \sum_{S_2\in \Valid(D_y^{L_1,f})} \ket{S_2}
        \end{split}
    \end{equation}
    This let's us compute
    \begin{equation}
        \begin{split}
            &V_3^{L_1}\frac{1}{\sqrt{|\Valid(D)|}}\ket{D}\sum_{S_2 \in \Valid(D)} \ket{S_2}\\
            &=\frac{1}{\sqrt{\Valid(D)}}  \sum_{S_2\in \Valid(D)}\frac{1}{\sqrt{|S_2\setminus D_{1,all}|}}\sum_{y\in S_2\setminus D_{1,all}} \ket{D_y^{L_1,f}}\ket{S_2}\\
            &=\frac{1}{\sqrt{{N-|D_{all}|\choose M_2-|D_{1,all}|}}}\frac{1}{\sqrt{M_2-|D_{1,all}|}} \sum_{y\in [N]\setminus D_{all}}  \ket{D_y^{L_1,f}}\sum_{\substack{S_2 \in \Valid(D)\\y\in S_2}}\ket{S_2}\\
            &=V_4^{L_1}\frac{1}{\sqrt{|\Valid(D)|}}\ket{D}\sum_{S_2 \in \Valid(D)} \ket{S_2}
        \end{split}
    \end{equation}

    On the other side, when $D_x\in D_{1,all}$
    \begin{equation}
        \begin{split}
            &V_4^{L_2} \frac{1}{\sqrt{|\Valid(D)|}}\ket{D}\sum_{S_2 \in \Valid(D)} \ket{S_2}\\
            &=ConS \cdot W_4^{L_2}\cdot  ConS^\dagger \frac{1}{\sqrt{|\Valid(D)|}}\ket{D}\sum_{S_2 \in \Valid(D)} \ket{S_2}\\
            &=ConS\cdot W_4^{L_2} \ket{D}\\
            &=ConS \ket{D}\\
            &=\frac{1}{\sqrt{|\Valid(D)|}}\ket{D}\sum_{S_2 \in \Valid(D)} \ket{S_2}\\
            &=V_3^{L_2}\frac{1}{\sqrt{|\Valid(D)|}}\ket{D}\sum_{S_2 \in \Valid(D)} \ket{S_2}
        \end{split}
    \end{equation}

    For $D_x\notin D_{1,all}$, we similarly have
    \begin{equation}
        \begin{split}
            &V_4^{L_1} \frac{1}{\sqrt{|\Valid(D)|}}\ket{D}\sum_{S_2 \in \Valid(D)} \ket{S_2}\\
            &=ConS \cdot W_4^{L_1}\cdot  ConS^\dagger \frac{1}{\sqrt{|\Valid(D)|}}\ket{D}\sum_{S_2 \in \Valid(D)} \ket{S_2}\\
            &=ConS\cdot W_4^{L_1} \ket{D}\\
            &=ConS \ket{D}\\
            &=\frac{1}{\sqrt{|\Valid(D)|}}\ket{D}\sum_{S_2 \in \Valid(D)} \ket{S_2}\\
            &=V_3^{L_1}\frac{1}{\sqrt{|\Valid(D)|}}\ket{D}\sum_{S_2 \in \Valid(D)} \ket{S_2}
        \end{split}
    \end{equation}

    Thus, it remains to show 
    $$V_3^{L_2}\frac{1}{\sqrt{|\Valid(D)|}}\ket{D}\sum_{S_2 \in \Valid(D)} \ket{S_2}= V_4^{L_2}\frac{1}{\sqrt{|\Valid(D)|}}\ket{D}\sum_{S_2 \in \Valid(D)} \ket{S_2}$$
    We compute
    \begin{equation}
        \begin{split}
            &V_4^{L_2}\frac{1}{\sqrt{|\Valid(D)|}}\ket{D}\sum_{S_2 \in \Valid(D)} \ket{S_2}\\
            &=ConS\cdot W_4^{L_2} \ket{D}\\
            &=ConS \frac{1}{\sqrt{[N]\setminus D_{all}}} \sum_{y\in [N]\setminus D_{all}} \ket{D_y^{L_2,f}}\\
            &=\sum_{y\in [N]\setminus D_{all}}\frac{1}{\sqrt{|\Valid(D_y^{L_2,f})|\cdot (N-|D_{all}|)}} \ket{D_y^{L_2,f}} \sum_{S_2\in \Valid(D_y^{L_2,f})} \ket{S_2}\\
            &=\frac{1}{\sqrt{{N-|D_{all}|-1\choose M_2-|D_{1,all}|}(N-|D_{all}|)}}\sum_{y\in [N]\setminus D_{all}} \ket{D_y^{L_2,f}} \sum_{S_2\in \Valid(D_y^{L_2,f})}\ket{S_2}
        \end{split}
    \end{equation}
    On the other hand,
    \begin{equation}
    \begin{split}
        &V_3^{L_2}\frac{1}{\sqrt{|\Valid(D)|}}\ket{D}\sum_{S_2 \in \Valid(D)} \ket{S_2}\\
        &=\frac{1}{\sqrt{{N-|D_{all}|\choose M_2-|D_{1,all}|}}} \sum_{S_2\in \Valid(D)} \frac{1}{\sqrt{N-|S_2\cup D_{all}|}} \ket{D_y^{L_2,f}}\sum_{y\in N\setminus (S_2\cup D_{all})} \ket{S_2}\\
        &=\frac{1}{\sqrt{{N-|D_{all}|\choose M_2-|D_{1,all}|}(N-M_2-|D_{all}|+|D_{1,all}|)}}\sum_{y\in [N]\setminus D_{all}} \ket{D_y^{L_2,f}} \sum_{\substack{S_2 \in \Valid(D)\\y\notin S_2}}\ket{S_2}\\
        &=\frac{1}{\sqrt{{N-|D_{all}|\choose M_2-|D_{1,all}|}(N-M_2-|D_{all}|+|D_{1,all}|)}}\sum_{y\in [N]\setminus D_{all}} \ket{D_y^{L_2,f}} \sum_{S_2\in \Valid(D_y^{L_2,f})}\ket{S_2}\\
        &=\frac{1}{\sqrt{{N-|D_{all}|-1\choose M_2-|D_{1,all}|}(N-|D_{all}|)}}\sum_{y\in [N]\setminus D_{all}} \ket{D_y^{L_2,f}} \sum_{S_2\in \Valid(D_y^{L_2,f})}\ket{S_2}\\
        &=V_4^{L_2}\frac{1}{\sqrt{|\Valid(D)|}}\ket{D}\sum_{S_2 \in \Valid(D)} \ket{S_2}
    \end{split}
    \end{equation}
    where the second to last equality comes from~\Cref{claim:validswapout}.
\end{proof}

\begin{lemma}
    $\A^{\Ora{V_5}}=\A^{\Ora{V_4}}$
\end{lemma}

\begin{proof}
    This follows immediately from the fact that $\Ora{V_5}$ and $\Ora{V_4}$ are identical up to an isometry on register $St$.
    \todo{maybe cite a fact in prelims}
\end{proof}

\begin{lemma}
    $\A^{\Ora{V_6}}=\A^{\Ora{V_5}}$
\end{lemma}

\begin{proof}
    To show this, we will define a partial isometry $T$ acting on $\mathcal{H}(R_{[N]})^{\otimes 4}$ such that $(I\otimes T) W_6= W_4 (I\otimes T)$.

    In particular, let $L,R \in R_{[N]}$. We will define $(L_1',R_1',L_2',R_2')$ such that
    $$T\ket{L}\ket{R}\ket{\emptyset}\ket{\emptyset} = \ket{L_1'}\ket{R_1'}\ket{L_2'}\ket{R_2'}$$

    We will define $Reachable(L,R)$ to be the set of $x\in [N]$ such that there exists a "path" connecting $x$ to $S_1$ in $L$ and $R$. Formally, we build $Reachable(L,R)$ as follows
    \begin{enumerate}
        \item $S_1 \subseteq Reachable(L,R)$
        \item If there exists $y \in Reachable(L,R)$ such that $(x,y)\in L\cup R$ or $(y,x)\in L\cup R$, then $x\in Reachable(L,R)$.
    \end{enumerate}

    We then define 
    $$L_1' = \{(x,y)\in L:x\in Reachable(L,R)\}$$
    $$R_1' = \{(x,y)\in R:x\in Reachable(L,R)\}$$
    $$L_2' = \{(x,y)\in L:x\notin Reachable(L,R)\}$$
    $$R_2' = \{(x,y)\in R:x\notin Reachable(L,R)\}$$

    We will then show that $(I\otimes T)W_6 = W_4(I\otimes T)$. In particular, we will show $(I\otimes T)W_6^{L_1} = W_4^{L_1} (I\otimes T)$ and $(I\otimes T)W_6^{L_2} = W_4^{L_2} (I\otimes T)$, and the same will hold for other variants by an analogous argument.

    Let $L,R\in R_{[N]}$, $x\in [N]$. Define $DB=\Ima(L)\cup \Dom(L)\cup \Ima(R)\cup \Dom(R)$.

    We will show that $(I\otimes T)W_6^{L_1} = W_4^{L} (I\otimes T)$. The other cases will follow analogously, and so together we will get $(I\otimes T)W_6 = W_4 (I\otimes T)$
    
    We will consider the case where $x\in Reachable(L,R)$ and $x\notin Reachable(L,R)$ separately.

    If $x\in Reachable(L,R)$,
    \begin{equation}
        \begin{split}
            &(I\otimes T) W_6^{L_1} \ket{x}\ket{L}\ket{R}\\
            &=T \frac{1}{\sqrt{N-|DB|}}\sum_{y \in [N]\setminus DB} \ket{y}\ket{L\cup \{(x,y)\}}\ket{R}\\
            &=\frac{1}{\sqrt{N-|DB|}} \sum_{y\in [N]\setminus DB} \ket{y}\ket{L_1'\cup \{(x,y)\}}\ket{R_1'}\ket{L_2'}\ket{R_2'}\\
            &=W_4^{L_1} \ket{x}\ket{L_1'}\ket{R_1'}\ket{L_2'}\ket{R_2'}\\
            &=W_4^{L_1} W_4^{L_2} \ket{x}\ket{L_1'}\ket{R_1'}\ket{L_2'}\ket{R_2'}\\
            &=W_4^{L_1} W_4^{L_2} (I\otimes T)\ket{x}\ket{L}\ket{R}\\
            &=W_4^{L}(I\otimes T)\ket{x}\ket{L}\ket{R}\\
        \end{split}
    \end{equation}
    Here the key observation is that since $x \in Reachable(L,R)$, $(x,y)$ will appear in the $L_1'$ register after applying $T$. Furthermore, since $x\in Reachable(L,R)$, if we define $D=(x,L_1',R_1',L_2',R_2')$, $D_x \in D_{1,all}$ and so $W_4^{L_2}\ket{D}=\ket{D}$.

    If $x\notin Reachable(L,R)$,
    \begin{equation}
        \begin{split}
            &(I\otimes T) W_6^{L_1} \ket{x}\ket{L}\ket{R}\\
            &=T \frac{1}{\sqrt{N-|DB|}}\sum_{y \in [N]\setminus DB} \ket{y}\ket{L\cup \{(x,y)\}}\ket{R}\\
            &=\frac{1}{\sqrt{N-|DB|}} \sum_{y\in [N]\setminus DB} \ket{y}\ket{L_1'}\ket{R_1'}\ket{L_2'\cup \{(x,y)\}}\ket{R_2'}\\
            &=W_4^{L_2} \ket{x}\ket{L_1'}\ket{R_1'}\ket{L_2'}\ket{R_2'}\\
            &=W_4^{L_2}W_4^{L_1} \ket{x}\ket{L_1'}\ket{R_1'}\ket{L_2'}\ket{R_2'}\\
            &=W_4^{L} (I\otimes T)\ket{x}\ket{L}\ket{R}
        \end{split}
    \end{equation}
    Here the key observation is that since $x \notin Reachable(L,R)$, and since $y\notin DB$ is definitely not in $Reachable(L,R)$, $(x,y)$ will appear in the $L_2'$ register after applying $T$. Furthermore, if we define $D=(x,L_1',R_1',L_2',R_2')$, we have $x\notin D_{1,all}$, and so $W_4^{L_1}\ket{D}=\ket{D}$.
\end{proof}

\begin{lemma}
    $\norm{\A^{\Ora{V_7}}-\A^{\Ora{V_6}}}_2^2\leq \negl(n)$
\end{lemma}

\begin{proof}
    This follows from the exact same techniques as~\Cref{claim:hyb2}.
\end{proof}

\begin{lemma}
    $\norm{\A^{\Ora{V_f}}-\A^{\Ora{V_7}}}_2^2\leq \negl(n)$
\end{lemma}

\begin{proof}
    This follows from~\Cref{thm:pathrecordingreverse}.
\end{proof}
\section{Combinatorial Lemma}\label{sec:combinatorics}

We make use of the following combinatorial lemma in the proof of~\Cref{claim:painful}. We would like to acknowledge that the proof sketch for this lemma was derived by Chat-GPT.

\begin{lemma}\label{lem:comb}
    Let $P$ be any set of size $|P|=N$ and let $2 \leq M \leq N/2$. Define $\mathcal{S} = \{S \subseteq P : |S| = M\}$. Let $\vec{a}:\mathcal{S} \to \C$ be such that
    $$\sum_{S\in \mathcal{S}} \alpha_S = 0$$
    $$\norm{\vec{\alpha}}_2 = 1$$
    Then,
    $$\sum_{y \in P}\abs{\sum_{S \in \mathcal{S}:y\in S} a_S}^2 \leq {N-2 \choose M-1}$$
\end{lemma}

\begin{proof}
    We first define an $N\times {N\choose M}$ matrix $B$ with rows indexed by $P$ and columns indexed by $\mathcal{S}$ as follows:
    $$B_{y,S} = \begin{cases}
        1&y\in S\\
        0&y \notin S
    \end{cases}$$
    
    For ease of notation, let $\vec{1}_{P},\vec{1}_{\mathcal{S}}$ be the all-ones vectors indexed by $P$ and $\mathcal{S}$ respectively. Let $\mathbf{1}_{P\times P}$ be the all-ones matrix with rows and columns both indexed by $P$.

    Let $\vec{a}:\mathcal{S} \to \C$ be such that $\norm{\vec{a}}_2=1$ and $\norm{\vec{a}}_1=0$. Observe that 
    \begin{equation*}
        \begin{split}
            \norm{B\cdot \vec{a}}_2^2 &= \sum_{y \in P} \abs{\sum_{S\in \mathcal{S}}B_{y,S} \cdot a_S}^2\\
            &=\sum_{y \in P}\abs{\sum_{S \in \mathcal{S}:y\in S} a_S}^2
        \end{split}
    \end{equation*}
    and so it suffices to show $\norm{B\vec{a}}_2^2 \leq {N-2 \choose M-1}$. But note that $\norm{B\vec{a}}_2^2 = \vec{a}^\dagger B^\dagger B \vec{a}$. In order to bound this value, we will compute the eigenvalues of $B^\dagger B = B^TB$.

    Note that the non-zero eigenvalues of $B^T B$ are the squared singular values of $B$ or $B^T$, and so the non-zero eigenvalues of $B^T B$ are the same as the non-zero eigenvalues of $BB^T$. We can write $BB^T$ as follows
    \begin{equation}
        \begin{split}
            (BB^T)_{x,y} = \sum_{S \in \mathcal{S}} Ind[x\in S]\cdot Ind[y\in S]\\
            =|\{S:\{x,y\}\subseteq S\}| = \begin{cases}
        {N-1\choose M-1} & \text{if} \ x=y\\
        {N-2 \choose M-2} & \text{if} \ x\neq y
    \end{cases}
        \end{split}
    \end{equation}
    Then
    \begin{equation}
        BB^T = {N-2\choose M-2}\mathbf{1}_{P\times P} + \left({N-1\choose M-1} - {N-2 \choose M-2}\right)I\\
        ={N-2\choose M-2}\mathbf{1}_{P\times P} + {N-2 \choose M-1}I
    \end{equation}

    Now, we can explicitly compute the spectrum of $BB^T$. For each $x\in P$,
    $$(BB^T \vec{1})_x = {N-1\choose M-1} + (N-1){N-2\choose M-2}$$
    and so $\vec{1}$ is an eigenvector with eigenvalue 
    \begin{equation}
        \begin{split}
            \lambda_1 &= {N-1\choose M-1} + (N-1){N-2\choose M-2}\\
            &=\frac{(N-1)!}{(M-1)!(N-M)!} + \frac{(N-1)(N-2)!}{(M-2)!(N-M)!}\\
            &=\frac{(N-1)!}{(M-1)!(N-M)!}+\frac{(M-1)(N-1)!}{(M-1)!(N-M)!}\\
            &=M\frac{(N-1)!}{(M-1)!(N-M)!}=M{N-1\choose M-1}
        \end{split}
    \end{equation}. Furthermore, consider any vector $\vec{v}$ such that $\vec{v}\cdot \vec{1}=0$. Then $\mathbf{1}\vec{v} = \vec{0}$ and so 
    $$BB^T \vec{v} = \left({N-2\choose M-2}\mathbf{1} + {N-2\choose M-1}I\right)\vec{v} = {N-2 \choose M-1}\vec{v}$$

    Therefore, $BB^T$ is a full rank matrix with eigenvalues $M{N-1\choose M-1}$ with multiplicity $1$ and ${N-2\choose M-1}$ with multiplicity $N-1$. Similarly, $B^TB$ has eigenvalues $M{N-1\choose M-1}$ with multiplicity $1$, ${N-2\choose M-1}$ with multiplicity $N-1$, and $0$ with multiplicity ${N\choose M} - N$.

    To complete the argument, we will show that the all-ones vector $\vec{1}_{\mathcal{S}}$ is an eigenvector of $B^TB$ with eigenvalue $M{N-1\choose M-1}$.
    \begin{equation}
        \begin{split}
            (B\vec{1})_S = \sum_{y \in P} Ind[y\in S] = |\{S : y\in S\}| = {N-1\choose M-1}\\
            (B^TB\vec{1})_y = \sum_{S\in \mathcal{S}} Ind[y\in S](B\vec{1})_{S} = M{N-1 \choose M-1}\\
            B^TB\vec{1} = M\cdot {N-1\choose M-1}\vec{1}
        \end{split}
    \end{equation}

    Thus, since the largest eigenvalue of $B^TB$ has multiplicity $1$ and is spanned by $\vec{1}_{\mathcal{S}}$, and since we know $\vec{a}\cdot \vec{1}_{\mathcal{S}} = 0$, we must have that $\vec{a}^\dagger B^TB\vec{a}$ is bounded by the second largest eigenvalue of $B^TB$, which is ${N-2\choose M-1}$.

    In detail, let $\vec{v}_1,\dots,\vec{v}_{{N\choose M}}$ be any basis with $\vec{v}_1 = \vec{1}$, $(\vec{v}_2,\dots,\vec{v}_N)$ orthogonal eigenvectors for the eigenvalue ${N-2\choose M-1}$, and $(\vec{v}_{N+1},\dots,\vec{v}_{N\choose M})$ a basis for the kernel of $B^TB$. Then, $\vec{a} = \sum_i a_i\vec{v}_1$, with $a_i=0$ since $\vec{a}\cdot \vec{1}_{\mathcal{S}}=0$. And so
    \begin{equation}
        \begin{split}
            \vec{a}^\dagger B^TB\vec{a}= \sum_{i,j} a_i^* a_j \vec{v}_i B^TB \vec{v}_j\\
            =\sum_{i \neq 0} |a_i|^2 \lambda_i \leq {N-2 \choose M-1} \sum_{i} |a_i|^2\\
            \leq {N-2 \choose M-1}
        \end{split}
    \end{equation}
\end{proof}

\section{Proof of Lemma B.5}\label{sec:painful}

\begin{lemma}[\Cref{claim:painful} restated]
    $$\norm{\A^{\vec{V_3}}-\A^{\vec{V_2}}}_2 \leq \negl(n)$$
\end{lemma}

We will show that $\Pi_{good}$ almost commutes with all operations used in $\A^{\vec{V_2}}$. In particular, we will show that $$\norm{\Pi_{good} V_2^{L_1} - V_2^{L_1} \Pi_{good}}_{op} \leq \negl(n)$$ and similar statements hold for all components of $V_2,\ol{V_2}$. We will also show that for all unitaries $A$ acting only on an internal register $\Reg{A}$ and register $\Reg{In}$, $\norm{(A_{\Reg{A},\Reg{In}} \otimes I_{\Reg{St}} (I_{\Reg{A}}\otimes \Pi_{good}) - (I_{\Reg{A}}\otimes \Pi_{good}) (A_{\Reg{A},\Reg{In}\otimes I_{\Reg{St}})}}_{op} \leq \negl(n)$. Note that $\ket{\phi_0} \in \Span(\Pi_{good})$. By~\Cref{prop:diamop,lem:opprod,prop:opinv} we get that $\norm{\Pi_{good} V_2 \Pi_{good} - V_2\Pi_{good}}_{op} \leq \negl(n)$ and the same holds for $V_2^\dagger, \ol{V_2},\ol{V_2}^\dagger$ as well as $A$ as described prior. 

Thus, there exist maps $W_1,\dots,W_{2t}$ such that
$$\A^{\vec{V_2}} = W_{2t} \cdots W_3 \cdot W_2 \cdot W_1 \ket{0}\ket{\phi_0}$$
and for each $i$, $\norm{\Pi_{good} W_i \Pi_{good} - W_i \Pi_{good}}_{\diamond} \leq \negl(n)$. 

Thus, if $\approx$ denotes that the two states are negligibly close in trace distance, we have by induction
\begin{equation}
\begin{split}
    &\A^{\vec{V_2}} = W_{2t} \cdots W_3 \cdot W_2 \cdot W_1 \ket{0}\ket{\phi_0}\\
    &=W_{2t} \cdots W_3 \cdot W_2 \cdot W_1 \cdot (I \otimes \Pi_{good})\ket{0}\ket{\phi_0}\\
    &\approx W_{2t} \cdots W_3 \cdot W_2 \cdot (I \otimes \Pi_{good}) \cdot W_1 \cdot (I \otimes \Pi_{good})\ket{0}\ket{\phi_0}\\
    &\approx W_{2t} \cdots W_3 \cdot (I \otimes \Pi_{good}) \cdot W_2 \cdot (I \otimes \Pi_{good}) \cdot W_1 \cdot (I \otimes \Pi_{good})\ket{0}\ket{\phi_0}\\
    &\approx W_{2t} \cdot (I \otimes \Pi_{good}) \cdots W_3 \cdot (I \otimes \Pi_{good}) \cdot W_2 \cdot (I \otimes \Pi_{good}) \cdot W_1 \cdot (I \otimes \Pi_{good})\ket{0}\ket{\phi_0}\\
    &=\A^{\vec{V_3}}
\end{split}
\end{equation}

In order to show that $\Pi_{good}$ approximately commutes with $\vec{V_2}$, we will show that it approximately commutes with all components of $\vec{V_2}$. We will only show the cases of $V_2^{L_1}$ and $V_2^{L_2}$, and all other cases will follow by symmetric arguments.

\begin{lemma}
    $\norm{\Pi_{good} V_2^{L_1} - V_2^{L_1} \Pi_{good}}_{op} \leq \negl(n)$
\end{lemma}

\begin{proof}
    It turns out that $\Pi_{good} V_2^{L_1} \Pi_{good} = V_2^{L_1} \Pi_{good}$. To show this, let $\ket{\phi} = \frac{1}{\sqrt{|\Valid(D)|}} \ket{D}\sum_{S_2\in \Valid(D)} \ket{S_2}$ be a basis state of $\Pi_{good}$. We will show that $V_2^{L_1} \ket{\phi} \in \Span(\Pi_{good})$. 

    If $D_x \notin D_1$, then for all $S_2\in \Valid(D)$, $D_x\notin S_2$. And so,
    \begin{equation}
        \begin{split}
            V_2^{L_1} \ket{\phi}\\
            =\frac{1}{\sqrt{|\Valid(D)|}}\sum_{S_2\in \Valid(D)} V_2^{L_1}( \ket{D} \ket{S_2})\\
            =\frac{1}{\sqrt{|\Valid(D)|}} \ket{D}\sum_{S_2\in \Valid(D)} \ket{S_2}\\
            =\ket{\phi}
        \end{split}
    \end{equation}
    which is clearly in $\Span(\Pi_{good})$.

    If $D_x\in D_1$, then for all $S_2\in \Valid(D)$, $D_1\subseteq S_2$ and so $D_x \in S_2$. Observe that for all $S_2\in \Valid(D)$, for all $y\in S_2\setminus D_1$, $y\neq D_x$ (since $D_x\in D_1$) and $y\notin D_2$ (since $D_2\cap S_2=\emptyset$). Thus, for all $y\in S_2\setminus D_1$, $y\in [N]\setminus D_{all}$.
    \begin{equation}
        \begin{split}
            &V_2^{L_1} \ket{\phi}\\
            &=\frac{1}{\sqrt{|\Valid(D)|}}\sum_{S_2\in \Valid(D)} V_2^{L_1}( \ket{D} \ket{S_2})\\
            &=\frac{1}{\sqrt{|\Valid(D)|}} \sum_{S_2\in \Valid(D)} \frac{1}{\sqrt{|S_2\setminus D_1|}} \sum_{y \in S_2\setminus D_1}\ket{D_y^{L_1,f}}\ket{S_2}\\
            &=\frac{1}{\sqrt{|\Valid(D)|\cdot (M_2-|D_1|)}} \sum_{y \in [N]\setminus D_{all}} \ket{D_y^{L_1,f}} \sum_{\substack{S_2\in \Valid(D)\\y\in S_2}}\ket{S_2}\\
            &=\frac{1}{\sqrt{|\Valid(D)|\cdot (M_2-|D_1|)}} \sum_{y \in [N]\setminus D_{all}} \ket{D_y^{L_1,f}} \sum_{\substack{S_2\in \Valid(D_y^{L_1,f})}}\ket{S_2}\\
        \end{split}
    \end{equation}
    where the last line follows from~\Cref{claim:validswapin}. But note that this is a superposition of basis states for $\Span(\Pi_{good})$, and so is clearly also in $\Span(\Pi_{good})$.

    It thus follows that $\Pi_{good}V_2^{L_1}\Pi_{good} = V_2^{L_1} \Pi_{good}$.

    And so, it remains to show 
    $$\norm{\Pi_{good} V_2^{L_1} \Pi_{good} - \Pi_{good}V_2^{L_1}}_{op}\leq \negl(n)$$
    Observe
    \begin{equation}
        \begin{split}
            \norm{\Pi_{good}V_2^{L_1} (I-\Pi_{good})}_{op}\\
            =\norm{\Pi_{good}V_2^{L_1} - \Pi_{good}V_2^{L_1}\Pi_{good}}_{op}\\
            =\norm{\Pi_{good}V_2^{L_1}\Pi_{good}-\Pi_{good}V_2^{L_1} }_{op}\\
        \end{split}
    \end{equation}
    and so it is sufficient to bound 
    \begin{equation}
        \begin{split}
            &\sup_{\abs{\braket{\phi}}^2\leq1} \norm{\Pi_{good} V_2^{L_1} (I-\Pi_{good})\ket{\phi}}_2^2\\
            &= \sup_{\substack{\ket{\phi} \in \Span(I-\Pi_{good})\\\abs{\braket{\phi}}^2\leq 1}} \norm{\Pi_{good} V_2^{L_1}\ket{\phi}}_2^2
        \end{split}
    \end{equation}
    We can define a spanning set for $\Span(I-\Pi_{good})$ consisting of two types of states
    \begin{enumerate}
        \item States of type $1$ will be states with an "invalid" $S_2$. These will be states of the form
        $$\ket{\phi_1} = \ket{D}\ket{S_2}$$
        for some $D\in \mathcal{DB}$, $S_2\notin \Valid(D)$.
        
        \item States of type $2$ will be states with a valid $S_2$ register, but which are orthogonal from the uniform superposition over valid registers. These will be states of the form
        $$\ket{\phi_2} = \ket{D} \sum_{S_2 \in \Valid(D)} \alpha_{S_2} \ket{S_2}$$
        for some $\vec{\alpha}:\mathcal{P}([N])\to \C$ satisfying
        $$\sum_{S_2\in \Valid(D)} \alpha_{S_2}=\norm{\vec{\alpha}}_1=0$$
        $$\sum_{S_2\in \Valid(D)} \alpha_{S_2}^2=\norm{\vec{\alpha}}_2^2=1$$
    \end{enumerate}
    We then make the following observations
    \begin{enumerate}
        \item States of type $1$ are mapped by $V_2^{L_1}$ into states of type $1$. In particular, let $D\in \mathcal{DB},S_2\notin\Valid(D)$, and $\ket{\phi_1} = \ket{D}\ket{S_2}$. Then if $D_x\in S_2$
        \begin{equation}
            \begin{split}
                &V_2^{L_1} \ket{\phi}\\
                &=V_2^{L_1}\ket{D}\ket{S_2}\\
                &=\ket{D_y^{L_1,f}}\ket{S_2}
            \end{split}
        \end{equation}
        which is also a state of type $1$, since $\Valid(D_y^{L_1,f})\subseteq \Valid(D)$ by~\Cref{claim:validswapin}. If $D_x\notin S_2$, then $V_2^{L_1} \ket{\phi}=\ket{\phi}$ which is also a state of type $1$.
        
        This implies that for all states $\ket{\phi_1}$ of type $1$, $\norm{\Pi_{good}V_2^{L_1}\ket{\phi_1}}_2^2=0$. 

        \item Let $D\in \mathcal{DB}$, $S_2\in \Valid(D)$. Define $$\Pi_{good}^D = \sum_{y \in [N]\setminus D_{all}} \ketbra{ValSt_{D_y^{L_1,f}}}$$
        Then 
        $$\Pi_{good} \ket{D}\ket{S_2} = \Pi_{good}^D \ket{D}\ket{S_2}$$
        and for any $D\neq D'$, by~\Cref{lem:outputortho},
        $$\Pi_{good}^{D'} \Pi_{good}^D = 0$$

        \item Let $A\neq B\in \mathcal{DB}$ and $\vec{\alpha},\vec{\beta}$ be such that 
        $$\ket{\phi_A} = \ket{A}\sum_{S_2\in \Valid(A)} \alpha_{S_2}\ket{S_2}$$
        $$\ket{\phi_B} = \ket{B}\sum_{S_2\in \Valid(B)} \beta_{S_2}\ket{S_2}$$
        are states of type $2$. Then
        \begin{equation}
        \begin{split}
            &\norm{\Pi_{good}V_2^{L_1}(\ket{\phi_A}+\ket{\phi_B})}_2^2 \\
            =& \norm{\Pi_{good}^A V_2^{L_1}\ket{\phi_A} + \Pi_{good}^B V_2^{L_1}\ket{\phi_B}}_2^2 \\
            =&|\bra{\phi_A} V_2^{L_1,\dagger} \Pi_{good}^A V_2^{L_1}\ket{\phi_A}+\bra{\phi_A} V_2^{L_1,\dagger} \Pi_{good}^A \Pi_{good}^B V_2^{L_1}\ket{\phi_B}\\
            &+\bra{\phi_B} V_2^{L_1,\dagger} \Pi_{good}^B\Pi_{good}^A V_2^{L_1}\ket{\phi_A}+\bra{\phi_B} V_2^{L_1,\dagger} \Pi_{good}^B V_2^{L_1}\ket{\phi_B}|\\
            =&|\bra{\phi_A} V_2^{L_1,\dagger} \Pi_{good}^A V_2^{L_1}\ket{\phi_A}+\bra{\phi_B} V_2^{L_1,\dagger} \Pi_{good}^B V_2^{L_1}\ket{\phi_B}|\\
            =&\norm{\Pi_{good}^A V_2^{L_1}\ket{\phi_A}}_2^2+\norm{\Pi_{good}^B V_2^{L_1}\ket{\phi_B}}_2^2\\
            =&\norm{\Pi_{good}V_2^{L_1}\ket{\phi_A}}_2^2+\norm{\Pi_{good} V_2^{L_1}\ket{\phi_B}}_2^2
        \end{split}
        \end{equation}
        
        \item For any state of type $2$, $\norm{V_2^{L_1} \ket{\phi_2}}^2 \leq \negl(n)$. In particular, let $D\in \mathcal{DB}$, and $\vec{\alpha}$ satisfying the requirements for states of type $2$. Let $$\ket{\phi_2} = \ket{D}\sum_{S_2\in \Valid(D)}\alpha_{S_2}\ket{S_2}$$. We will split this up into two cases.
        \begin{enumerate}
            \item  Consider the case when $D_x\in D_1$. Then,
            \begin{equation}
            \begin{split}
                &V_2^{L_1} \ket{\phi_2}\\
                =& \sum_{S_2\in \Valid(D)} \alpha_{S_2} V_2^{L_1} (\ket{D}\ket{S_2})\\
                =& \sum_{S_2\in \Valid(D)} \frac{1}{\sqrt{|S_2\setminus D_1|}} \alpha_{S_2} \sum_{y\in S_2\setminus D_1} \ket{D_y^{L_1,f}}\ket{S_2}\\
                =&\frac{1}{\sqrt{M_2-|D_1|}}\sum_{y\in [N]\setminus D_{all}} \ket{D_y^{L_1,f}}\sum_{\substack{S_2\in \Valid(D)\\y\in S_2}} \alpha_{S_2}\ket{S_2}\\
                =&\frac{1}{\sqrt{M_2-|D_1|}} \sum_{y\in [N]\setminus D_{all}} \ket{D_y^{L_1,f}} \sum_{S_2\in \Valid(D_y^{L_1,f})} \alpha_{S_2} \ket{S_2}
            \end{split}
            \end{equation}
            where the last line follows from~\Cref{claim:validswapin}.
    
            Observe that for all $y\in [N]\setminus D_{all}$, since $D_x\in D_1$, $|\Valid(D_y^{L_1,f})|={N-|D_{all}|-1\choose M_2-|D_1|-1}$.
            
            We can then compute
            \begin{equation}
                \begin{split}
                    &\norm{\Pi_{good} V_2^{L_1} \ket{\phi_2}}_2^2\\
                    =&\norm{\Pi_{good}^D V_2^{L_1} \ket{\phi_2}}_2^2\\
                    =&\sum_{y\in [N]\setminus D_{all}} \abs{\bra{ValSt_{D_y^{L_1,f}}} V_2^{L_1}\ket{\phi_2}}^2\\
                    =&\sum_{y\in [N]\setminus D_{all}} \abs{\frac{1}{\sqrt{|\Valid(D_y^{L_1,f})|\cdot (M_2-|D_1|)}} \sum_{S_2\in \Valid(D_y^{L_1,f}) }\ket{S_2}}^2\\
                    =&\frac{1}{(M_2-|D_1|)\cdot {N-|D_{all}|-1\choose M_2-|D_1|-1}}\sum_{y\in [N]\setminus D_{all}}\abs{\sum_{\substack{S_2\in \Valid(D)\\y\in S_2}}\alpha_{S_2}}^2\\
                    =&\frac{1}{(M_2-|D_1|)\cdot {N-|D_{all}|-1\choose M_2-|D_1|-1}}\sum_{y\in [N]\setminus D_{all}}\abs{\sum_{\substack{S_2\subseteq [N]\setminus D_{all}\\|S_2|=M_2-|D_1|\\y\in S_2}}\alpha_{S_2 \cup D_1}}^2
                \end{split}
            \end{equation}
            So setting $P=[N]\setminus D_{all}$ and $M=M_2-|D_1|$ in~\Cref{lem:comb}, we get
            \begin{equation}
                \begin{split}
                    &\norm{\Pi_{good} V_2^{L_1} \ket{\phi_2}}_2^2\\
                    \leq& \frac{{N-|D_{all}|-2 \choose M_2 - |D_1| - 1}}{(M_2-|D_1|)\cdot {N-|D_{all}|-1\choose M_2-|D_1|-1}}\\
                    \leq& \negl(n)
                \end{split}
            \end{equation}
            \item Now, consider the case where $D_x\notin D_1$. Then for all $S_2\in \Valid(D)$, we know $D_x\notin S_2$. Thus, we get
            \begin{equation}
                \begin{split}
                    &V_2^{L_1} \ket{\phi_2}\\
                    =&\sum_{S_2\in \Valid(D)} \alpha_{S_2}V_2^{L_1}(\ket{D}\ket{S_2})\\
                    =&\ket{D}\sum_{S_2\in \Valid(D)}\alpha_{S_2}\ket{S_2}\\
                    =&\ket{\phi_2}
                \end{split}
            \end{equation}
            which is in $\Span(I-\Pi_{good})$.
        \end{enumerate}
    \end{enumerate}

    Finally, we can write $\ket{\phi} \in \Span(I-\Pi_{good})$ as a superposition of states of types $1$ and $2$ where, for each database $D$, $D$ appears at most once as a state of type $1$ or type $2$. In particular, for all databases $D\in \mathcal{DB}$ and subsets $S_2 \in \mathcal{P}([N])$, there exist constants $b_{D,S_2}, a_{D}, \alpha_{D,S_2}$ such that
    $$\ket{\phi} = \sum_{D\in \mathcal{DB}}\left(\ket{D}\left(\sum_{S_2\notin \Valid(D)} b_{D,S_2}\ket{S_2} + a_D\sum_{S_2\in \Valid(D)} \alpha_{D,S_2} \ket{S_2}\right)\right)$$
    $$\sum_{D} |a_D|^2 + \sum_{D,S_2} |b_{D,S_2}|^2=1$$
    and for all $D$,
    $$\sum_{S_2\in \Valid(D)}|\alpha_{D,S_2}|^2=1$$
    $$\sum_{S_2\in \Valid(D)}\alpha_{D,S_2}=0$$
    Then by our first observation,
    \begin{equation}
    \begin{split}
        &\norm{\Pi_{good}V_2^{L_1}\ket{\phi}}_2^2\\
        =&\norm{\Pi_{good}V_2^{L_1}\left(\sum_{D\in \mathcal{DB}}\ket{D}\left(\sum_{S_2\notin \Valid(D)} b_{D,S_2}\ket{S_2} + a_D\sum_{S_2\in \Valid(D)} \alpha_{D,S_2} \ket{S_2}\right)\right)}_2^2\\
        =&\norm{\Pi_{good}V_2^{L_1}\left(\sum_{D\in \mathcal{DB}}a_D\ket{D} \sum_{S_2\in \Valid(D)} \alpha_{D,S_2} \ket{S_2}\right)}_2^2\\
        \end{split}
    \end{equation}
    And by our third and fourth observations,
    \begin{equation}
    \begin{split}
        &\norm{\Pi_{good}V_2^{L_1}\ket{\phi}}_2^2\\
        =&\sum_{D\in \mathcal{DB}} \norm{a_D \Pi_{good} V_2^{L_1} \ket{D}\sum_{S_2\in \Valid(D) \alpha_{D,S_2}\ket{S_2}}}_2^2\\
        \leq& \left(\sum_{D\in \mathcal{DB}} |a_D|^2 \right)\cdot \negl(n)\\
        =&\negl(n)
        \end{split}
    \end{equation}
\end{proof}

\begin{lemma}
    $\norm{\Pi_{good} V_2^{L_2} - V_2^{L_2} \Pi_{good}}_{op} \leq \negl(n)$
\end{lemma}

\begin{proof}
    The proof of this will mostly follow the same argument as the previous lemma. There are only two portions of the proof which differ. In particular, we need to separately argue that $$\Pi_{good} V_2^{L_2} \Pi_{good} = V_2^{L_2} \Pi_{good}$$
    and that for $\ket{\phi_2}$ a state of type 2,
    $$\norm{V_2^{L_2}\ket{\phi_2}}_2^2 \leq \negl(n)$$

    Let us first show $\Pi_{good} V_2^{L_2} \Pi_{good} = V_2^{L_2} \Pi_{good}$. Let $\ket{\phi} = \frac{1}{\sqrt{M_2}} \ket{D}\sum_{S_2\in \Valid(D)} \ket{S_2}$ be a basis state of $\Pi_{good}$. We will show that $V_2^{L_2} \ket{\phi} \in \Span(\Pi_{good})$. 

    If $D_x \in D_1$, then for all $S_2\in \Valid(D)$, $D_x\in S_2$. And so,
    \begin{equation}
        \begin{split}
            &V_2^{L_2} \ket{\phi}\\
            =&\frac{1}{\sqrt{M_2-|D_{1}|}}\sum_{S_2\in \Valid(D)} V_2^{L_2}( \ket{D} \ket{S_2})\\
            =&\frac{1}{\sqrt{M_2-|D_1|}} \ket{D}\sum_{S_2\in \Valid(D)} \ket{S_2}\\
            =&\ket{\phi}
        \end{split}
    \end{equation}
    which is clearly in $\Span(\Pi_{good})$.

    If $D_x\notin D_1$, then for all $S_2\in \Valid(D)$, $D_x \in S_2$. Observe that for all $S_2\in \Valid(D)$, for all $y\in [N]\setminus (S_2\cup D_{all})$, $y\neq D_x$ and $y\notin D_2$ (since $D_2\cup \{D_x\}\subseteq D_{all}$). Thus, for all $y\in [N]\setminus D_{all}$, $y\in [N]\setminus (S_2\cup D_{all})$.
    \begin{equation}
        \begin{split}
            &V_2^{L_2} \ket{\phi}\\
            =&\frac{1}{\sqrt{M_2}}\sum_{S_2\in \Valid(D)} V_2^{L_2}( \ket{D} \ket{S_2})\\
            =&\frac{1}{\sqrt{M_2}} \sum_{S_2\in \Valid(D)}\frac{1}{\sqrt{|[N]\setminus (S_2\cup D_{all})|}} \sum_{y \in [N]\setminus (S_2\cup D_{all})}\ket{D_y^{L_2,f}}\ket{S_2}\\
            =&\frac{1}{\sqrt{M_2\cdot (N-M_2-|D_2\cup \{D_x\}|)}} \sum_{y \in [N]\setminus D_{all}} \ket{D_y^{L_2,f}} \sum_{\substack{S_2\in \Valid(D)\\y\notin S_2}}\ket{S_2}\\
            =&\frac{1}{\sqrt{M_2\cdot (N-M_2-|D_2\cup \{D_x\}|)}} \sum_{y \in [N]\setminus D_{all}} \ket{D_y^{L_2}} \sum_{\substack{S_2\in \Valid(D_y^{L_2,f})}}\ket{S_2}\\
        \end{split}
    \end{equation}
    where the last line follows from~\Cref{claim:validswapout}. But note that this is a superposition of basis states for $\Span(\Pi_{good})$, and so is clearly also in $\Span(\Pi_{good})$.

    It thus follows that $\Pi_{good}V_2^{L_2}\Pi_{good} = V_2^{L_2} \Pi_{good}$.

    And so it remains to show that for $\ket{\phi_2}$ a state of type 2,
    $$\norm{V_2^{L_2}\ket{\phi_2}}_2^2 \leq \negl(n)$$

    We will write
    $$\ket{\phi_2} = \ket{D} \sum_{S_2\in \Valid(D)} \alpha_{S_2}\ket{S_2}$$

    Just as in the previous lemma, we split into two cases
    \begin{enumerate}
        \item If $D_x \in D_1$, then for all $S_2\in \Valid(D)$, we know $D_x \in S_2$. And so we get
        \begin{equation}
            V_2^{L_2} \ket{\phi_2} = \ket{\phi_2}
        \end{equation}
        \item If $D_x \notin D_1$, then 
        \begin{equation}
        \begin{split}
            &V_2^{L_2} \ket{\phi_2}\\
            =&\sum_{S_2\in \Valid(D)} \alpha_{S_2} V_2^{L_2}(\ket{D}\ket{S_2})\\
            =&\sum_{S_2\in \Valid(D)}\frac{1}{\sqrt{|[N]\setminus D_{all}|}}\alpha_{S_2} \sum_{y\in [N]\setminus D_{all}}\ket{D_y^{L_2,f}}\ket{S_2}\\
            =&\frac{1}{\sqrt{N-M_2-|D_2\cup \{D_x\}|}}\sum_{y\in [N]\setminus D_{all}}\ket{D_y^{L_2,f}} \sum_{\substack{S_2\in \Valid(D)\\y\notin S_2}}\alpha_{S_2}\ket{S_2}\\
            =&\frac{1}{\sqrt{N-M_2-|D_2\cup \{D_x\}|}}\sum_{y\in [N]\setminus D_{all}}\ket{D_y^{L_2,f}} \sum_{\substack{S_2\in \Valid(D_y^{L_2,f})}}\alpha_{S_2}\ket{S_2}
        \end{split}
        \end{equation}
        where the last line follows from~\Cref{claim:validswapout}. Observe that for all $y\in [N]\setminus D_{all}$,
        $$|\Valid(D_y^{L_2,f})|= {N-|D_{all}| - 1\choose M_2 - |D_1|}$$
        Analogously to the $V_2^{L_1}$ case, we have
        \begin{equation}
            \begin{split}
                &\norm{\Pi_{good}V_2^{L_2}\ket{\phi_2}}_2^2\\
                =&\frac{1}{(N-M_2-|D_2\cup \{D_x\}|)\cdot {N-|D_{all}|-1\choose M_2-|D_1|}} \sum_{y\in [N]\setminus D_{all}}\abs{\sum_{\substack{S_2\in \Valid(D)\\y\notin S_2}} \alpha_{S_2}}^2\\
                =&\frac{1}{(N-M_2-|D_2\cup \{D_x\}|)\cdot {N-|D_{all}|-1\choose M_2-|D_1|}} \sum_{y\in [N]\setminus D_{all}}\abs{\sum_{\substack{S_2\in \Valid(D)\\y\in S_2}} \alpha_{S_2}}^2\\
                =&\frac{1}{(N-M_2-|D_2\cup \{D_x\}|)\cdot {N-|D_{all}|-1\choose M_2-|D_1|}} \sum_{y\in [N]\setminus D_{all}}\abs{\sum_{\substack{S_2\subseteq [N]\setminus D_{all}\\|S_2|=M_2-|D_1|\\y\in S_2}} \alpha_{S_2\cup D_1}}^2
            \end{split}
        \end{equation}
        where the second equality comes from the fact that $\sum_{S_2\in \Valid(D)} \alpha_{S_2}=0$. So setting $P=[N]\setminus D_{all}$ and $M=M_2-|D_1|$ in~\Cref{lem:comb}, we get
        \begin{equation}
        \begin{split}
            &\norm{\Pi_{good}V_2^{L_1}\ket{\phi_2}}_2^2\\
            \leq& \frac{{N-|D_{all}|-2\choose M_2-|D_1|-1}}{(N-M_2-|D_2\cup \{D_x\}|)\cdot {N-|D_{all}|-1\choose M_2-|D_1|}}\\
            \leq& \negl(n)
            \end{split}
        \end{equation}
    \end{enumerate}
\end{proof}

We will conclude the argument by showing that unitaries $A$ which do not touch register $\Reg{St}$ approximately commute with $\Pi_{good}$.

\begin{lemma}
    Let $A$ be a unitary acting on $\Reg{A},\Reg{In}$. Then
    $$\norm{(A_{\Reg{A},\Reg{In}} \otimes I_{\Reg{St}} (I_{\Reg{A}}\otimes \Pi_{good}) - (I_{\Reg{A}}\otimes \Pi_{good}) (A_{\Reg{A},\Reg{In}\otimes I_{\Reg{St}})}}_{op} \leq \negl(n)$$
\end{lemma}

\begin{proof}
    We will define a different projector $\wt{\Pi}_{good}$ which is close to $\Pi_{good}$ in operator norm and exactly commutes with all such unitaries $A$. In particular, define $\wt{\Valid}(D)$ to be the same as $\Valid(D)$ except for the condition on $D_x$. That is, $\wt{\Valid}(D) = \{S_2:D_1\subseteq S_2,D_2\cap S_2=\emptyset\}$. Similarly, we define
    $$\ket{\wt{ValSt}_D} = \frac{1}{\sqrt{|\wt{\Valid}(D)|}}\ket{D}\sum_{S_2\in \wt{\Valid}(D)} \ket{S_2}$$
    Define $\wt{\Pi}_{good}= \sum_D \ketbra{\wt{ValSt}_D}$.

    We will separately show 
    $$\norm{\Pi_{good}\wt{\Pi}_{good}-\Pi_{good}}_{op}\leq \negl(n)$$
    $$\norm{\wt{\Pi}_{good}\Pi_{good}-\Pi_{good}}_{op}\leq \negl(n)$$

    By~\Cref{prop:opbound}, it is sufficient to show that for all unit vectors $\ket{\phi}$,
    $$|\bra{\phi} \wt{\Pi}_{good}\Pi_{good}\wt{\Pi}_{good}\ket{\phi}|^2 \geq 1- \negl(n)$$
    $$|\bra{\phi} \Pi_{good}\wt{\Pi}_{good}\Pi_{good}\ket{\phi}|^2 \geq 1 - \negl(n)$$
    That is, for all unit vectors $\ket{\phi}\in \Span(\Pi_{good})$,
    $$\norm{\wt{\Pi}_{good}\ket{\phi}}_2^2\geq 1- \negl(n)$$
    and for all unit vectors $\ket{\psi}\in \Span(\wt{\Pi}_{good})$,
    $$\norm{{\Pi_{good}}\ket{\phi}}_2^2\geq 1- \negl(n)$$

    Note that $|\wt{\Valid}(D)|={N-|D_1|-|D_2|\choose M_2-|D_1|}$. Meanwhile, $|\Valid(D)|\geq {N-|D_1|-|D_2|-1\choose M_2-|D_1|}$

    And so we observe
    \begin{equation}
    \begin{split}
        &\braket{\wt{ValSt}_D}{ValSt_D}\\
        =&\frac{|\Valid(D)\cap \wt{\Valid}(D)|}{\sqrt{|\Valid(D)|\cdot |\wt{\Valid}(D)|}}\\
        \geq& \frac{|\Valid(D)|}{|\wt{\Valid}(D)|}\\
        \geq& \frac{{N-|D_1| - |D_2|-1\choose M_2-|D_1|}}{{N-|D_1|-|D_2|\choose M_2-|D_1|}}\\
        =&\frac{(N-|D_1|-|D_2|-1)!(M_2-|D_1|)!(N-M_2-|D_2|)!}{(N-|D_1|-|D_2|)!(M_2-|D_1|)!(N-M_2-|D_2|-1)!}\\
        =&\frac{N-M_2-|D_2|}{N-|D_1|-|D_2|}\\
        \geq& \frac{N-M_2-2t}{N}\\
        \geq& 1-\negl(n)
        \end{split}
    \end{equation}
    and for $D\neq D'$,
    $$\braket{\wt{ValSt}_D}{ValSt_D}=0$$

    \begin{enumerate}
        \item Let $\ket{\phi} = \sum_{D} \alpha_D \ket{ValSt_D}$ be a unit vector in $\Span(\Pi_{good})$. Define 
        $$\ket{\wt{\phi}} = \sum_D \alpha_D \ket{\wt{ValSt}_D}$$
        Then 
        \begin{equation}
            \begin{split}
                &\braket{\wt{\phi}}{\phi}\\
                =&\sum_{D}|\alpha_D|^2 (1-\negl(n))\\
                \geq& 1-\negl(n)
            \end{split}
        \end{equation}
        Thus, since $\ket{\wt{\phi}}\in \Span(\wt{\Pi}_{good})$, $$\norm{\wt{\Pi_{good}} \ket{\phi}}_2^2 \geq 1-\negl(n)$$
        \item A symmetric argument gives that for all $\ket{\psi} \in \Span(\wt{\Pi}_{good})$,
        $$\norm{\Pi_{good}\ket{\psi}}_2^2 \geq 1-\negl(n)$$
    \end{enumerate}
    Thus, $\norm{\Pi_{good}-\wt{\Pi}_{good}}_{op}\leq \negl(n)$.

    But it is clear that for any unitary $A$, $$(A_{\Reg{A},\Reg{In}}\otimes I_{\Reg{St}})\cdot (I_{\Reg{A}}\otimes \wt{\Pi}_{good}) = (I_{\Reg{A}}\otimes \wt{\Pi}_{good})\cdot (A_{\Reg{A},\Reg{In}}\otimes I_{St}).$$ And so by~\Cref{lem:opprod}, where $\approx$ means negligibly close in operator distance, we have
    \begin{equation}
    \begin{split}
        A\Pi_{good}\approx A\wt{\Pi}_{good} = \wt{\Pi}_{good}A \approx \Pi_{good}A
    \end{split}
    \end{equation}
\end{proof}
\fi

\end{document}